\documentclass{article}

\usepackage{amssymb,amsmath,graphicx}
\usepackage{color,enumerate,comment,boxedminipage,tkz-graph}
\usepackage{pgf,float,hyperref}
\usepackage{subcaption}
\usepackage[T1]{fontenc}
\usepackage{xspace}
\usepackage{paralist}
\usepackage{lineno}
\usepackage{microtype}
\usepackage{fullpage}
\usepackage{tikz}
\usetikzlibrary{math,calc}
\usepackage{lmodern}
\usepackage[normalem]{ulem}

\usepackage{amsthm}

\newtheorem{theorem}{Theorem}
\newtheorem{lemma}[theorem]{Lemma}

\newcounter{claim}[theorem]

\newtheorem{claimm}[claim]{Claim} %[theorem]

\DeclareMathOperator{\cw}{cw}
\DeclareMathOperator{\tw}{tw}

\newtheorem{oproblem}{Open Problem}

\DeclareMathOperator{\diamondgraph}{diamond}
\DeclareMathOperator{\gem}{gem}

\DeclareMathOperator{\paw}{paw}

\newcommand{\ssi}{\subseteq_i}
\newcommand{\si}{\supseteq_i}
\newcommand{\NP}{{\sf NP}}

\usepackage{authblk}

\usepackage{todonotes}

\title{Clique-Width: Harnessing the Power of Atoms\thanks{The research in this paper received support from the Leverhulme Trust (RPG-2016-258).
Masa\v r\' ik and Novotn\' a were supported by Charles University student grants (SVV-2017-260452 and GAUK 1277018) and GA\v CR project (17-09142S).
The last author was supported by Polish National Science Centre grant no. 2018/31/D/ST6/00062.
An extended abstract of this paper appeared in the proceedings of WG 2020~\cite{DMNPR20}.}
}

\author[1]{Konrad K. Dabrowski}%\inst{1}
\author[2,3,4]{Tom\'{a}\v{s} Masa\v{r}\'ik}%\inst{2}%,3}
\author[2,3]{Jana Novotn\'a}%\inst{2}%,3}
\author[1]{Dani\"el Paulusma}%\inst{1}
\author[3,5]{Pawe{\l}~Rz{\k{a}}\.{z}ewski}%\inst{3}}%,4}}%\thanks{ Supported by Polish National Science Centre grant no. 2018/31/D/ST6/00062.}}
\affil[1]{Department of Computer Science, Durham University, UK}
\affil[ ]{\texttt{\{konrad.dabrowski,daniel.paulusma\}@durham.ac.uk}}
\affil[2]{Faculty of Mathematics and Physics, Charles University, Prague, Czech Republic}
\affil[ ]{\texttt{\{masarik,janca\}@kam.mff.cuni.cz}}
\affil[3]{Institute of Informatics, University of Warsaw, Poland}
\affil[4]{Department of Mathematics, Simon Fraser University, Canada}
\affil[5]{Faculty of Mathematics and Information Science, Warsaw University of Technology, Warsaw, Poland}
\affil[ ]{\texttt{p.rzazewski@mini.pw.edu.pl}}

\date{}

\newenvironment{inproof}{\noindent {\emph{Proof of Claim.}}}{\hfill$\diamond$\smallskip}

\begin{document}

\maketitle

\begin{abstract}
\noindent
Many \NP-complete graph problems are polynomial-time solvable on graph classes of bounded clique-width.  Several of these problems are polynomial-time solvable on a hereditary graph class~${\cal G}$ if they are so on the atoms (graphs with no clique cut-set) of~${\cal G}$. Hence, we initiate a systematic study into boundedness of clique-width of atoms of hereditary graph classes.  A graph $G$ is {\it $H$-free} if $H$ is not an induced subgraph of $G$, and it is {\it $(H_1,H_2)$-free} if it is both $H_1$-free and $H_2$-free. A class of $H$-free graphs has bounded clique-width if and only if its atoms have this property. This is no longer true for $(H_1,H_2)$-free graphs, as evidenced by one known example. We prove  the existence of another such pair $(H_1,H_2)$ and classify the boundedness of clique-width on $(H_1,H_2)$-free atoms for all but 18 cases.
\end{abstract}

\section{Introduction}\label{s-intro}
Many hard graph problems become tractable when restricting the input to some graph class.
The two central questions are ``for which graph classes does a graph problem become tractable'' and ``for which graph classes does it stay computationally hard?''
Ideally, we wish to answer these questions for a large set of problems simultaneously instead of considering individual problems one by one.

Graph width parameters~\cite{DJP19,Gu17,HOSG08,KLM09,Va12} help to make such results possible.
A graph class has {\em bounded} width if there is a constant~$c$ such that the width of all its members is at most~$c$.
As we discuss below, there are several meta-theorems that provide sufficient conditions for a problem to be tractable on a graph class of bounded width.

Two popular width parameters are treewidth ($\tw$) and clique-width ($\cw$).
For every graph~$G$ the inequality $\cw(G)\leq 3\cdot 2^{\tw(G)-1}$ holds~\cite{CR05}.
Hence, every problem that is polynomial-time solvable on graphs of bounded clique-width is also polynomial-time solvable on graphs of bounded treewidth.
However, the converse statement does not hold: there exist graph problems, such as {\sc List Colouring}, which are polynomial-time solvable on graphs of bounded treewidth~\cite{JS97}, but \NP-complete on graphs of bounded clique-width~\cite{CO00}.
Thus, the trade-off between treewidth and clique-width is that the former can be used to solve more problems, but the latter is {\em more powerful} in the sense that it can be used to solve problems for larger graph classes.

Courcelle~\cite{Co90} proved that every graph problem definable in MSO$_2$ is linear-time solvable on graphs of bounded treewidth.
Courcelle, Makowsky, and Rotics~\cite{CMR00} showed that every graph problem definable in the more restricted logic MSO$_1$ is polynomial-time solvable even on graphs of bounded clique-width (see~\cite{CE12} for details on MSO$_1$ and MSO$_2$).
Since then, several clique-width meta-theorems for graph problems not definable in MSO$_1$ have been developed~\cite{EGW01,GK03,KR03,Ra07}.

All of the above meta-theorems require a constant-width decomposition of the graph.
We can compute such a decomposition in polynomial time for treewidth~\cite{Bo96} and clique-width~\cite{OS06}, but for other width parameters, such as mim-width, which is even more powerful than clique-width~\cite{Va12},
it is not known whether this is possible and this problem may turn out to be harder.
For instance, unless $\mathsf{NP} = \mathsf{ZPP}$, there is no constant-factor approximation algorithm for mim-width that runs in polynomial time~\cite{SV16}.
Meta-theorems for mim-width~\cite{BV13,BTV13} currently require an appropriate constant-width decomposition as part of the input (which may still be found in polynomial time for some graph classes).
 
\medskip
\noindent
{\bf Our Focus.} In our paper we concentrate on {\em clique-width}\footnote{See Section~\ref{s-pre} for a definition of clique-width and other terminology used in Section~\ref{s-intro}.} in an attempt to find {\em larger} graph classes for which certain \NP-complete graph problems become tractable without the requirement of an appropriate decomposition as part of the input.
The type of graph classes we consider all have the natural property that they are closed under vertex deletion.
Such graph classes are said to be {\em hereditary} and there is a long-standing study on boundedness of clique-width for hereditary graph classes (see, for example, \cite{BDJLPZ17,BL02,BDHP16,BDHP17,BKM06,BLM04b,BLM04,BM02,DDP17,DHP0,DLP17,DLRR12,DP14,DP16,Gu17,KLM09,MR99}).

Besides capturing many well-known classes, the framework of hereditary graph classes also enables us to perform a {\em systematic} study of a width parameter or graph problem.
This is because every hereditary graph class~${\cal G}$ is readily seen to be uniquely characterized by a minimal (but not necessarily finite) set~${\cal F}_{\cal G}$ of forbidden induced subgraphs.
If $|{\cal F}_{\cal G}|=1$ or $|{\cal F}_{\cal G}|=2$, then~${\cal G}$ is said to be {\em monogenic} or {\em bigenic}, respectively.
Monogenic and bigenic graph classes already have a rich structure, and studying their properties has led to deep insights into the complexity of bounding graph parameters and solving graph problems.
This is evidenced, for example, by extensive studies on the classes of bull-free graphs~\cite{Ch12} or claw-free graphs~\cite{CS05,HMLW19}, and surveys for graph problems or parameters specifically restricted to bigenic graph classes~\cite{DJP19,GJPS17}.

It is well known (see e.g.~\cite{DP16}) that a monogenic class of graphs has bounded clique-width if and only if it is a subclass of the class~${\cal G}$ with ${\cal F}_{\cal G}=\{P_4\}$.
The survey~\cite{DJP19} gives a state-of-the-art theorem on the boundedness and unboundedness of clique-width of bigenic graph classes.
Unlike treewidth, for which a complete dichotomy is known~\cite{BBJPPV20}, and mim-width, for which there is an infinite number of open cases~\cite{BHMPP20}, this state-of-the-art theorem shows that there are still five open cases (up to an equivalence relation); see also Section~\ref{s-soa}.
From the same theorem we observe that many graph classes are of unbounded clique-width.
However, if a graph class has unbounded clique-width, then this does not mean that a graph problem must be \NP-hard on this class.
For example, {\sc Colouring} is polynomial-time solvable on the (bigenic) class of $(C_4,P_6)$-free graphs~\cite{GHP18}, which contains the class of split graphs and thus has unbounded clique-width~\cite{MR99}.
In this case it turns out that the {\em atoms} (graphs with no clique cut-set) in the class of $(C_4,P_6)$-free graphs {\em do} have bounded clique-width.
This immediately gives us an algorithm for the whole class of $(C_4,P_6)$-free graphs due to Tarjan's decomposition theorem~\cite{Tarjan85}.

In fact, Tarjan's result holds not only for {\sc Colouring}, but also for many other graph problems.
For instance, several other classical graph problems, such as {\sc Minimum Fill-In}, {\sc Maximum Clique}, {\sc Maximum Weighted Independent Set}~\cite{Tarjan85} (see~\cite{Al03} for the unweighted variant) and {\sc Maximum Induced Matching}~\cite{BM11} are polynomial-time solvable on a hereditary graph class~${\cal G}$ if and only if this is the case on the atoms of~${\cal G}$.
Hence, we aim to investigate, in a systematic way, the following natural research question:

\medskip
\noindent
\emph{Which hereditary graph classes of \emph{unbounded} clique-width have the property that their atoms have \emph{bounded} clique-width?} 

\medskip
\noindent
{\bf Known Results.}
For monogenic graph classes, the restriction to atoms does not yield any algorithmic advantages, as shown by Gaspers et al.~\cite{GHP18}.
\begin{theorem}[\cite{GHP18}]\label{t-atoms}
Let~$H$ be a graph.
The class of $H$-free atoms has bounded clique-width if and only if the class of $H$-free graphs has bounded clique-width (so, if and only if~$H$ is an induced subgraph of~$P_4$).
\end{theorem}

\noindent
The result for $(C_4,P_6)$-free graphs~\cite{GHP18} shows that the situation is different for bigenic classes.
We are aware of two more hereditary graph classes~${\cal G}$ with this property, but in both cases $|{\cal F}_{\cal G}|> 2$.
Split graphs, or equivalently, $(C_4,C_5,2P_2)$-free graphs have unbounded clique-width~\cite{MR99}, but split atoms are complete graphs and have clique-width at most~$2$.
Cameron et al.~\cite{CSHK17} proved that $(\mbox{cap},C_4)$-free odd-signable atoms have clique-width at most~$48$, whereas the class of all $(\mbox{cap},C_4)$-free odd-signable graphs contains the class of split graphs and thus has unbounded clique-width.
We refer to~\cite{FFHHL20,FHHHL17} for some examples of algorithms for {\sc Colouring} on hereditary graph classes that rely on boundedness of clique-width of atoms of subclasses.

\medskip
\noindent
{\bf Our Results.} 
Due to Theorem~\ref{t-atoms}, and motivated by the aforementioned algorithmic applications, we focus on the atoms of bigenic graph classes.
Recall that the class of $(C_4,P_6)$-free graphs has unbounded clique-width but its atoms have bounded clique-width~\cite{GHP18}.
This also holds, for instance, for its subclass of $(C_4,2P_2)$-free graphs and thus for $(C_4,P_5)$-free graphs and $(C_4,P_2+\nobreak P_3)$-free graphs.
We determine a new, incomparable case where we forbid~$2P_2$ and~$\overline{P_2+P_3}$ (also known as the {\em paraglider}~\cite{HK20}); see \figurename~\ref{fig:fib-graphs} for illustrations of these forbidden induced subgraphs.

\begin{figure}[h]
\begin{center}
\begin{tabular}{cc}
\scalebox{0.6}{
\begin{minipage}{0.3\textwidth}
\begin{center}
\begin{tikzpicture}[every node/.style={circle,fill, minimum size=0.07cm}]
\node at (0,0) {};
\node at (0,2) {};
\node at (1,0) {};
\node at (1,2) {};
\draw (0,0) -- (0,2);
\draw (1,0) -- (1,2);
\end{tikzpicture}
\end{center}
\end{minipage}}
&
\scalebox{0.6}{
\begin{minipage}{0.3\textwidth}
\begin{center}
\begin{tikzpicture}[every node/.style={circle,fill, minimum size=0.07cm}]
\node at (0,0) {};
\node at (0,2) {};
\node at (2,0) {};
\node at (2,2) {};
\node at (1,1) {};
\draw (0,0) -- (0,2) -- (2,2) -- (2,0) -- cycle;
\draw (2,2) -- (1,1)--(0,0);
\draw (0,2) -- (1,1);
\end{tikzpicture}
\end{center}
\end{minipage}}
\\
\\
$2P_2$ & $\overline{P_2+P_3}$ 
\end{tabular}
\end{center}
\caption{\label{fig:fib-graphs}The two forbidden induced subgraphs from Theorem~\ref{thm:triplet}.}
\end{figure}
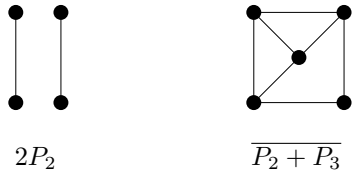

\begin{theorem}\label{thm:triplet}
The class of $(2P_2,\overline{P_2+P_3})$-free atoms has bounded clique-width (whereas the class of $(2P_2,\allowbreak \overline{P_2+P_3})$-free graphs has unbounded clique-width).
\end{theorem}

\noindent
We prove Theorem~\ref{thm:triplet} in Section~\ref{s-triplet} after first giving an outline.
Our approach shares some similarities with the approach Malyshev and Lobanova~\cite{ML17} used to show that {\sc (Weighted) Colouring} is polynomial-time solvable on $(P_5,\overline{P_2+P_3})$-free graphs.
We explain the differences between both approaches and the new ingredients of our proof in detail in Section~\ref{s-triplet}.
Here, we only discuss a complication that makes proving boundedness of clique-width of atoms more difficult in general.
Namely, when working with atoms, we need to be careful with performing complementation operations.
In particular, a class of $(H_1,H_2)$-free graphs has bounded clique-width if only if the class of $(\overline{H_1},\overline{H_2})$-free graphs has bounded clique-width.
However, this equivalence relation no longer holds for classes of $(H_1,H_2)$-free atoms.
For example, $(C_4,P_5)$-free (and even $(C_4,P_6)$-free) atoms have bounded clique-width~\cite{GHP18}, but we will prove that $(\overline{C_4},\overline{P_5})$-free atoms have unbounded clique-width.

We also identify a number of new bigenic graph classes whose atoms already have unbounded clique-width.
We prove this by giving two general techniques for constructing atoms of unbounded clique-width (see Lemmas~\ref{lem:no-false-twin} and~\ref{lem:no-comp-P1or2P1}) and by modifying existing graph constructions for proving unbounded clique-width of the whole class; see Section~\ref{s-unbounded}.
Combining the constructions from Section~\ref{s-unbounded} with Theorem~\ref{thm:triplet} and the state-of-art theorem on clique-width from~\cite{DJP19} yields the following summary; see Section~\ref{s-pre} for definitions of the notation used.

\begin{theorem}\label{thm:classification2-atoms}
For graphs $H_1$ and $H_2$, let~${\cal G}$ be the class of $(H_1,H_2)$-free graphs.
\begin{enumerate}
\item The class of atoms in~${\cal G}$ has bounded clique-width if 
\begin{enumerate}[(i)]
\item \label{thm:classification2-atoms:bdd:P4} $H_1$ or $H_2 \ssi P_4$
\item \label{thm:classification2-atoms:bdd:ramsey} $H_1 = \paw$ or~$K_s$ and $H_2=P_1+\nobreak P_3$ or~$tP_1$ for some $s,t\geq 1$
\item \label{thm:classification2-atoms:bdd:P_1+P_3} $H_1 \ssi \paw$ and $H_2 \ssi K_{1,3}+\nobreak 3P_1,\; K_{1,3}+\nobreak P_2,\;\allowbreak P_1+\nobreak P_2+\nobreak P_3,\;\allowbreak P_1+\nobreak P_5,\;\allowbreak P_1+\nobreak S_{1,1,2},\;\allowbreak P_2+\nobreak P_4,\;\allowbreak P_6,\; \allowbreak S_{1,1,3}$ or $S_{1,2,2}$
\item \label{thm:classification2-atoms:bdd:P_1+P_3b} $H_1 \ssi P_1+P_3$ and $H_2 \ssi \overline{K_{1,3}+\nobreak 3P_1},\; \overline{K_{1,3}+\nobreak P_2},\;\allowbreak \overline{P_1+\nobreak P_2+\nobreak P_3},\;\allowbreak \overline{P_1+\nobreak P_5},\;\allowbreak \overline{P_1+\nobreak S_{1,1,2}},\;\allowbreak \overline{P_2+\nobreak P_4},\;\allowbreak \overline{P_6},\; \allowbreak \overline{S_{1,1,3}}$ or~$\overline{S_{1,2,2}}$
\item \label{thm:classification2-atoms:bdd:2P_1+P_2} $H_1 \ssi \diamondgraph$ and $H_2\ssi P_1+\nobreak 2P_2,\; 3P_1+\nobreak P_2$ or~$P_2+\nobreak P_3$
\item \label{thm:classification2-atoms:bdd:2P_1+P_2b} $H_1 \ssi 2P_1+P_2$ and $H_2\ssi \overline{P_1+\nobreak 2P_2},\; \overline{3P_1+\nobreak P_2}$ or~$\overline{P_2+\nobreak P_3}$
\item \label{thm:classification2-atoms:bdd:P_1+P_4} $H_1 \ssi \gem$ and $H_2 \ssi P_1+\nobreak P_4$ or~$P_5$
\item \label{thm:classification2-atoms:bdd:P_1+P_4b} $H_1 \ssi P_1+P_4$ and $H_2 \ssi \overline{P_5}$
\item \label{thm:classification2-atoms:bdd:K_13} $H_1\ssi K_3+\nobreak P_1$ and $H_2 \ssi K_{1,3}$, 
\item \label{thm:classification2-atoms:bdd:2P1_P3} $H_1\ssi \overline{2P_1+\nobreak P_3}$ and $H_2\ssi 2P_1+\nobreak P_3$
\item $H_1\ssi P_6$ and $H_2\ssi C_4$, or
\item $H_1 \ssi 2P_2$ and $H_2\ssi \overline{P_2+P_3}$.\\[-5pt]
\end{enumerate}
\item The class of atoms in~${\cal G}$ has unbounded clique-width if 
\begin{enumerate}[(i)]
\item \label{thm:classification2-atoms:unbdd:not-in-S} $H_1\not\in {\cal S}$ and $H_2 \not \in {\cal S}$
\item \label{thm:classification2-atoms:unbdd:not-in-co-S} $H_1\notin \overline{{\cal S}}$ and $H_2 \not \in \overline{{\cal S}}$
\item \label{thm:classification2-atoms:unbdd:K_13} $H_1 \si K_3+\nobreak P_1$ and $H_2 \si 4P_1$ or~$2P_2$
\item \label{thm:classification2-atoms:unbdd:K_13b} $H_1 \si K_{1,3}$ and $H_2 \si K_4$ or~$C_4$

\item \label{thm:classification2-atoms:unbdd:2P_1+P_2} $H_1 \si \diamondgraph$ and $H_2 \si K_{1,3},\; 5P_1,\; P_2+\nobreak P_4$ or $P_1+\nobreak P_6$
\item \label{thm:classification2-atoms:unbdd:2P_1+P_2b} $H_1 \si 2P_1+\nobreak P_2$ and $H_2 \si K_3+P_1,\; K_5,\; \overline{P_2+\nobreak P_4}$ or~$\overline{P_6}$
\item \label{thm:classification2-atoms:unbdd:3P_1} $H_1 \si K_3$ and $H_2 \si 2P_1+\nobreak 2P_2,\; 2P_1+\nobreak P_4,\; 4P_1+\nobreak P_2,\; 3P_2$ or~$2P_3$
\item \label{thm:classification2-atoms:unbdd:3P_1b} $H_1 \si 3P_1$ and $H_2 \si \overline{2P_1+\nobreak 2P_2},\; \overline{2P_1+\nobreak P_4},\; \overline{4P_1+\nobreak P_2},\; \overline{3P_2}$ or~$\overline{2P_3}$
\item \label{thm:classification2-atoms:unbdd:4P_1} $H_1 \si K_4$ and $H_2 \si P_1 +\nobreak P_4,\; 3P_1+\nobreak P_2$ or~$2P_2$ 
\item \label{thm:classification2-atoms:unbdd:4P_1b} $H_1 \si 4P_1$ and $H_2 \si \gem,\; \overline{3P_1+\nobreak P_2}$ or~$C_4$
\item \label{thm:classification2-atoms:unbdd:gem} $H_1 \si \gem,\; \overline{P_1+2P_2}$ or~$\overline{P_2+P_3}$ and $H_2 \si P_1+\nobreak 2P_2$ or~$P_6$
\item \label{thm:classification2-atoms:unbdd:gemb} $H_1 \si P_1+\nobreak P_4$ and $H_2 \si \overline{P_1+\nobreak 2P_2}$, or
\item $H_1 \si 2P_2$ and $H_2 \si \overline{P_2+P_4},\; \overline{3P_2}$ or~$\overline{P_5}$.
\end{enumerate}
\end{enumerate}
\end{theorem}

\noindent
We prove Theorem~\ref{thm:classification2-atoms} in Section~\ref{s-main}.
Due to this theorem, we are left with 18 open cases, which we list in Section~\ref{s-main} (see Open Problem~\ref{o-atoms}).
In Section~\ref{s-con} we discuss directions for future work.

\section{Preliminaries}\label{s-pre}
We first give some general graph terminology and notation, followed by some terminology related to clique-width.

\subsection{General Graph Terminology}

Let~$G$ be a graph.
For a subset $S\subseteq V(G)$, the subgraph of~$G$ {\em induced by}~$S$ is the graph~$G[S]$, which has vertex set~$S$ and edge set $\{uv\; |\; uv\in E(G), u,v\in S\}$.
If $S=\{s_1,\ldots,s_r\}$, we may write $G[s_1,\ldots,s_r]$ instead of $G[\{s_1,\ldots,s_r\}]$.
We write $F\ssi G$ to indicate that~$F$ is an induced subgraph of~$G$.
For a subset $S\subseteq V(G)$, we let $G\setminus S = G[V(G)\setminus S]$.
A {\em (connected) component} of~$G$ is a maximal connected subgraph of~$G$.
The {\em neighbourhood} of a vertex $u\in V(G)$ is the set $N(u)=\{v\in V(G)\; |\; uv\in E(G)\}$.
Two vertices in $G$ are {\em false twins} if they have the same neighbourhood; note that such vertices must be non-adjacent.

An \emph{independent set} in a graph~$G$ is a subset of~$V(G)$ that consists of pairwise non-adjacent vertices; a \emph{clique} in~$G$ is a subset of pairwise adjacent vertices.
A clique $K\subseteq V(G)$ is a \emph{clique cut-set} of~$G$ if $G\setminus K$ is disconnected.
A graph with no clique cut-sets is an {\em atom}.
Note that atoms are connected by definition.

Let~$H$ be a graph.
A graph~$G$ is {\em $H$-free} if~$G$ does not contain~$H$ as an induced subgraph.
Let $\{H_1,\ldots,H_p\}$ be a set of graphs.
Then~$G$ is {\em $(H_1,\ldots,H_p)$-free} if it is $H_i$-free for all $i\in\{1,\ldots,p\}$.

Let~$X$ and~$Y$ be two disjoint vertex subsets of a graph~$G$.
The edges between~$X$ and~$Y$ form a {\em matching} if every vertex in~$X$ is adjacent to at most one vertex in~$Y$ and vice versa.
A vertex $x \in V(G) \setminus Y$ is {\em (anti-)complete} to~$Y$ if it is (non-)adjacent to every vertex in~$Y$.
Similarly, $X$ is {\em complete} to~$Y$ if every vertex of~$X$ is complete to~$Y$ and {\em anti-complete} to~$Y$ if every vertex of~$X$ is anti-complete to~$Y$.
A vertex $u\in V(G)$ is {\em dominating} if it is complete to $V(G)\setminus \{u\}$.
For $k \geq 1$, a \emph{$k$-subdivision} of~$G$ is the operation of replacing each edge~$uv$ of~$G$ with a $(k+1)$-edge path, whose end-vertices are identified with~$u$ and~$v$, respectively.
The {\em complement}~$\overline{G}$ of~$G$ has vertex set $V(\overline{G})=V(G)$ and edge set $E(\overline{G})=\{uv\; |\; u,v\in V(G), u \neq v, uv\notin E(G)\}$.
The {\em line graph} of~$G$ is the graph with vertex set~$E(G)$ and an edge between two vertices~$e_1$ and~$e_2$ if and only if~$e_1$ and~$e_2$ share a common end-vertex in~$G$.

A graph is {\em bipartite} if its vertex set can be partitioned into two (possibly empty) independent sets.
A graph is {\em complete multi-partite} if its vertex can be partitioned into~$r$ independent sets $V_1,\ldots,V_r$ for some integer $r\geq 1$ such that~$V_i$ is complete to~$V_j$ for every pair~$i,j$ with $1\leq i<j\leq r$; if $r=2$, we say that the graph is {\em complete bipartite}.

The graph~$G_1+\nobreak G_2$ is the {\em disjoint union} of two vertex-disjoint graphs~$G_1$ and~$G_2$ and has vertex set $V(G_1)\cup V(G_2)$ and edge set $E(G_1)\cup E(G_2)$.
The graph~$rG$ is the disjoint union of~$r$ copies of a graph~$G$.

The graphs $C_t$, $K_t$, and~$P_t$ denote the cycle, complete graph, and path on~$t$ vertices, respectively.
The graph~$K_{s,t}$ denotes the complete bipartite graph whose two partition classes contain~$s$ and~$t$ vertices, respectively.
The $\paw$ is the graph~$\overline{P_1+P_3}$, the $\diamondgraph$ is the graph~$\overline{2P_1+P_2}$, and the $\gem$ is the graph~$\overline{P_1+P_4}$.
The claw is the graph with vertices $x$, $y_1$, $y_2$, $y_3$ and edges~$xy_i$ for $i \in \{1,2,3\}$.
The {\em subdivided claw}~$S_{h,i,j}$, for $1\leq h\leq i\leq j$ is the tree with one vertex~$x$ of degree~$3$ and exactly three leaves, which are of distance~$h$,~$i$ and~$j$ from~$x$, respectively.
We let~${\cal S}$ denote the class of graphs every connected component of which is either a subdivided claw or a path on at least one vertex.
Note that $S_{1,1,1}=K_{1,3}$ and that this graph is isomorphic to the claw.

A graph is {\em chordal} if it has no induced cycles on more than four vertices, that is, if it is $(C_4,C_5,\ldots)$-free.
A graph is {\em co-chordal} if its complement is chordal.
A graph is {\em split} if its vertex set can be partitioned into a clique~$K$ and an independent set~$I$.
A graph is {\em bipartite chain} if it is bipartite, say with bipartition classes~$X$ and~$Y$, such that the vertices of~$X$ can be ordered $x_1,\ldots,x_p$ with the property that $N(x_1)\subseteq N(x_2) \subseteq \ldots \subseteq N(x_p)$.
The following observation is well known and easy to see.

\begin{lemma}\label{l-e1}
A graph is bipartite chain if and only if it is bipartite and $2P_2$-free.
\end{lemma}

\noindent

\subsection{Clique-width}
The {\em clique-width} of a graph~$G$, denoted by~$\cw(G)$, is the minimum number of labels needed to construct~$G$ using the following four operations:
\begin{enumerate}
\item create a new graph consisting of a single vertex~$v$ with label~$i$;
\item take the disjoint union of two labelled graphs~$G_1$ and~$G_2$;
\item add an edge between every vertex with label~$i$ and every vertex with label~$j$ ($i\neq j$);
\item relabel every vertex with label~$i$ to have label~$j$.
\end{enumerate}
\noindent
A class of graphs~${\cal G}$ has {\em bounded} clique-width if there is a constant~$c$ such that $\cw(G)\leq c$ for every $G\in {\cal G}$; otherwise the clique-width of~${\cal G}$ is {\em unbounded}.

For an induced subgraph~$G'$ of a graph~$G$, the {\em subgraph complementation} acting on~$G$ with respect to~$G'$ replaces every edge of~$G'$ by a non-edge, and vice versa.
Hence, the resulting graph has vertex set~$V(G)$ and edge set $(E(G) \setminus E(G')) \cup E(\overline{G'})$.
For two disjoint vertex subsets~$S$ and~$T$ in~$G$, the {\em bipartite complementation} acting on~$G$ with respect to~$S$ and~$T$ replaces every edge with one end-vertex in~$S$ and the other in~$T$ by a non-edge and vice versa.

For a constant $k\geq 0$ and a graph operation~$\gamma$, a graph class~${\cal G'}$ is {\em $(k,\gamma)$-obtained} from a graph class~${\cal G}$ if
\begin{enumerate}[(i)]
\item every graph in~${\cal G'}$ is obtained from a graph in~${\cal G}$ by performing~$\gamma$ at most~$k$ times, and\\[-18pt]
\item for every $G\in {\cal G}$, there exists at least one graph in~${\cal G'}$ obtained from~$G$ by performing~$\gamma$ at most~$k$ times.
\end{enumerate}
Then~$\gamma$ {\em preserves} boundedness of clique-width if for every constant~$k$ and every graph class~${\cal G}$, every graph class~${\cal G}'$ that is $(k,\gamma)$-obtained from~${\cal G}$ has bounded clique-width if and only if~${\cal G}$ has bounded clique-width.

\begin{enumerate}[\bf F{a}ct 1.]
\item \label{fact:del-vert}Vertex deletion preserves boundedness of clique-width~\cite{LR04}.
\item \label{fact:comp}Subgraph complementation preserves boundedness of clique-width~\cite{KLM09}.
\item \label{fact:bip}Bipartite complementation preserves boundedness of clique-width~\cite{KLM09}.
\end{enumerate}

We finish this section with making two further observations that we will need later on.
First, we make the following well-known observation on bipartite chain graphs, which is readily seen.

\begin{lemma}\label{l-e2}
Bipartite chain graphs have clique-width at most~$3$.
\end{lemma}

Let $G=(K\cup I, E)$ be a split graph with clique~$K$ and independent set~$I$.
If there is a vertex $v \in I$ with $N(v) \subsetneq K$, then~$N(v)$ is a clique cut-set of~$G$.
Furthermore, if $|I| > 1$ then~$K$ is a clique cut-set.
It follows that split atoms are complete graphs.
It is readily seen that complete graphs have clique-width at most~$2$.
Hence, we can make the following observation.

\begin{lemma}\label{l-e3}
Split atoms are complete graphs and have clique-width at most~$2$.
\end{lemma}

\section{The Proof of Theorem~\ref{thm:triplet}}\label{s-triplet}
Here, we prove Theorem~\ref{thm:triplet}, namely that the class of $(2P_2,\overline{P_2+P_3})$-free atoms has bounded clique-width.
Our approach is based on the following three claims:
\begin{enumerate}[(i)]
\renewcommand{\theenumi}{(\roman{enumi})}
\renewcommand{\labelenumi}{(\roman{enumi})}
\item \label{cl-hasC5}$(2P_2,\overline{P_2+P_3})$-free atoms with an induced~$C_5$ have bounded clique-width.
\item \label{cl-hasC4}$(2P_2,\overline{P_2+P_3})$-free atoms with an induced~$C_4$ have bounded clique-width.
\item \label{cl-split}$(C_4,C_5,2P_2,\overline{P_2+P_3})$-free atoms have bounded clique-width.
\end{enumerate}
We prove Claims~\ref{cl-hasC5} and~\ref{cl-hasC4} in Lemmas~\ref{l-tripletwithc5} and~\ref{l-c5free}, respectively, whereas Claim~\ref{cl-split} follows from the fact that $(C_4,C_5,2P_2)$-free graphs are split graphs and so, by Lemma~\ref{l-e3}, the atoms in this class are complete graphs and therefore have clique-width at most~$2$.
We partition the vertex set of an arbitrary $(2P_2,\overline{P_2+P_3})$-free atom~$G$ into a number of different subsets according to their neighbourhoods in an induced~$C_5$ in Lemma~\ref{l-tripletwithc5} or an induced~$C_4$ in Lemma~\ref{l-c5free}.
We then analyse the properties of these different subsets of~$V(G)$ and how they are connected to each other, and use this knowledge to apply a number of appropriate vertex deletions, subgraph complementations and bipartite complementations.
These operations will modify~$G$ into a graph~$G'$ that is a disjoint union of a number of smaller ``easy'' graphs known to have ``small'' clique-width.
We then use Facts~\ref{fact:del-vert}--\ref{fact:bip} to conclude that~$G$ also has small clique-width.

This approach works, as we will:
\begin{itemize}
\item apply the vertex deletions, subgraph complementations, and bipartite complementations only a constant number of times; 
\item not use the properties of being an atom or being $(2P_2,\overline{P_2+P_3})$-free once we ``leave the graph class'' due to applying the above graph operations.
\end{itemize}

Our approach is similar to the approach used by Malyshev and Lobanova~\cite{ML17} for showing that {\sc Colouring} is polynomial-time solvable on the superclass of $(P_5,\overline{P_2+P_3})$-free graphs.
However, we note the following two techniques that can be used in the design of algorithms for {\sc Colouring} on hereditary graph classes, but cannot be used for proving boundedness of clique-width.
Both these techniques were used in~\cite{ML17}.

\paragraph{\bf 1. Prime atoms restriction.}
A set $X\subseteq V(G)$ is a \emph{module} if all vertices in~$X$ have the same set of neighbours in $V(G)\setminus X$.
A module~$X$ in a graph~$G$ is {\em trivial} if it contains either all or at most one vertex of~$G$.
A graph~$G$ is \emph{prime} if it has no non-trivial modules.
To solve {\sc Colouring} in polynomial time on some hereditary graph class~${\cal G}$, one may restrict to prime atoms in~${\cal G}$~\cite{HL}.
Malyshev and Lobanova proved that $(P_5,\overline{P_2+P_3})$-free prime atoms with an induced~$C_5$ are $3P_1$-free or have a bounded number of vertices.
In both cases, {\sc Colouring} can be solved in polynomial time.
We cannot make the pre-assumption that our atoms are prime.
To see this, let~$G$ be a split graph that is not complete.
Add two new non-adjacent vertices~$u$ and~$v$ to~$G$ and make them complete to the rest of~$V(G)$.
Let~${\cal G}$ be the (hereditary) graph class that consists of all these ``enhanced'' split graphs and their induced subgraphs.
These enhanced split graphs are atoms, which have unbounded clique-width due to Fact~\ref{fact:del-vert} and the fact that split graphs have unbounded clique-width~\cite{MR99}.
However, the prime atoms in~${\cal G}$ are~$P_1$ and~$P_2$,\footnote{Let~$D$ be a prime atom in~${\cal G}$.
As~$D$ is prime, $D$ cannot contain both~$u$ and~$v$.
This means that~$D$ is split graph.
By Lemma~\ref{l-e3}, as~$D$ is an atom, it must be a complete graph.
As~$D$ is prime, this implies that $|V(D)|\leq 2$.}
which have clique-width~$1$ and~$2$, respectively.
 
\paragraph{\bf 2. Perfect graphs restriction.}
Malyshev and Lobanova observed that $(P_5,\overline{P_2+P_3},C_5)$-free graphs are perfect.
Hence, {\sc Colouring} can be solved in polynomial time on such graphs~\cite{GLS84}.
However, being perfect does not imply boundedness of clique-width.
For instance, split graphs are perfect graphs with unbounded clique-width~\cite{MR99}.

\begin{lemma}\label{l-tripletwithc5}
The class of $(2P_2, \overline{P_2+P_3})$-free atoms that contain an induced~$C_5$ has bounded clique-width.
\end{lemma}

\begin{proof}
Suppose~$G$ is a $(2P_2,\overline{P_2+P_3})$-free atom containing an induced cycle~$C$ on five vertices, say $v_1,\ldots,v_5$ in that order.
For $S \subseteq \{1,\ldots,5\}$, let~$V_S$ be the set of vertices $x \in V(G) \setminus V(C)$ such that $N(x)\cap V(C)=\{v_i \;|\; i \in S\}$.

To simplify notation, in the following claims, subscripts on vertices and vertex sets should be interpreted modulo~$5$ and whenever possible we will write~$V_i$ instead of~$V_{\{i\}}$, write~$V_{i,j}$ instead of~$V_{\{i,j\}}$, and so on.

\begin{claimm}\label{clm:C5-V1V12V123-empty}
For $i \in \{1,\ldots,5\}$, $V_i \cup V_{i,i+1} \cup V_{i-1,i,i+1}$ are empty.
\end{claimm}
\begin{inproof}
Suppose, for contradiction, that $x \in V_2 \cup V_{2,3} \cup V_{1,2,3}$.
Then $G[x,v_2,v_4,v_5]$ is a~$2P_2$, a contradiction.
The claim follows by symmetry.
\end{inproof}

\medskip
\noindent
By Claim~\ref{clm:C5-V1V12V123-empty}, the only non-empty sets~$V_S$ are those of the form $V_\emptyset$, $V_{i,i+2}$, $V_{i,i+1,i+3}$, $V_{i,i+1,i+2,i+3}$ and~$V_{1,2,3,4,5}$.
We now prove a sequence of claims.

\begin{claimm}\label{clm:C5-V13V0-indep}
For $i \in \{1,\ldots,5\}$, $V_\emptyset \cup V_{i,i+2}$ is independent.
\end{claimm}
\begin{inproof}
Suppose, for contradiction, that $x,y \in V_\emptyset \cup V_{1,3}$ are adjacent.
Then $G[v_4,v_5,x,y]$ is a~$2P_2$, a contradiction.
The claim follows by symmetry.
\end{inproof}

\begin{claimm}\label{clm:C5-V124V1234-small}
For $i \in \{1,\ldots,5\}$, $|V_{i,i+1,i+3} \cup V_{i,i+1,i+2,i+3}| \leq 1$.
\end{claimm}
\begin{inproof}
Suppose, for contradiction that there are distinct vertices $x,y \in V_{1,2,4} \cup V_{1,2,3,4}$.
Then $G[v_1,v_4,x,v_5,y]$ or $G[x,y,v_1,v_4,v_2]$ is a~$\overline{P_2+P_3}$ if~$x$ is adjacent or non-adjacent to~$y$, respectively, a contradiction.
The claim follows by symmetry.
\end{inproof}

\begin{claimm}\label{clm:C5-V13V35-almost-anti}
For $i \in \{1,\ldots,5\}$, there is at most one edge between~$V_{i,i+2}$ and~$V_{i,i-2}$.
\end{claimm}
\begin{inproof}
Suppose, for contradiction, that a vertex $x \in V_{1,3}$ has two neighbours $y,y' \in V_{1,4}$.
By Claim~\ref{clm:C5-V13V0-indep}, the sets~$V_{1,3}$ and~$V_{1,4}$ are independent.
In particular, this means that~$y$ is non-adjacent to~$y'$.
Therefore $G[y,y',x,v_4,v_1]$ is a~$\overline{P_2+P_3}$, a contradiction.
It follows that every vertex in~$V_{1,3}$ has at most one neighbour in~$V_{1,4}$.
By symmetry, every vertex in~$V_{1,4}$ has at most one neighbour in~$V_{1,3}$ and so the edges between~$V_{1,3}$ and~$V_{1,4}$ form a matching.
Since~$G$ is $2P_2$-free, it follows that there is at most one edge between~$V_{1,3}$ and~$V_{1,4}$.
The claim follows by symmetry.
\end{inproof}

\newpage
\begin{claimm}\label{clm:C5-V13V24-comp}
For $i \in \{1,\ldots,5\}$, $V_{i,i+2}$ is complete to $V_{i-1,i+1} \cup V_{i+1,i+3}$.
\end{claimm}
\begin{inproof}
Suppose, for contradiction, that $x \in V_{1,3}$ is non-adjacent to $y \in V_{2,4}$.
Then $G[x,v_1,y,v_4]$ is a~$2P_2$, a contradiction.
The claim follows by symmetry.
\end{inproof}

\begin{claimm}\label{clm:C5-V12345-domin}
If $x \in V_{1,2,3,4,5}$, then~$x$ is complete to $V(G) \setminus \{x\}$.
In particular, this implies that~$V_{1,2,3,4,5}$ is a clique.
\end{claimm}
\begin{inproof}
Let $x \in V_{1,2,3,4,5}$ and suppose, for contradiction, that $y \in V(G) \setminus \{x\}$ is non-adjacent to~$x$.
Clearly $y \notin V(C)$.
If $y \in V_{1,2,3,4,5}$, then $G[x,y,v_1,v_4,v_2]$ is a~$\overline{P_2+P_3}$, a contradiction.
If~$y$ is adjacent to~$v_i$ and~$v_{i+2}$, but not to~$v_{i+1}$ for some $i \in \{1,\ldots,5\}$, then $G[v_i,v_{i+2},v_{i+1},y,x]$ is a~$\overline{P_2+P_3}$, a contradiction.
By Claim~\ref{clm:C5-V1V12V123-empty}, it follows that $y \in V_\emptyset$.
Note that this implies that every vertex of~$V_{1,2,3,4,5}$ is adjacent to every other vertex in $V(G) \setminus V_\emptyset$.

Since~$G$ is an atom, $N(y)$ cannot be a clique, and so it must contain two non-adjacent vertices, say~$u$ and~$v$.
By Claim~\ref{clm:C5-V13V0-indep}, $u,v \notin V_\emptyset$ and for all $i \in \{1,\ldots,5\}$, $u,v \notin V_{i,i+2}$.
Since every vertex of~$V_{1,2,3,4,5}$ is adjacent to every other vertex in $V(G) \setminus V_\emptyset$, neither~$u$ nor~$v$ is equal to~$x$ and, furthermore, $x$ is adjacent to both~$u$ and~$v$.
By Claim~\ref{clm:C5-V1V12V123-empty}, it follows that~$u$ and~$v$ must each have at least three neighbours in~$C$.
Therefore~$u$ and~$v$ must have a common neighbour in~$C$; let~$v_i$ be such a common neighbour.
Now $G[u,v,x,y,v_i]$ is a~$\overline{P_2+P_3}$, a contradiction.
This completes the proof of the claim.
\end{inproof}

\begin{figure}[ht]
\centering
\includegraphics[scale=0.8, page=1]{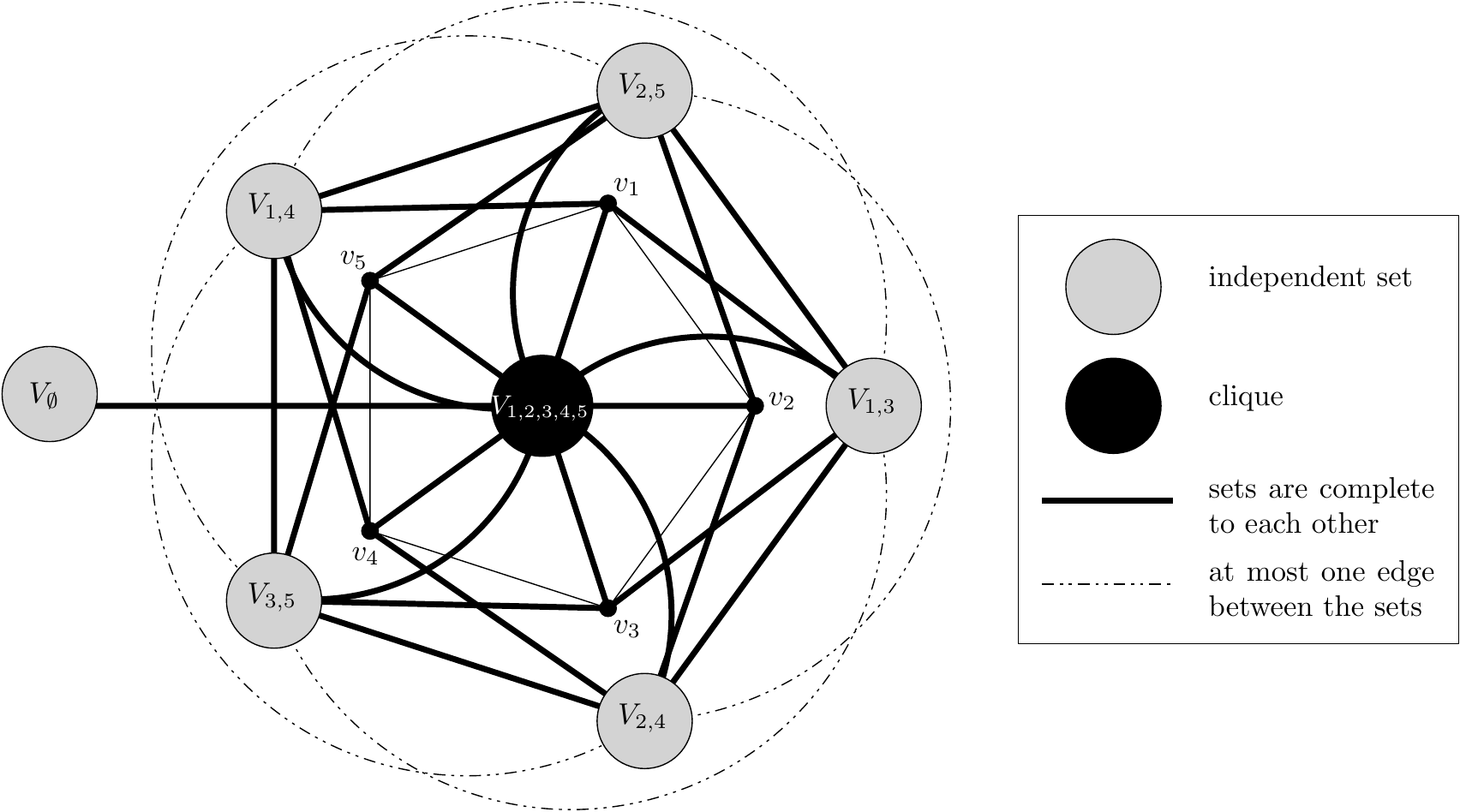}
\caption{The configuration of the sets in the proof of Lemma~\ref{l-tripletwithc5} obtained after deleting the vertex sets~$V_{i,i+1,i+3}$ and~$V_{i,i+1,i+2,i+3}$ for $i \in \{1,2,3,4,5\}$.
Note that vertices in~$V_\emptyset$ now have no neighbours outside~$V_{1,2,3,4,5}$.}\label{fig:lemmawithc5}
\end{figure}

\medskip
\noindent
We now show how to use a bounded number of vertex deletions, complementations and bipartite complementations to change~$G$ into an edgeless graph.
First, by Claim~\ref{clm:C5-V124V1234-small}, we can make~$V_{i,i+1,i+3}$ and~$V_{i,i+1,i+2,i+3}$ empty for all $i \in \{1,\ldots,5\}$ by deleting at most five vertices.
See \figurename~\ref{fig:lemmawithc5} for an illustration of the resulting graph.
Next, by Claim~\ref{clm:C5-V12345-domin} we can apply a bipartite complementation between~$V_{1,2,3,4,5}$ and the rest of the graph to disconnect $G[V_{1,2,3,4,5}]$ from it.
Next, by Claim~\ref{clm:C5-V12345-domin} we can apply a complementation to~$V_{1,2,3,4,5}$, which turns it into an independent set.
Now, by Claim~\ref{clm:C5-V1V12V123-empty}, the only other vertices remaining are those in~$C$, those in~$V_\emptyset$ and those in~$V_{i,i+2}$ for $i \in \{1,\ldots,5\}$.
Next, by Claim~\ref{clm:C5-V13V35-almost-anti}, we can make~$V_{i,i+2}$ anti-complete to~$V_{i,i-2}$ for all $i \in \{1,\ldots,5\}$ by deleting at most five vertices.
By Claim~\ref{clm:C5-V13V0-indep}, the only remaining edges are those between $V_{i-1,i+1} \cup \{v_i\}$ and $V_{i,i+2} \cup \{v_{i+1}\}$ for $i \in \{1,\ldots,5\}$.
By Claim~\ref{clm:C5-V13V24-comp} combined with the definition of~$V_{i,i+2}$, we can apply a bipartite complementation between each of these pairs to remove all remaining edges of the graph.
Thus, applying at most ten vertex deletions, six bipartite complementations and one complementation operation to~$G$, we obtain an edgeless graph, which has clique-width~$1$.
By Facts~\ref{fact:del-vert}, \ref{fact:comp} and~\ref{fact:bip}, it follows that~$G$ has bounded clique-width.
\end{proof}

\begin{lemma}\label{l-c5free}
The class of $(2P_2, \overline{P_2+P_3})$-free atoms that contain an induced~$C_4$ has bounded clique-width.
\end{lemma}

\begin{proof}
Suppose~$G$ is a $(2P_2,\overline{P_2+P_3})$-free atom containing an induced cycle~$C$ on four vertices, say $v_1,\ldots,v_4$ in that order.
By Lemma~\ref{l-tripletwithc5}, we may assume that~$G$ is $C_5$-free.
For $S \subseteq \{1,\ldots,4\}$, let~$V_S$ be the set of vertices $x \in V(G) \setminus V(C)$ such that $N(x)\cap V(C)=\{v_i \;|\; i \in S\}$.

To simplify notation, in the following claims, subscripts on vertices and vertex sets should be interpreted modulo~$4$ and whenever possible we will write~$V_i$ instead of~$V_{\{i\}}$, write~$V_{i,j}$ instead of~$V_{\{i,j\}}$, and so on.

\begin{claimm}\label{clm:C4-V123-empty}
For $i \in \{1,\ldots,4\}$, $V_{i,i+1,i+2}$ is empty.
\end{claimm}
\begin{inproof}
Suppose, for contradiction, that $x \in V_{1,2,3}$.
Then $G[v_1,v_3,v_2,v_4,x]$ is a~$\overline{P_2+P_3}$, a contradiction.
The claim follows by symmetry.
\end{inproof}

\smallskip
\noindent
See \figurename~\ref{fig:c5free-init} for an illustration of the remaining sets~$V_S$ that can be non-empty.

\begin{figure}[ht]
\centering
\includegraphics[scale=0.8, page=2]{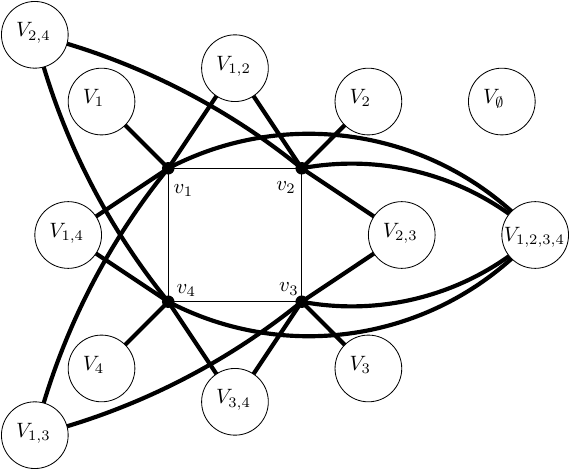}
\caption{The possible non-empty sets~$V_S$ in the initial situation in the proof of Lemma~\ref{l-c5free}.
Edges between the different sets~$V_S$ are not drawn.} \label{fig:c5free-init}
\end{figure}

\begin{claimm}\label{clm:C4-V0V1V2V12-indep}
For $i \in \{1,\ldots,4\}$, $V_\emptyset \cup V_i \cup V_{i+1} \cup V_{i,i+1}$ is an independent set.
\end{claimm}
\begin{inproof}
Suppose, for contradiction, that $x,y \in V_\emptyset \cup V_1 \cup V_2 \cup V_{1,2}$ are adjacent.
Then $G[x,y,v_3,v_4]$ is a~$2P_2$, a contradiction.
The claim follows by symmetry.
\end{inproof}

\begin{claimm}\label{clm:C4-V12-V13-indep}
For $i \in \{1,\ldots,4\}$, $V_{i,i+1} \cup V_{i,i+2}$ and $V_{i,i+1} \cup V_{i+1,i+3}$ are independent sets.
\end{claimm}
\begin{inproof}
Suppose, for contradiction, that $x,y \in V_{1,2} \cup V_{1,3}$ are adjacent.
By Claim~\ref{clm:C4-V0V1V2V12-indep}, $x$ and~$y$ cannot both be in~$V_{1,2}$, so assume without loss of generality that $x \in V_{1,3}$.
Now $G[x,v_2,v_1,v_3,y]$ is a~$\overline{P_2+P_3}$ (regardless of whether $y \in V_{1,2}$ or $y \in V_{1,3}$), a contradiction.
The claim follows by symmetry.
\end{inproof}

\begin{claimm}\label{clm:C4-V1234-bdd-cw}
$G[V_{1,2,3,4}]$ is $(P_1+\nobreak P_2)$-free and so it has bounded clique-width.
\end{claimm}
\begin{inproof}
Suppose, for contradiction, that $x,y,y' \in V_{1,2,3,4}$ induce a~$P_1+\nobreak P_2$ in~$G$.
Then $G[v_1,v_3,y,x,y']$ is a~$\overline{P_2+P_3}$, a contradiction.
Therefore~$G[V_{1,2,3,4}]$ is $(P_1+\nobreak P_2)$-free and so~$P_4$-free, so it has bounded clique-width by Theorem~\ref{t-atoms}.
\end{inproof}

\newpage
\begin{claimm}\label{clm:C4-V13-V1234-comp}
For $i \in \{1,2\}$, $V_{i,i+2}$ is complete to~$V_{1,2,3,4}$.
\end{claimm}
\begin{inproof}
Suppose, for contradiction, that $x \in V_{1,3}$ is non-adjacent to $y \in V_{1,2,3,4}$.
Then $G[v_1,v_3,v_2,x,y]$ is a~$\overline{P_2+P_3}$, a contradiction.
The claim follows by symmetry.
\end{inproof}

\begin{claimm}\label{clm:C4-V1V12-or-V3V23-empty}
For $i \in \{1,2,3,4\}$ either $V_{i-1} \cup V_{i-1,i}$ or $V_{i,i+1} \cup V_{i+1}$ is empty.
\end{claimm}
\begin{inproof}
Suppose, for contradiction, that $x \in V_1 \cup V_{1,2}$ and $y \in V_{2,3} \cup V_3$.
Then $G[v_1,x,y,v_3,v_4]$ is a~$C_5$ or $G[x,v_1,y,v_3]$ is a~$2P_2$ if~$x$ is adjacent or non-adjacent to~$y$, respectively, a contradiction.
The claim follows by symmetry.
\end{inproof}

\begin{claimm}\label{clm:C4-V0-non-empty-nice-nbhd}
If $x \in V_\emptyset$, then~$x$ has at least two neighbours in one of~$V_{1,3}$ and~$V_{2,4}$ and is anti-complete to the other.
Furthermore, in this case~$x$ is complete to~$V_{1,2,3,4}$.
\end{claimm}
\begin{inproof}
Suppose $x \in V_\emptyset$.
Since~$G$ is not an atom, $N(x)$ cannot be a clique, and so must contain two non-adjacent vertices $y,y'$.
By Claims~\ref{clm:C4-V123-empty} and~\ref{clm:C4-V0V1V2V12-indep}, and the definition of~$V_\emptyset$, it follows that $y,y' \in V_{1,3} \cup V_{2,4} \cup V_{1,2,3,4}$.
If $y,y' \in V_{1,2,3,4}$, then $G[y,y',v_1,x,v_2]$ is a~$\overline{P_2+P_3}$, a contradiction.
By Claim~\ref{clm:C4-V13-V1234-comp}, $V_{1,2,3,4}$ is complete to $V_{1,3} \cup V_{2,4}$, so it follows that $y,y' \in V_{1,3} \cup V_{2,4}$.
If $y \in V_{1,3}$ and $y' \in V_{2,4}$, then $G[v_1,v_2,y',x,y]$ is a~$C_5$, a contradiction.
It follows that $y,y' \in V_{1,3}$ or $y,y' \in V_{2,4}$.

Suppose $y,y' \in V_{1,3}$.
If $z \in V_{2,4}$ is a neighbour of~$x$, then~$z$ must be adjacent to~$y$ and~$y'$ (since, as shown above, $x$ cannot have a pair of non-adjacent neighbours one of which is in~$V_{1,3}$ and the other of which is in~$V_{2,4}$), in which case $G[y,y',x,v_1,z]$ is a~$\overline{P_2+P_3}$, a contradiction.
Therefore~$x$ cannot have a neighbour in~$V_{2,4}$.
If $z \in V_{1,2,3,4}$ is a non-neighbour of~$x$, then~$z$ must be adjacent to~$y$ and~$y'$ by Claim~\ref{clm:C4-V13-V1234-comp}, so $G[y,y',v_1,x,z]$ is a~$\overline{P_2+P_3}$, a contradiction.
Therefore~$x$ is complete to~$V_{1,2,3,4}$.
The claim follows by symmetry.
\end{inproof}

\begin{claimm}\label{clm:C4-V12V34-small}
For $i \in \{1,2\}$, $|V_{i,i+1} \cup V_{i+2,i+3}| \leq 2$.
\end{claimm}
\begin{inproof}
Suppose, for contradiction, that $|V_{1,2} \cup V_{3,4}| \geq 3$.
First note that if $x \in V_{1,2}$, $y \in V_{3,4}$ are non-adjacent, then $G[v_1,x,v_3,y]$ is a~$2P_2$, a contradiction.
Therefore~$V_{1,2}$ is complete to~$V_{3,4}$.
By Claim~\ref{clm:C4-V0V1V2V12-indep}, both~$V_{1,2}$ and~$V_{3,4}$ are independent sets.
If $x \in V_{1,2}$ and $y,y' \in V_{3,4}$, then $G[y,y',v_3,x,v_4]$ is a~$\overline{P_2+P_3}$, a contradiction.
By symmetry, we conclude that either~$V_{1,2}$ or~$V_{3,4}$ is empty.

Suppose~$V_{3,4}$ is empty, so~$V_{1,2}$ contains at least three vertices and let $x \in V_{1,2}$ be such a vertex.
Since~$G$ is an atom, $N(x)$ cannot be a clique, so~$x$ must have two neighbours $y,y'$ that are non-adjacent.
By Claims~\ref{clm:C4-V123-empty}, \ref{clm:C4-V0V1V2V12-indep}, \ref{clm:C4-V12-V13-indep} and~\ref{clm:C4-V1V12-or-V3V23-empty}, and the definition of~$V_{1,2}$, every neighbour of $x \in V_{1,2}$ lies in $\{v_1,v_2\} \cup V_{1,2,3,4}$.
Since~$v_1$ is complete to $\{v_2\} \cup V_{1,2,3,4}$ and~$v_2$ is complete to $\{v_1\} \cup V_{1,2,3,4}$, it follows that $y,y' \in V_{1,2,3,4}$.
Now $G[y,y',v_1,v_3,x]$ is a~$\overline{P_2+P_3}$, a contradiction.
The claim follows by symmetry.
\end{inproof}

\begin{claimm}\label{clm:C4-V_1-at-most-1-V13-nbr}
For $i \in \{1,2,3,4\}$, $V_i$ is complete to~$V_{1,2,3,4}$ and at most one vertex of~$V_{i,i+2}$ has neighbours in~$V_i$.
\end{claimm}
\begin{inproof}
Suppose $x \in V_1$.
Since~$G$ is an atom, $x$ must have two neighbours $y,y'$ that are non-adjacent.
By Claims~\ref{clm:C4-V123-empty}, \ref{clm:C4-V0V1V2V12-indep} and~\ref{clm:C4-V1V12-or-V3V23-empty}, and the definition of~$V_1$, every neighbour of~$x$ lies in $\{v_1\} \cup V_{1,3} \cup V_{2,4} \cup V_{1,2,3,4}$.
If $y,y' \in V_{1,3} \cup V_{1,2,3,4}$, then $G[y,y',v_1,v_3,x]$ is a~$\overline{P_2+P_3}$, a contradiction.
The vertex~$v_1$ is complete to $V_{1,3} \cup V_{1,2,3,4}$.
Therefore without loss of generality, we may assume $y \in V_{2,4}$.
Furthermore, note that~$V_{1,3}$ is an independent set by Claim~\ref{clm:C4-V12-V13-indep}, so~$x$ has at most one neighbour in~$V_{1,3}$.
Since~$V_1$ is an independent set by Claim~\ref{clm:C4-V0V1V2V12-indep}, it follows that $G[V_1 \cup V_{1,3}]$ is a bipartite graph with parts~$V_1$ and~$V_{1,3}$.
Since~$G$ is $2P_2$-free, it follows that no two vertices in~$V_1$ can have different neighbours in~$V_{1,3}$.
Therefore at most one vertex of~$V_{1,3}$ has a neighbour in~$V_1$.
Now if $z \in V_{1,2,3,4}$, then~$z$ is adjacent to~$y$ by Claim~\ref{clm:C4-V13-V1234-comp}.
If~$x$ is non-adjacent to~$z$, then $G[v_1,y,v_2,x,z]$ is a~$\overline{P_2+P_3}$, a contradiction.
We conclude that~$V_1$ is complete to~$V_{1,2,3,4}$.
The claim follows by symmetry.
\end{inproof}

\begin{figure}[ht]
\centering
\includegraphics[scale=0.8,page=4]{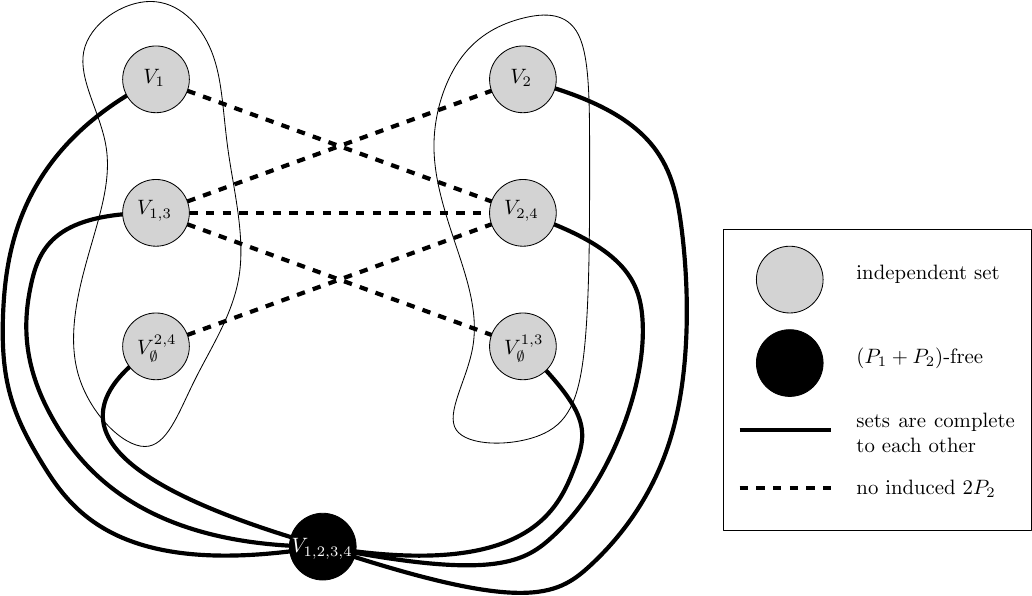}
\caption{The configuration in the proof of Lemma~\ref{l-c5free} after deleting the at most six vertices in $V_{1,2} \cup V_{2,3} \cup V_{3,4} \cup V_{1,4} \cup V(C)$ along with at most one vertex in each of~$V_{1,3}$ and~$V_{2,4}$.
The sets~$V_\emptyset^{2,4} \cup V_1 \cup V_{1,3}$ and~$V_\emptyset^{1,3} \cup V_2 \cup V_{2,4}$ are independent.}\label{fig:case2}
\end{figure}

\medskip
\noindent
We now proceed as follows.
By Claim~\ref{clm:C4-V123-empty}, the set $V_{1,2,3} \cup V_{2,3,4} \cup V_{1,3,4} \cup V_{1,2,4}$ is empty.
By Claims~\ref{clm:C4-V1V12-or-V3V23-empty} and~\ref{clm:C4-V12V34-small}, there are at most two vertices in $V_{1,2} \cup V_{2,3} \cup V_{3,4} \cup V_{1,4}$, so after doing at most two vertex deletions, we may assume these sets are empty (note that the resulting graph may no longer be an atom).
Applying four further vertex deletions, we can remove the cycle~$C$ from~$G$.
By Claim~\ref{clm:C4-V_1-at-most-1-V13-nbr}, at most one vertex of~$V_{1,3}$ (resp.~$V_{2,4}$) has a neighbour in~$V_1$ (resp.~$V_2$).
Therefore, applying at most two further vertex deletions, we may assume that~$V_{1,3}$ is anti-complete to~$V_1$ and~$V_{2,4}$ is anti-complete to~$V_2$.
By Claim~\ref{clm:C4-V1V12-or-V3V23-empty}, we may assume without loss of generality that~$V_3$ and~$V_4$ are empty (see \figurename~\ref{fig:case2} for an illustration of the resulting graph).

The remaining vertices of~$G$ all lie in $V_\emptyset \cup V_1 \cup V_2 \cup V_{1,3} \cup V_{2,4} \cup V_{1,2,3,4}$ and by Fact~\ref{fact:del-vert}, it suffices to show that this modified graph has bounded clique-width.
By Claims~\ref{clm:C4-V13-V1234-comp}, \ref{clm:C4-V0-non-empty-nice-nbhd} and~\ref{clm:C4-V_1-at-most-1-V13-nbr}, $V_{1,2,3,4}$ is complete to $V_\emptyset \cup V_1 \cup V_2 \cup V_{1,3} \cup V_{2,4}$, and so applying a bipartite complementation between these two sets disconnects $G[V_{1,2,3,4}]$ from the rest of the graph.
By Claim~\ref{clm:C4-V1234-bdd-cw}, $G[V_{1,2,3,4}]$ has bounded clique-width, so by Fact~\ref{fact:bip}, we may assume~$V_{1,2,3,4}$ is empty.

By Claim~\ref{clm:C4-V0-non-empty-nice-nbhd}, we can partition~$V_\emptyset$ into the set~$V_\emptyset^{1,3}$ of vertices that have neighbours in~$V_{1,3}$ and the set~$V_\emptyset^{2,4}$ of vertices that have neighbours in~$V_{2,4}$.
Now Claims~\ref{clm:C4-V0V1V2V12-indep} and~\ref{clm:C4-V12-V13-indep} imply that $V_\emptyset^{2,4} \cup V_1 \cup V_{1,3}$ and $V_\emptyset^{1,3} \cup V_2 \cup V_{2,4}$ are independent sets (recall that~$V_{1,3}$ is now anti-complete to~$V_1$ and~$V_{2,4}$ is now anti-complete to~$V_2$), and so $G[V_\emptyset \cup V_1 \cup V_2 \cup V_{1,3} \cup V_{2,4}]$ is a $2P_2$-free bipartite graph, so it is a bipartite chain graph by Lemma~\ref{l-e1} and thus has bounded clique-width by Lemma~\ref{l-e2}.
By Fact~\ref{fact:del-vert}, this completes the proof.
\end{proof}

\noindent
The class of split graphs is the class of $(C_4,C_5,2P_2)$-free graphs.
Since split graphs therefore form a subclass of the class of $(2P_2,\overline{P_2+P_3})$-free graphs, and split graphs have unbounded clique-width~\cite{MR99}, it follows that $(2P_2,\overline{P_2+P_3})$-free graphs also have unbounded clique-width.
Recall that by Lemma~\ref{l-e3}, split atoms are complete graphs and therefore have clique-width at most~$2$.
The $(2P_2,\overline{P_2+P_3})$-free atoms that are not split must therefore contain an induced~$C_4$ or~$C_5$.
Applying Lemmas~\ref{l-tripletwithc5} and \ref{l-c5free}, we obtain Theorem~\ref{thm:triplet}, which we restate below.

\medskip
\noindent
{\bfseries Theorem~\ref{thm:triplet} (restated).}
{\itshape The class of $(2P_2,\overline{P_2+P_3})$-free atoms has bounded clique-width (whereas the class of $(2P_2,\overline{P_2+P_3})$-free graphs has unbounded clique-width).}

\section{Clique-Width Summary for General Bigenic Classes}\label{s-soa}
In this section we present the state-of-art for boundedness of clique-width of general bigenic classes.
We will use these results in the next section, where we prove our results on unboundedness of clique-width of atoms in bigenic classes.

Let $H_1,H_2,H_3,H_4$ be four graphs.
Then the classes of $(H_1,H_2)$-free graphs and $(H_3,H_4)$-free graphs are said to be {\it equivalent} if the unordered pair $H_3,H_4$ can be obtained from the unordered pair $H_1,H_2$ by some combination of the operations: (i) complementing both graphs in the pair, and (ii) if one of the graphs in the pair is~$3P_1$, replacing it with $P_1+\nobreak P_3$ or vice versa.
If two classes are equivalent, then one of them has bounded clique-width if and only if the other one does~\cite{DP16}.

Recall that the subdivided claw~$S_{h,i,j}$, for $1\leq h\leq i\leq j$ is the tree with one vertex~$x$ of degree~$3$ and exactly three leaves, which are of distance~$h$,~$i$ and~$j$ from~$x$, respectively.
Also recall that~${\cal S}$ denotes the class of graphs every connected component of which is either a subdivided claw or a path.
Moreover, recall that the $\paw$ is the graph~$\overline{P_1+P_3}$, the $\diamondgraph$ is the graph~$\overline{2P_1+P_2}$ and the $\gem$ is the graph~$\overline{P_1+P_4}$.

\begin{theorem}[\cite{DJP19}]\label{thm:classification2}
Let~${\cal G}$ be a class of graphs defined by two forbidden induced subgraphs.
Then:
\begin{enumerate}
\item \label{thm:classification2:known-bdd} ${\cal G}$ has bounded clique-width if it is equivalent to a class of $(H_1,H_2)$-free graphs such that one of the following holds:
\begin{enumerate}[(i)]
\renewcommand{\theenumii}{(\roman{enumii})}
\renewcommand{\labelenumii}{(\roman{enumii})}
\item \label{thm:classification2:bdd:P4} $H_1$ or $H_2 \ssi P_4$
\item \label{thm:classification2:bdd:ramsey} $H_1=K_s$ and $H_2=tP_1$ for some $s,t\geq 1$
\item \label{thm:classification2:bdd:P_1+P_3} $H_1 \ssi \paw$ and $H_2 \ssi K_{1,3}+\nobreak 3P_1,\; K_{1,3}+\nobreak P_2,\;\allowbreak P_1+\nobreak P_2+\nobreak P_3,\;\allowbreak P_1+\nobreak P_5,\;\allowbreak P_1+\nobreak S_{1,1,2},\;\allowbreak P_2+\nobreak P_4,\;\allowbreak P_6,\; \allowbreak S_{1,1,3}$ or~$S_{1,2,2}$
\item \label{thm:classification2:bdd:2P_1+P_2} $H_1 \ssi \diamondgraph$ and $H_2\ssi P_1+\nobreak 2P_2,\; 3P_1+\nobreak P_2$ or~$P_2+\nobreak P_3$
\item \label{thm:classification2:bdd:P_1+P_4} $H_1 \ssi \gem$ and $H_2 \ssi P_1+\nobreak P_4$ or~$P_5$
\item \label{thm:classification2:bdd:K_13} $H_1\ssi K_3+\nobreak P_1$ and $H_2 \ssi K_{1,3}$, or
\item \label{thm:classification2:bdd:2P1_P3} $H_1\ssi \overline{2P_1+\nobreak P_3}$ and $H_2\ssi 2P_1+\nobreak P_3$.\\
\end{enumerate}
\item \label{thm:classification2:known-unbdd} ${\cal G}$ has unbounded clique-width if it is equivalent to a class of $(H_1,H_2)$-free graphs such that one of the following holds:
\begin{enumerate}[(i)]
\renewcommand{\theenumii}{(\roman{enumii})}
\renewcommand{\labelenumii}{(\roman{enumii})}
\item \label{thm:classification2:unbdd:not-in-S} $H_1\not\in {\cal S}$ and $H_2 \not \in {\cal S}$
\item \label{thm:classification2:unbdd:not-in-co-S} $H_1\notin \overline{{\cal S}}$ and $H_2 \not \in \overline{{\cal S}}$
\item \label{thm:classification2:unbdd:K_13or2P_2} $H_1 \si K_3+\nobreak P_1$ or~$C_4$ and $H_2 \si 4P_1$ or~$2P_2$
\item \label{thm:classification2:unbdd:2P_1+P_2} $H_1 \si \diamondgraph$ and $H_2 \si K_{1,3},\; 5P_1,\; P_2+\nobreak P_4$ or~$P_6$
\item \label{thm:classification2:unbdd:3P_1} $H_1 \si K_3$ and $H_2 \si 2P_1+\nobreak 2P_2,\; 2P_1+\nobreak P_4,\; 4P_1+\nobreak P_2,\; 3P_2$ or~$2P_3$
\item \label{thm:classification2:unbdd:4P_1} $H_1 \si K_4$ and $H_2 \si P_1 +\nobreak P_4$ or~$3P_1+\nobreak P_2$, or
\item \label{thm:classification2:unbdd:gem} $H_1 \si \gem$ and $H_2 \si P_1+\nobreak 2P_2$.
\end{enumerate}
\end{enumerate}
\end{theorem}

\noindent
As mentioned in Section~\ref{s-intro}, Theorem~\ref{thm:classification2} does not cover five (non-equivalent) cases (see also Open Problem~\ref{o-atoms}, where these open cases are marked with a~$^*$).

\begin{oproblem}\label{oprob:twographs}
Does the class of $(H_1,H_2)$-free graphs have bounded or unbounded clique-width when:
\begin{enumerate}[(i)]
\renewcommand{\theenumi}{(\roman{enumi})}
\renewcommand{\labelenumi}{(\roman{enumi})}
\item\label{oprob:twographs:3P_1}$H_1=K_3$ and $H_2 \in \{P_1+\nobreak S_{1,1,3},\allowbreak S_{1,2,3}\}$
\item\label{oprob:twographs:2P_1+P_2}$H_1=\diamondgraph$ and $H_2 \in \{P_1+\nobreak P_2+\nobreak P_3,\allowbreak P_1+\nobreak P_5\}$
\item\label{oprob:twographs:P_1+P_4}$H_1=\gem$ and $H_2=P_2+\nobreak P_3$.
\end{enumerate}
\end{oproblem}

\section{Atoms of Unbounded Clique-Width}\label{s-unbounded}
In this section we show our results for pairs $(H_1,H_2)$, for which the class of $(H_1,H_2)$-free atoms has unbounded clique-width.
We start by giving a number of known and new lemmas, some of which have wider applicability. 

\begin{lemma}[\cite{DP16}]\label{lem:generalunbounded}
For $m\geq 0$ and $n >\nobreak m+\nobreak 1$ the clique-width of a graph~$G$ is at least $\lfloor\frac{n-1}{m+1}\rfloor+\nobreak 1$ if~$V(G)$ has a partition into sets $V_{i,j} \; (i,j \in \{0,\ldots,n\})$ with the following properties:
\begin{enumerate}
\item \label{prop:v_i0-small}$|V_{i,0}| \leq 1$ for all $i\in\{1,\ldots,n\}$
\item \label{prop:v_0j-small}$|V_{0,j}| \leq 1$ for~all~$j\in\{1,\ldots,n\}$
\item \label{prop:v_ij-nonempty}$|V_{i,j}|\geq 1$ for all $i,j\in\{1,\ldots,n\}$
\item \label{prop:row-connected}$G[\cup^n_{j=0}V_{i,j}]$ is connected for all $i\in\{1,\ldots,n\}$
\item \label{prop:column-connected}$G[\cup^n_{i=0}V_{i,j}]$ is connected for all $j\in\{1,\ldots,n\}$
\item \label{prop:v_k0-nbrs}for $i,j,k\in\{1,\ldots,n\}$, if a vertex of~$V_{k,0}$ is adjacent to a vertex of~$V_{i,j}$ then $i \leq k$
\item \label{prop:v_0k-nbrs}for $i,j,k\in\{1,\ldots,n\}$, if a vertex of~$V_{0,k}$ is adjacent to a vertex of~$V_{i,j}$ then $j \leq k$, and
\item \label{prop:v_ij-nbrs}for $i,j,k,\ell\in\{1,\ldots,n\}$, if a vertex of~$V_{i,j}$ is adjacent to a vertex of~$V_{k,\ell}$ then $|k-i|\leq m$ and $|\ell-j| \leq m$.
\end{enumerate}
\end{lemma}

The next lemma concerns walls.
We do not formally define the wall, but instead we refer to \figurename~\ref{fig:walls}, in which three examples of walls of different heights are depicted; see, for example,~\cite{Ch15} for a formal definition.

\begin{figure}[h]
\begin{center}
\begin{minipage}{0.2\textwidth}
\centering
\begin{tikzpicture}[scale=0.45, every node/.style={scale=0.4}]
\GraphInit[vstyle=Simple]
\SetVertexSimple[MinSize=6pt]
\Vertex[x=1,y=0]{v10}
\Vertex[x=2,y=0]{v20}
\Vertex[x=3,y=0]{v30}
\Vertex[x=4,y=0]{v40}
\Vertex[x=5,y=0]{v50}

\Vertex[x=0,y=1]{v01}
\Vertex[x=1,y=1]{v11}
\Vertex[x=2,y=1]{v21}
\Vertex[x=3,y=1]{v31}
\Vertex[x=4,y=1]{v41}
\Vertex[x=5,y=1]{v51}

\Vertex[x=0,y=2]{v02}
\Vertex[x=1,y=2]{v12}
\Vertex[x=2,y=2]{v22}
\Vertex[x=3,y=2]{v32}
\Vertex[x=4,y=2]{v42}

\Edges(    v10,v20,v30,v40,v50)
\Edges(v01,v11,v21,v31,v41,v51)
\Edges(v02,v12,v22,v32,v42)

\Edge(v01)(v02)

\Edge(v10)(v11)

\Edge(v21)(v22)

\Edge(v30)(v31)

\Edge(v41)(v42)

\Edge(v50)(v51)

\end{tikzpicture}
\end{minipage}
\begin{minipage}{0.3\textwidth}
\centering
\begin{tikzpicture}[scale=0.45, every node/.style={scale=0.4},rotate=180]
\GraphInit[vstyle=Simple]
\SetVertexSimple[MinSize=6pt]
\Vertex[x=1,y=0]{v10}
\Vertex[x=2,y=0]{v20}
\Vertex[x=3,y=0]{v30}
\Vertex[x=4,y=0]{v40}
\Vertex[x=5,y=0]{v50}
\Vertex[x=6,y=0]{v60}
\Vertex[x=7,y=0]{v70}

\Vertex[x=0,y=1]{v01}
\Vertex[x=1,y=1]{v11}
\Vertex[x=2,y=1]{v21}
\Vertex[x=3,y=1]{v31}
\Vertex[x=4,y=1]{v41}
\Vertex[x=5,y=1]{v51}
\Vertex[x=6,y=1]{v61}
\Vertex[x=7,y=1]{v71}

\Vertex[x=0,y=2]{v02}
\Vertex[x=1,y=2]{v12}
\Vertex[x=2,y=2]{v22}
\Vertex[x=3,y=2]{v32}
\Vertex[x=4,y=2]{v42}
\Vertex[x=5,y=2]{v52}
\Vertex[x=6,y=2]{v62}
\Vertex[x=7,y=2]{v72}

\Vertex[x=1,y=3]{v13}
\Vertex[x=2,y=3]{v23}
\Vertex[x=3,y=3]{v33}
\Vertex[x=4,y=3]{v43}
\Vertex[x=5,y=3]{v53}
\Vertex[x=6,y=3]{v63}
\Vertex[x=7,y=3]{v73}

\Edges(    v10,v20,v30,v40,v50,v60,v70)
\Edges(v01,v11,v21,v31,v41,v51,v61,v71)
\Edges(v02,v12,v22,v32,v42,v52,v62,v72)
\Edges(    v13,v23,v33,v43,v53,v63,v73)

\Edge(v01)(v02)

\Edge(v10)(v11)
\Edge(v12)(v13)

\Edge(v21)(v22)

\Edge(v30)(v31)
\Edge(v32)(v33)

\Edge(v41)(v42)

\Edge(v50)(v51)
\Edge(v52)(v53)

\Edge(v61)(v62)

\Edge(v70)(v71)
\Edge(v72)(v73)
\end{tikzpicture}
\end{minipage}
\begin{minipage}{0.35\textwidth}
\centering
\begin{tikzpicture}[scale=0.45, every node/.style={scale=0.4}]
\GraphInit[vstyle=Simple]
\SetVertexSimple[MinSize=6pt]
\Vertex[x=1,y=0]{v10}
\Vertex[x=2,y=0]{v20}
\Vertex[x=3,y=0]{v30}
\Vertex[x=4,y=0]{v40}
\Vertex[x=5,y=0]{v50}
\Vertex[x=6,y=0]{v60}
\Vertex[x=7,y=0]{v70}
\Vertex[x=8,y=0]{v80}
\Vertex[x=9,y=0]{v90}

\Vertex[x=0,y=1]{v01}
\Vertex[x=1,y=1]{v11}
\Vertex[x=2,y=1]{v21}
\Vertex[x=3,y=1]{v31}
\Vertex[x=4,y=1]{v41}
\Vertex[x=5,y=1]{v51}
\Vertex[x=6,y=1]{v61}
\Vertex[x=7,y=1]{v71}
\Vertex[x=8,y=1]{v81}
\Vertex[x=9,y=1]{v91}

\Vertex[x=0,y=2]{v02}
\Vertex[x=1,y=2]{v12}
\Vertex[x=2,y=2]{v22}
\Vertex[x=3,y=2]{v32}
\Vertex[x=4,y=2]{v42}
\Vertex[x=5,y=2]{v52}
\Vertex[x=6,y=2]{v62}
\Vertex[x=7,y=2]{v72}
\Vertex[x=8,y=2]{v82}
\Vertex[x=9,y=2]{v92}

\Vertex[x=0,y=3]{v03}
\Vertex[x=1,y=3]{v13}
\Vertex[x=2,y=3]{v23}
\Vertex[x=3,y=3]{v33}
\Vertex[x=4,y=3]{v43}
\Vertex[x=5,y=3]{v53}
\Vertex[x=6,y=3]{v63}
\Vertex[x=7,y=3]{v73}
\Vertex[x=8,y=3]{v83}
\Vertex[x=9,y=3]{v93}

\Vertex[x=0,y=4]{v04}
\Vertex[x=1,y=4]{v14}
\Vertex[x=2,y=4]{v24}
\Vertex[x=3,y=4]{v34}
\Vertex[x=4,y=4]{v44}
\Vertex[x=5,y=4]{v54}
\Vertex[x=6,y=4]{v64}
\Vertex[x=7,y=4]{v74}
\Vertex[x=8,y=4]{v84}

\Edges(    v10,v20,v30,v40,v50,v60,v70,v80,v90)
\Edges(v01,v11,v21,v31,v41,v51,v61,v71,v81,v91)
\Edges(v02,v12,v22,v32,v42,v52,v62,v72,v82,v92)
\Edges(v03,v13,v23,v33,v43,v53,v63,v73,v83,v93)
\Edges(v04,v14,v24,v34,v44,v54,v64,v74,v84)

\Edge(v01)(v02)
\Edge(v03)(v04)

\Edge(v10)(v11)
\Edge(v12)(v13)

\Edge(v21)(v22)
\Edge(v23)(v24)

\Edge(v30)(v31)
\Edge(v32)(v33)

\Edge(v41)(v42)
\Edge(v43)(v44)

\Edge(v50)(v51)
\Edge(v52)(v53)

\Edge(v61)(v62)
\Edge(v63)(v64)

\Edge(v70)(v71)
\Edge(v72)(v73)

\Edge(v81)(v82)
\Edge(v83)(v84)

\Edge(v90)(v91)
\Edge(v92)(v93)
\end{tikzpicture}
\end{minipage}
\caption{Walls of height $2$, $3$ and~$4$, respectively.}\label{fig:walls}
\end{center}
\end{figure}
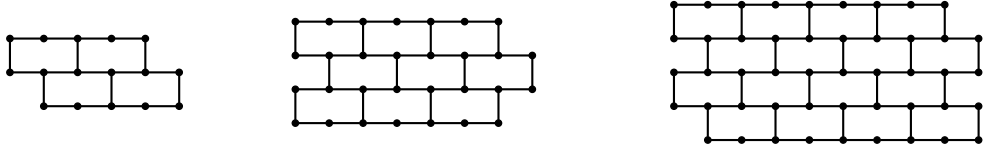

\begin{lemma}[\cite{LR06}]\label{lem:walls}
For any constant $k\geq 0$, the class of $k$-subdivided walls has unbounded clique-width.
\end{lemma}

\begin{lemma}\label{lem:need-S-co-S}
Let $H_1,H_2$ be graphs.
If $H_1,H_2 \notin {\cal S}$ or $\overline{H_1},\overline{H_2} \notin {\cal S}$ then the class of $(H_1,H_2)$-free atoms has unbounded clique-width.
\end{lemma}

\begin{proof}
Let $k=\max(|V(H_1)|,|V(H_2)|)$.
Let~$H$ be a $k$-subdivided wall of height at least~$2$ (see \figurename~\ref{fig:walls}) and note that the clique-width of such graphs is unbounded by Lemma~\ref{lem:walls}.
By Fact~\ref{fact:comp}, it follows that the clique-width of graphs of the form~$\overline{H}$ is also unbounded.

We claim that if $H_1,H_2 \notin {\cal S}$ then~$H$ is $(H_1,H_2)$-free.
Let $i \in \{1,2\}$.
It is easy to verify that if~$H_i$ contains a cycle, then~$H$ is $H_i$-free (due to the choice of~$k$).
Similarly, if~$H_i$ contains an induced tree with two vertices of degree at least~$3$ or a vertex of degree at least~$4$, then~$H$ is $H_i$-free.
Therefore, if~$H_i$ is an induced subgraph of~$H$, then~$H_i$ is a forest and every component of~$H_i$ must be a tree in which at most one vertex has degree~$3$ and all other vertices have degree at most~$2$.
In other words, if~$H_i$ is an induced subgraph of~$H$, then $H_i \in {\cal S}$.
We conclude that if $H_1,H_2 \notin {\cal S}$, then~$H$ is $(H_1,H_2)$-free.
This also implies that if $\overline{H_1},\overline{H_2} \notin {\cal S}$, then~$\overline{H}$ is $(H_1,H_2)$-free.

It remains to show that~$H$ and~$\overline{H}$ are atoms.
Indeed, $H$ is a bipartite graph, so every clique cut-set consists of at most two vertices; it is easy to verify that there is no vertex whose removal disconnects~$H$ and no edge such that removing both of its end-vertices disconnects~$H$.
Therefore~$H$ is indeed an atom.

Now, $\overline{H}$ is a co-bipartite graph, so it can be partitioned into two cliques~$A$ and~$B$.
Note that $|A|,|B|>12$ by construction.
Suppose, for contradiction, that~$X$ is a clique cut-set in~$\overline{H}$.
Let $Y=V(\overline{H}) \setminus X$ and note that~$\overline{H}[Y]$ is disconnected, so it contains two vertices~$a,b$ that are non-adjacent.
Since~$A$ is a clique and~$B$ is a clique, we may assume $a \in A$ and $b \in B$.
Now~$Y$ cannot contain vertices $a' \in A$, $b' \in B$ that are adjacent in~$\overline{H}$, as in that case~$\{a',b'\}$ would dominate~$\overline{H}$, contradicting the assumption that~$\overline{H}[Y]$ is disconnected.
In~$H$ every vertex has either two or three neighbours, so in~$\overline{H}$ every vertex has either two or three non-neighbours.
Since $a \in A \cap Y$, there can be at most three vertices in $B \cap Y$ and similarly, there can be at most three vertices in $A \cap Y$.
Since every vertex in~$B \cap Y$ has at most three non-neighbours in~$A$, it follows that at most nine vertices of~$A$ have non-neighbours in~$B \cap Y$.
Since $|A|>12 \geq 9 + |A \cap Y|$, there must be a vertex in $a' \in A \setminus Y = A \cap X$ that has no non-neighbours in~$B \cap Y$ and therefore has a  non-neighbour $b' \in B \setminus Y = B \cup X$.
This contradicts the fact that~$X$ is a clique in~$\overline{H}$.
Therefore~$\overline{H}$ is indeed an atom.
\end{proof}

\begin{lemma}\label{lem:no-false-twin}
Let~${\cal H}$ be a set of graphs such that no graph in~${\cal H}$ contains a pair of vertices that are false twins.
Then the class of ${\cal H}$-free atoms has bounded clique-width if and only if the class of ${\cal H}$-free graphs does.
\end{lemma}
\begin{proof}
Clearly, if the class of ${\cal H}$-free graphs has bounded clique-width, then the class of ${\cal H}$-free atoms does.
Now suppose that the class of ${\cal H}$-free graphs has unbounded clique-width.
Let~${\cal F}$ be the class of connected ${\cal H}$-free graphs on at least two vertices.
Since the clique-width of a graph is equal to the maximum of the clique-widths of its components, it follows that~${\cal F}$ has unbounded clique-width.
For every graph $F \in {\cal F}$, we construct the graph~$F'$, which has vertex set $V(F')=\{v,v' \;|\; v \in V(F)\}$ and edge set $E(F') = \{uv,uv',u'v,u'v' \;|\; uv \in E(F)\}$.
So, for every $v\in V(F)$ we have introduced a new vertex~$v'$ such that~$v$ and~$v'$ are false twins in~$F'$.
Let~${\cal F'}$ be the set of such graphs~$F'$.
Since for every $F \in {\cal F}$, the graph~$F'$ contains~$F$ as an induced subgraph, it follows that~${\cal F'}$ has unbounded clique-width.

We claim that every graph in~${\cal F'}$ is an atom.
Indeed, suppose, for contradiction, that~$X$ is a clique cut-set of a graph $F' \in {\cal F'}$.
Since for every $v \in V(F)$, $v$ is non-adjacent to~$v'$ in~$F'$, it follows that at most one of~$v$ and~$v'$ is in~$X$.
Since~$v$ and~$v'$ are false twins in~$F'$ we may replace all vertices $v \in X \cap V(F)$ by their false twins~$v'$ and the resulting set~$X'$ will still be a clique cut-set.
By construction, the graph~$F$ is connected and every vertex in $V(F')\setminus X'$ has a neighbour in~$V(F)$ in the graph $F'\setminus X'$.
Therefore $F' \setminus X'$ is connected, a contradiction.
It follows that every graph in~${\cal F'}$ is indeed an atom.

It remains to show that the graphs in~${\cal F'}$ are ${\cal H}$-free.
Indeed, suppose, for contradiction, that $H \in {\cal H}$ is an induced subgraph of $F' \in {\cal F'}$.
Since for every $v \in V(F)$, the vertices~$v$ and~$v'$ are false twins in~$F'$, and~$H$ does not have a pair of false twins, it follows that at most one of~$v$ and~$v'$ is in the induced copy of~$H$ found in~$F'$.
Furthermore, if~$v'$ is in this induced copy, then we can replace it by~$v$.
Thus we find that there is an induced copy of~$H$ in~$F'$ all of whose vertices lie in~$V(F)$.
Therefore~$H$ is an induced subgraph of~$F$.
This is a contradiction as $F \in {\cal F}$ and the graphs in~${\cal F}$ are ${\cal H}$-free.
We have therefore shown that the graphs in~${\cal F'}$ are ${\cal H}$-free atoms and that~${\cal F'}$ has unbounded clique-width.
This completes the proof.
\end{proof}

Observe that the condition in the following lemma holds if and only if for every graph $H \in {\cal H}$, the graph~$\overline{H}$ does not have a component isomorphic to~$P_1$ or~$P_2$.

\begin{lemma}\label{lem:no-comp-P1or2P1}
Let~${\cal H}$ be a set of graphs such that no graph in~${\cal H}$ contains a dominating vertex and no graph in~${\cal H}$ contains a pair of non-adjacent vertices that are complete to the remainder of the graph.
Then the class of ${\cal H}$-free atoms has bounded clique-width if and only if the class of ${\cal H}$-free graphs does.
\end{lemma}
\begin{proof}
Clearly, if the class of ${\cal H}$-free graphs has bounded clique-width, then the class of ${\cal H}$-free atoms does.
Now suppose that the class of ${\cal H}$-free graphs has unbounded clique-width.
Let~${\cal F}$ be the class of ${\cal H}$-free graphs that contain at least one non-edge.
Since complete graphs have clique-width at most~$2$ and the class of ${\cal H}$-free graphs has unbounded clique-width, it follows that~${\cal F}$ has unbounded clique-width.
For every graph $F \in {\cal F}$, we construct the graph~$F'$ by adding two new vertices $x,x'$ and adding edges to make~$\{x,x'\}$ complete to the remainder of the graph (note that~$x$ is non-adjacent to~$x'$ in~$F'$).
Let~${\cal F'}$ be the set of such graphs~$F'$.
Since for every $F \in {\cal F}$, the graph~$F'$ contains~$F$ as an induced subgraph, it follows that~${\cal F'}$ has unbounded clique-width.

We claim that every graph in~${\cal F'}$ is an atom.
Suppose, for contradiction, that $F' \in {\cal F'}$ has a clique cut-set~$X$.
Since~$x$ and~$x'$ are non-adjacent, it follows that either~$x$ or~$x'$ are not in~$X$; since~$x$ and~$x'$ are false twins, we may assume $x \notin X$.
Since~$F$ is not a complete graph, there must be a vertex $y \in V(F) \setminus X$.
Since~$x$ is complete to $V(F) \setminus X$ in~$F' \setminus X$, every vertex of~$V(F) \setminus X$ is in the same component of $F' \setminus X$ as~$x$.
Since~$y$ is complete to $\{x,x'\} \setminus X$ in~$F' \setminus X$, every vertex of $\{x,x'\} \setminus X$ is in the same component of $F' \setminus X$ as~$y$.
Therefore~$F' \setminus X$ is connected, a contradiction.
It follows that every graph in~${\cal F'}$ is indeed an atom.

It remains to show that the graphs in~${\cal F'}$ are ${\cal H}$-free.
Indeed, suppose, for contradiction, that $H \in {\cal H}$ is an induced subgraph of $F' \in {\cal F'}$.
Since~$H$ does not contain a pair of non-adjacent vertices that are complete to the rest of the graph, this induced copy of~$H$ in~$F'$ cannot contain both~$x$ and~$x'$.
Since~$H$ does not have a dominating vertex, the induced copy of~$H$ in~$F'$ cannot contain exactly one of~$x$ and~$x'$.
Therefore the induced copy of~$H$ in~$F'$ must consist of only vertices in~$V(F)$.
Therefore~$H$ is an induced subgraph of~$F$.
This is a contradiction as $F \in {\cal F}$ and the graphs in~${\cal F}$ are ${\cal H}$-free.
We have therefore shown that the graphs in~${\cal F'}$ are ${\cal H}$-free atoms and that~${\cal F'}$ has unbounded clique-width.
This completes the proof.
\end{proof}

\begin{figure}[h]
\begin{center}
\begin{tabular}{cccc}
\scalebox{0.6}{
\begin{minipage}{0.3\textwidth}
\begin{center}%C_4
\begin{tikzpicture}[every node/.style={circle,fill, minimum size=0.07cm}]
\node at (0,0) {};
\node at (0,2) {};
\node at (2,0) {};
\node at (2,2) {};
\draw (0,0) -- (0,2) -- (2,2) -- (2,0) -- (0,0);
\end{tikzpicture}
\end{center}
\end{minipage}}
&
\scalebox{0.6}{
\begin{minipage}{0.3\textwidth}
\begin{center}%K_{1,3}
\begin{tikzpicture}[every node/.style={circle,fill, minimum size=0.07cm}]
\node at (0,0) {};
\node at (1,0) {};
\node at (2,0) {};
\node at (1,2) {};
\draw (1,2) -- (0,0);
\draw (1,2) -- (1,0);
\draw (1,2) -- (2,0);
\end{tikzpicture}
\end{center}
\end{minipage}}
&
\scalebox{0.6}{
\begin{minipage}{0.3\textwidth}
\begin{center}%K_4
\begin{tikzpicture}[every node/.style={circle,fill, minimum size=0.07cm}]
\node at (0,0) {};
\node at (0,2) {};
\node at (2,0) {};
\node at (2,2) {};
\draw (0,0) -- (0,2) -- (2,2) -- (2,0) -- (0,0) -- (2,2);
\draw (0,2) -- (2,0);
\end{tikzpicture}
\end{center}
\end{minipage}}
&
\scalebox{0.6}{
\begin{minipage}{0.4\textwidth}
\begin{center}%diamond
\begin{tikzpicture}[every node/.style={circle,fill, minimum size=0.07cm}]
\node at (0,0) {};
\node at (0,2) {};
\node at (2,0) {};
\node at (2,2) {};
\draw (0,0) -- (0,2) -- (2,2) -- (2,0) -- (0,0) -- (2,2);
\end{tikzpicture}
\end{center}
\end{minipage}}
\\
\\
$C_4$ & $K_{1,3}$ & $K_4$ & $\overline{2P_1+P_2}$\\
\\
\scalebox{0.6}{
\begin{minipage}{0.3\textwidth}
\begin{center}%2P_2
\begin{tikzpicture}[every node/.style={circle,fill, minimum size=0.07cm}]
\node at (0,0) {};
\node at (0,2) {};
\node at (2,0) {};
\node at (2,2) {};
\draw (0,0) -- (0,2);
\draw (2,0) -- (2,2);
\end{tikzpicture}
\end{center}
\end{minipage}}
&
\scalebox{0.6}{
\begin{minipage}{0.3\textwidth}
\begin{center}%K_3+P_1
\begin{tikzpicture}[every node/.style={circle,fill, minimum size=0.07cm}]
\node at (0,0) {};
\node at (1,0.57735026919) {};%1/sqrt(3)
\node at (1,1.73205080757) {};%sqrt(3)
\node at (2,0) {};
\draw (0,0) -- (1,1.73205080757) -- (2,0) -- (0,0);
\end{tikzpicture}
\end{center}
\end{minipage}}
&
\scalebox{0.6}{
\begin{minipage}{0.3\textwidth}
\begin{center}%4P_1
\begin{tikzpicture}[every node/.style={circle,fill, minimum size=0.07cm}]
\node at (0,0) {};
\node at (0,2) {};
\node at (2,0) {};
\node at (2,2) {};
\end{tikzpicture}
\end{center}
\end{minipage}}
&
\scalebox{0.6}{
\begin{minipage}{0.4\textwidth}
\begin{center}%2P_1+P_2
\begin{tikzpicture}[every node/.style={circle,fill, minimum size=0.07cm}]
\node at (0,0) {};
\node at (0,2) {};
\node at (2,0) {};
\node at (2,2) {};
\draw (2,0) -- (0,0);
\end{tikzpicture}
\end{center}
\end{minipage}}
\\
\\
$2P_2$ & $K_3+\nobreak P_1$ & $4P_1$ & $2P_1+\nobreak P_2$
\end{tabular}
\end{center}
\caption{\label{fig:C4-K13-K4-diamond}The forbidden induced subgraphs for the classes of $(C_4,K_{1,3},K_4,\overline{2P_1+P_2})$-free graphs and $(2P_2,K_3+\nobreak P_1,\allowbreak 4P_1,\allowbreak 2P_1+\nobreak P_2)$-free graphs mentioned in Lemma~\ref{lem:C4-K13-K4-diamond}.}
\end{figure}
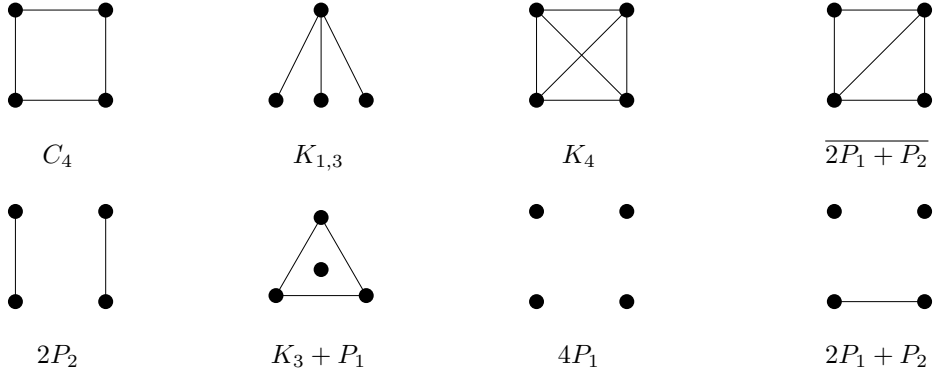

\begin{lemma}\label{lem:C4-K13-K4-diamond}
The class of $(C_4,K_{1,3},K_4,\overline{2P_1+P_2})$-free atoms and the class of $(2P_2,K_3+\nobreak P_1,\allowbreak 4P_1,\allowbreak 2P_1+\nobreak P_2)$-free atoms have unbounded clique-width (see \figurename~\ref{fig:C4-K13-K4-diamond} for illustrations of the forbidden induced subgraphs).
\end{lemma}

\begin{proof}
Brandst\"adt et al.~\cite[Theorem~10(ii)]{BELL06} constructed a family of graphs~$H_n$ that have unbounded clique-width and are $(C_4,K_{1,3},K_4,\overline{2P_1+P_2})$-free.
The graph~$H_n$ is constructed from the $1$-subdivided $n\times n$ grid by adding new edges incident to the vertices added by the subdivision as follows: in each cell of the subdivided grid, the left vertex added by the subdivision is made adjacent to the top one, and the bottom vertex added by the subdivision is made adjacent to the right one (see also \figurename~\ref{fig:def:C4-K13-K4-diamond} or see~\cite[Section~6.2]{BELL06} for a formal definition).
However, the graph~$H_n$ has clique cut-sets, so it is not an atom.
On the other hand, since the class of $(C_4,K_{1,3},K_4,\overline{2P_1+P_2})$-free graphs has unbounded clique-width, Fact~\ref{fact:comp} implies that the class of $(2P_2,K_3+\nobreak P_1,\allowbreak 4P_1,\allowbreak 2P_1+\nobreak P_2)$-free graphs has unbounded clique-width.
We observe that every graph in $\{2P_2,K_3+\nobreak P_1,\allowbreak 4P_1,\allowbreak 2P_1+\nobreak P_2\}$ has no dominating vertex and no two non-adjacent vertices that are complete to the remainder of the graph.
Therefore, by Lemma~\ref{lem:no-comp-P1or2P1}, the class of $(2P_2,K_3+\nobreak P_1,\allowbreak 4P_1,\allowbreak 2P_1+\nobreak P_2)$-free atoms has unbounded clique-width.

\begin{figure}[h]
\begin{center}
\includegraphics[scale=0.7, page=1]{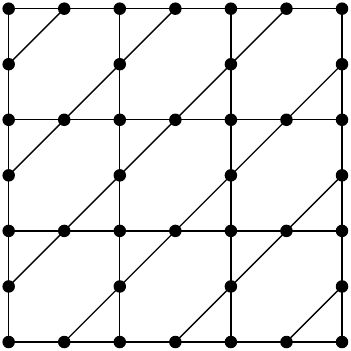}
\caption{The graph~$H_n$ from the proof of Lemma~\ref{lem:C4-K13-K4-diamond} ($n=4$ shown).}
\label{fig:def:C4-K13-K4-diamond}
\end{center}
\end{figure}

We now prove that the class of $(C_4,K_{1,3},K_4,\overline{2P_1+P_2})$-free atoms has unbounded clique-width.
Consider a wall of height~$k \geq 2$ and let~$J_k$ be its line graph.
It is easy to verify that for every~$k$, the graph~$H_n$ contains~$J_k$ as an induced subgraph if~$n$ is sufficiently large.
Similarly, for every~$n$, the graph~$J_k$ contains~$H_n$ as an induced subgraph if~$k$ is sufficiently large.
Therefore, by~\cite[Theorem~10(ii)]{BELL06}, the graph~$J_k$ is $(C_4,K_{1,3},K_4,\overline{2P_1+P_2})$-free and this family of graphs has unbounded clique-width (the former can also be seen by inspection and latter can also be seen by using Lemma~\ref{lem:generalunbounded}).
Every clique in~$J_k$ contains at most three vertices and it is easy to verify that~$J_k$ does not contain a clique cut-set on at most three vertices, so~$J_k$ is an atom.
This completes the proof.
\end{proof}

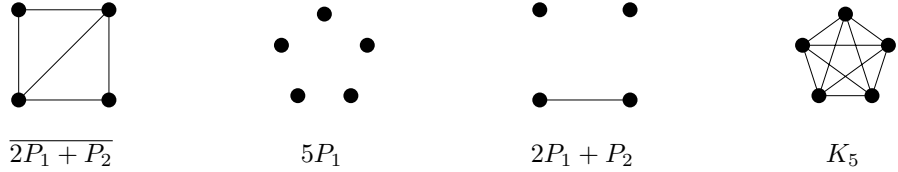
\begin{figure}[h]
\begin{center}
\begin{tabular}{cccc}
\scalebox{0.6}{
\begin{minipage}{0.3\textwidth}
\begin{center}%diamond
\begin{tikzpicture}[every node/.style={circle,fill, minimum size=0.07cm}]
\node at (0,0) {};
\node at (0,2) {};
\node at (2,0) {};
\node at (2,2) {};
\draw (0,0) -- (0,2) -- (2,2) -- (2,0) -- (0,0) -- (2,2);
\end{tikzpicture}
\end{center}
\end{minipage}}
&
\scalebox{0.6}{
\begin{minipage}{0.3\textwidth}
\begin{center}%5P_1
\begin{tikzpicture}[every node/.style={circle,fill, minimum size=0.07cm}]
\node at (18:1) {};
\node at (90:1) {};
\node at (162:1) {};
\node at (234:1) {};
\node at (306:1) {};
\end{tikzpicture}
\end{center}
\end{minipage}}
&
\scalebox{0.6}{
\begin{minipage}{0.3\textwidth}
\begin{center}%2P_1+P_2
\begin{tikzpicture}[every node/.style={circle,fill, minimum size=0.07cm}]
\node at (0,0) {};
\node at (0,2) {};
\node at (2,0) {};
\node at (2,2) {};
\draw (2,0) -- (0,0);
\end{tikzpicture}
\end{center}
\end{minipage}}
&
\scalebox{0.6}{
\begin{minipage}{0.3\textwidth}
\begin{center}%K_5
\begin{tikzpicture}[every node/.style={circle,fill, minimum size=0.07cm}]
\node at (18:1) {};
\node at (90:1) {};
\node at (162:1) {};
\node at (234:1) {};
\node at (306:1) {};
\draw (18:1) -- (90:1) -- (162:1) -- (234:1) -- (306:1) -- (18:1);
\draw (18:1) -- (162:1) -- (306:1) -- (90:1) -- (234:1) -- (18:1);
\end{tikzpicture}
\end{center}
\end{minipage}}
\\
\\
$\overline{2P_1+P_2}$ & $5P_1$ & $2P_1+\nobreak P_2$ & $K_5$\\
\end{tabular}
\end{center}
\caption{\label{fig:diamond-5P1}The forbidden induced subgraphs for the classes of $(\overline{2P_1+P_2},5P_1)$-free graphs and $(2P_1+\nobreak P_2,\allowbreak K_5)$-free graphs mentioned in Lemma~\ref{lem:diamond-5P1}.}
\end{figure}

\begin{lemma}\label{lem:diamond-5P1}
The class of $(\overline{2P_1+P_2},5P_1)$-free atoms and the class of $(2P_1+\nobreak P_2,\allowbreak K_5)$-free atoms has unbounded clique-width (see \figurename~\ref{fig:diamond-5P1} for illustrations of the forbidden induced subgraphs).
\end{lemma}

\begin{proof}
We use the construction from~\cite{DGP14}, which was used to show that $(\overline{2P_1+P_2},\allowbreak 5P_1)$-free graphs have unbounded clique-width.
Consider a wall of height~$2n+1$ for some $n \geq 2$.
Colour the vertices on the top row with colours $1,2,3,4,1,2,3,4,\ldots$ and on the next row with colours $3,4,1,2,3,4,1,2,\ldots$, then alternate these colourings on the following rows, so that no vertex has two neighbours that have the same colour (see also \figurename~\ref{fig:diamond-5P1-free}).
Add edges to make each colour class into a clique and let~$G_n$ be the resulting graph.
Now~$G_n$ is $(\overline{2P_1+P_2},5P_1)$-free and the family of such graphs had unbounded clique-width~\cite{DGP14} (the former can also be seen by inspection and the latter follows from combining Lemma~\ref{lem:generalunbounded} with Fact~\ref{fact:comp}).
By Fact~\ref{fact:comp}, the family of graphs~$\overline{G_n}$ also has unbounded clique-width.

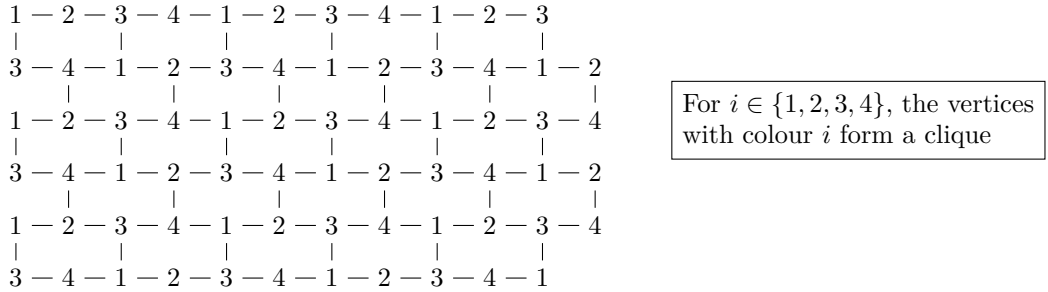
\begin{figure}[h]
\begin{center}
\begin{tikzpicture}[scale=0.7]
\path
      (0,-1) node[](x0ym1) {3}
      (1,-1) node[](x1ym1) {4}
      (2,-1) node[](x2ym1) {1}
      (3,-1) node[](x3ym1) {2}
      (4,-1) node[](x4ym1) {3}
      (5,-1) node[](x5ym1) {4}
      (6,-1) node[](x6ym1) {1}
      (7,-1) node[](x7ym1) {2}
      (8,-1) node[](x8ym1) {3}
      (9,-1) node[](x9ym1) {4}
      (10,-1) node[](x10ym1) {1}

      (0,0) node[](x0y0) {1}
      (1,0) node[](x1y0) {2}
      (2,0) node[](x2y0) {3}
      (3,0) node[](x3y0) {4}
      (4,0) node[](x4y0) {1}
      (5,0) node[](x5y0) {2}
      (6,0) node[](x6y0) {3}
      (7,0) node[](x7y0) {4}
      (8,0) node[](x8y0) {1}
      (9,0) node[](x9y0) {2}
      (10,0) node[](x10y0) {3}
      (11,0) node[](x11y0) {4}

      (0,1) node[](x0y1) {3}
      (1,1) node[](x1y1) {4}
      (2,1) node[](x2y1) {1}
      (3,1) node[](x3y1) {2}
      (4,1) node[](x4y1) {3}
      (5,1) node[](x5y1) {4}
      (6,1) node[](x6y1) {1}
      (7,1) node[](x7y1) {2}
      (8,1) node[](x8y1) {3}
      (9,1) node[](x9y1) {4}
      (10,1) node[](x10y1) {1}
      (11,1) node[](x11y1) {2}

      (0,2) node[](x0y2) {1}
      (1,2) node[](x1y2) {2}
      (2,2) node[](x2y2) {3}
      (3,2) node[](x3y2) {4}
      (4,2) node[](x4y2) {1}
      (5,2) node[](x5y2) {2}
      (6,2) node[](x6y2) {3}
      (7,2) node[](x7y2) {4}
      (8,2) node[](x8y2) {1}
      (9,2) node[](x9y2) {2}
      (10,2) node[](x10y2) {3}
      (11,2) node[](x11y2) {4}

      (0,3) node[](x0y3) {3}
      (1,3) node[](x1y3) {4}
      (2,3) node[](x2y3) {1}
      (3,3) node[](x3y3) {2}
      (4,3) node[](x4y3) {3}
      (5,3) node[](x5y3) {4}
      (6,3) node[](x6y3) {1}
      (7,3) node[](x7y3) {2}
      (8,3) node[](x8y3) {3}
      (9,3) node[](x9y3) {4}
      (10,3) node[](x10y3) {1}
      (11,3) node[](x11y3) {2}

      (0,4) node[](x0y4) {1}
      (1,4) node[](x1y4) {2}
      (2,4) node[](x2y4) {3}
      (3,4) node[](x3y4) {4}
      (4,4) node[](x4y4) {1}
      (5,4) node[](x5y4) {2}
      (6,4) node[](x6y4) {3}
      (7,4) node[](x7y4) {4}
      (8,4) node[](x8y4) {1}
      (9,4) node[](x9y4) {2}
      (10,4) node[](x10y4) {3}
;

\draw (x0ym1) -- (x1ym1) -- (x2ym1) -- (x3ym1) -- (x4ym1) -- (x5ym1) -- (x6ym1) -- (x7ym1) -- (x8ym1) -- (x9ym1) -- (x10ym1);
\draw (x0y0) -- (x1y0) -- (x2y0) -- (x3y0) -- (x4y0) -- (x5y0) -- (x6y0) -- (x7y0) -- (x8y0) -- (x9y0) -- (x10y0) -- (x11y0) ;
\draw (x0y1) -- (x1y1) -- (x2y1) -- (x3y1) -- (x4y1) -- (x5y1) -- (x6y1) -- (x7y1) -- (x8y1) -- (x9y1) -- (x10y1) -- (x11y1);
\draw (x0y2) -- (x1y2) -- (x2y2) -- (x3y2) -- (x4y2) -- (x5y2) -- (x6y2) -- (x7y2) -- (x8y2) -- (x9y2) -- (x10y2) -- (x11y2);
\draw (x0y3) -- (x1y3) -- (x2y3) -- (x3y3) -- (x4y3) -- (x5y3) -- (x6y3) -- (x7y3) -- (x8y3) -- (x9y3) -- (x10y3) -- (x11y3);
\draw (x0y4) -- (x1y4) -- (x2y4) -- (x3y4) -- (x4y4) -- (x5y4) -- (x6y4) -- (x7y4) -- (x8y4) -- (x9y4) -- (x10y4);

\draw (x1y0) -- (x1y1) ;
\draw (x1y2) -- (x1y3) ;
\draw (x3y0) -- (x3y1) ;
\draw (x3y2) -- (x3y3) ;
\draw (x5y0) -- (x5y1) ;
\draw (x5y2) -- (x5y3) ;
\draw (x7y0) -- (x7y1) ;
\draw (x7y2) -- (x7y3) ;
\draw (x9y0) -- (x9y1) ;
\draw (x9y2) -- (x9y3) ;
\draw (x11y0) -- (x11y1) ;
\draw (x11y2) -- (x11y3) ;

\draw (x0ym1) -- (x0y0) ;
\draw (x0y1) -- (x0y2) ;
\draw (x0y3) -- (x0y4) ;
\draw (x2ym1) -- (x2y0) ;
\draw (x2y1) -- (x2y2) ;
\draw (x2y3) -- (x2y4) ;
\draw (x4ym1) -- (x4y0) ;
\draw (x4y1) -- (x4y2) ;
\draw (x4y3) -- (x4y4) ;
\draw (x6ym1) -- (x6y0) ;
\draw (x6y1) -- (x6y2) ;
\draw (x6y3) -- (x6y4) ;
\draw (x8ym1) -- (x8y0) ;
\draw (x8y1) -- (x8y2) ;
\draw (x8y3) -- (x8y4) ;
\draw (x10ym1) -- (x10y0) ;
\draw (x10y1) -- (x10y2) ;
\draw (x10y3) -- (x10y4) ;

\node[draw=black, rectangle, align=left, inner sep=4pt] (label) at (16,2) {For $i\in\{1,2,3,4\}$, the vertices\\ with colour~$i$ form a clique};
\end{tikzpicture}
\caption{The graph~$G_n$ from the proof of Lemma~\ref{lem:diamond-5P1} ($n=2$ shown).}
\label{fig:diamond-5P1-free}
\end{center}
\end{figure}

It remains to show that~$G_n$ and~$\overline{G_n}$ are atoms.
Let~$V_i$ be the set of vertices with colour~$i$.
Suppose, for contradiction, that~$G_n$ has a clique cut-set~$X$.
If $X \subseteq V_i$ for some $i \in \{1,2,3,4\}$, then all vertices of $G_n \setminus V_i$ are in the same component of $G_n \setminus X$.
Since every vertex in~$V_i$ has at least one neighbour outside of~$V_i$, it follows that every vertex of $G_n \setminus X$ is in the same component of $G_n \setminus X$ in this case, a contradiction.
We may therefore assume that~$X$ contains vertices in at least two sets~$V_i$.
By construction, each vertex in a set~$V_i$ has at most one neighbour in each~$V_j$ for $j \in \{1,2,3,4\} \setminus \{i\}$.
Therefore~$X$ has at most one vertex in each~$V_i$.
Therefore, there must be a vertex in $V_1\setminus X$ that has a neighbour in each of $V_2 \setminus X$, $V_3 \setminus X$ and $V_4 \setminus X$.
Since each set~$V_i$ is a clique, it follows that $G_n \setminus X$ is connected.
This contradiction implies that~$G_n$ is indeed an atom.
Now suppose, for contradiction, that~$\overline{G_n}$ has a clique cut-set~$X$.
Since $V_1,\ldots,V_4$ are independent sets in~$\overline{G_n}$, $X$ contains at most one vertex of any~$V_i$.
Since in~$\overline{G_n}$ every vertex of~$V_i$ has at most one non-neighbour in each~$V_j$ for $j \in \{1,2,3,4\} \setminus \{i\}$, it follows that $\overline{G_n} \setminus X$ must be connected.
This contradiction implies that~$\overline{G_n}$ is indeed an atom.
\end{proof}

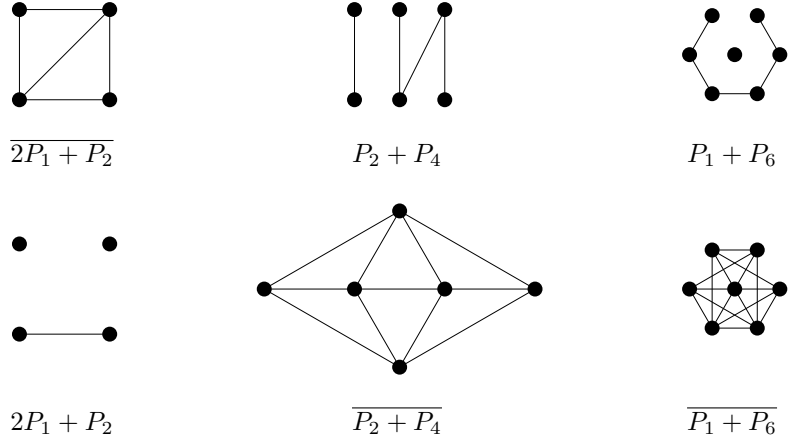
\begin{figure}[h]
\begin{center}
\begin{tabular}{ccc}
\scalebox{0.6}{
\begin{minipage}{0.4\textwidth}
\begin{center}%diamond
\begin{tikzpicture}[every node/.style={circle,fill, minimum size=0.07cm}]
\node at (0,0) {};
\node at (0,2) {};
\node at (2,0) {};
\node at (2,2) {};
\draw (0,0) -- (0,2) -- (2,2) -- (2,0) -- (0,0) -- (2,2);
\end{tikzpicture}
\end{center}
\end{minipage}}
&
\scalebox{0.6}{
\begin{minipage}{0.4\textwidth}
\begin{center}%P2+P4
\begin{tikzpicture}[every node/.style={circle,fill, minimum size=0.07cm}]
\node at (0,0) {};
\node at (0,2) {};
\node at (1,0) {};
\node at (1,2) {};
\node at (2,0) {};
\node at (2,2) {};
\draw (0,0) -- (0,2);
\draw (1,0) -- (1,2);
\draw (2,0) -- (2,2) -- (1,0);
\end{tikzpicture}
\end{center}
\end{minipage}}
&
\scalebox{0.6}{
\begin{minipage}{0.4\textwidth}
\begin{center}%P1+P6
\begin{tikzpicture}[every node/.style={circle,fill, minimum size=0.07cm}]
\node at (0,0) {};
\node at (0:1) {};
\node at (60:1) {};
\node at (120:1) {};
\node at (180:1) {};
\node at (240:1) {};
\node at (300:1) {};
\draw (120:1) -- (180:1) -- (240:1) -- (300:1) -- (0:1) -- (60:1);
\end{tikzpicture}
\end{center}
\end{minipage}}
\\
\\
$\overline{2P_1+P_2}$ & $P_2+\nobreak P_4$ & $P_1+\nobreak P_6$\\
\\
\scalebox{0.6}{
\begin{minipage}{0.4\textwidth}
\begin{center}%2P_1+P_2
\begin{tikzpicture}[every node/.style={circle,fill, minimum size=0.07cm}]
\node at (0,0) {};
\node at (0,2) {};
\node at (2,0) {};
\node at (2,2) {};
\draw (2,0) -- (0,0);
\end{tikzpicture}
\end{center}
\end{minipage}}
&
\scalebox{0.6}{
\begin{minipage}{0.4\textwidth}
\begin{center}%co-P2+P4
\begin{tikzpicture}[every node/.style={circle,fill, minimum size=0.07cm}]
\node at (0,1.73205080757) {};%sqrt(3)
\node at (0,-1.73205080757) {};%sqrt(3)
\node at (1,0) {};
\node at (3,0) {};
\node at (-1,0) {};
\node at (-3,0) {};
\draw (-3,0) -- (-1,0) -- (1,0) -- (3,0);
\draw (0,1.73205080757) -- (1,0) -- (0,-1.73205080757);
\draw (0,1.73205080757) -- (3,0) -- (0,-1.73205080757);
\draw (0,1.73205080757) -- (-1,0) -- (0,-1.73205080757);
\draw (0,1.73205080757) -- (-3,0) -- (0,-1.73205080757);
\end{tikzpicture}
\end{center}
\end{minipage}}
&
\scalebox{0.6}{
\begin{minipage}{0.4\textwidth}
\begin{center}%co-P1+P6
\begin{tikzpicture}[every node/.style={circle,fill, minimum size=0.07cm}]
\node at (0,0) {};
\node at (0:1) {};
\node at (60:1) {};
\node at (120:1) {};
\node at (180:1) {};
\node at (240:1) {};
\node at (300:1) {};
\draw (0,0) -- (0:1) -- (120:1) -- (240:1) -- (0:1);
\draw (0,0) -- (60:1) -- (300:1) -- (180:1) -- (60:1);
\draw (0,0) -- (120:1) -- (60:1);
\draw (0,0) -- (180:1) -- (240:1) -- (300:1) -- (0:1);
\draw (0,0) -- (240:1);
\draw (0,0) -- (300:1);
\end{tikzpicture}
\end{center}
\end{minipage}}
\\
\\
$2P_1+\nobreak P_2$ & $\overline{P_2\nobreak+P_4}$ & $\overline{P_1+P_6}$
\end{tabular}
\end{center}
\caption{\label{fig:diamond-P2+P4-forb}The forbidden induced subgraphs for the classes of $(\overline{2P_1+P_2},P_2+\nobreak P_4,P_1+\nobreak P_6)$-free graphs and $(2P_1+\nobreak P_2,\allowbreak \overline{P_2\nobreak+P_4}, \overline{P_1+P_6})$-free graphs mentioned in Lemma~\ref{lem:diamond-P2+P4}.}
\end{figure}

\begin{lemma}\label{lem:diamond-P2+P4}
The class of $(\overline{2P_1+P_2},P_2+\nobreak P_4,P_1+\nobreak P_6)$-free and the class of $(2P_1+\nobreak P_2,\allowbreak \overline{P_2\nobreak+P_4}, \overline{P_1+P_6})$-free atoms have unbounded clique-width (see \figurename~\ref{fig:diamond-P2+P4-forb} for illustrations of the forbidden induced subgraphs).
\end{lemma}
\begin{proof}
We modify the construction of the graph~$G_n$, which was used in~\cite{DHP0} to prove that $(\overline{2P_1+P_2},P_2+\nobreak P_4)$-free graphs have unbounded clique-width.
Consider a wall of height $n\geq 2$.
A wall is a bipartite graph; let~$A$ and~$C$ be the two sets in its bipartition.
Consider a $1$-subdivision of the wall and let~$B$ be the set of vertices introduced by the subdivision.
Finally, we add edges to make~$A$ complete to~$C$.
Let~$G_n$ be the resulting graph.
Then~$G_n$ is $(\overline{2P_1+P_2},P_2+\nobreak P_4)$-free and the family of such graphs~$G_n$ has unbounded clique-width~\cite{DHP0} (the former also follows by inspection and the latter follows from combining Lemma~\ref{lem:walls} with Fact~\ref{fact:bip}).
Let~$H_n$ be the graph obtained from~$G_n$ by adding a vertex~$x$ complete to~$B$, see \figurename~\ref{fig:diamond-P2+P4}.
Since~$H_n$ contains~$G_n$ as an induced subgraph, the family of graphs~$H_n$ has unbounded clique-width. 

\begin{figure}[h]
\begin{center}
\begin{tikzpicture}[scale=0.7,  every node/.style={inner sep=1pt,minimum size=1mm,circle}  ,xscale=-1]

\node  (cO-10) at (-1*2,2)  {$C$};
\node  (aE-11) at (-1*2,4)  {$A$};
  
   \foreach \y in {0,1}{
     \foreach \x in {0,1,2}{
 
        \node (cE\x\y)   at (\x*4,\y*4)    {$C$};
        \node (aE\x\y)   at (\x*4+2,\y*4)  {$A$};
        \node (bcE\x\y)  at (\x*4+1,\y*4)  {$B$};
        \node (baE\x\y)  at (\x*4+3,\y*4)  {$B$};
        \draw (cE\x\y)--(bcE\x\y)--(aE\x\y)--(baE\x\y);
       
        \node (aO\x\y)   at (\x*4,\y*4+2)   {$A$};
        \node (cO\x\y)   at (\x*4+2,\y*4+2) {$C$};
        \node (baO\x\y)  at (\x*4+1,\y*4+2)   {$B$};
        \node (bcO\x\y)  at (\x*4+3,\y*4+2)   {$B$};
        \draw (aO\x\y)--(baO\x\y)--(cO\x\y)--(bcO\x\y);
         }
    \node (cE3\y) at (3*4,\y*4) {$C$};
    \node (aO3\y) at (3*4,\y*4+2) {$A$};
    
    \foreach \x in {0,1,2}{
            \pgfmathtruncatemacro{\j}{\x+1}
        \draw (baE\x\y)--(cE\j\y);
        \draw (bcO\x\y)--(aO\j\y);
    }
     
     \foreach \x in {0,1,2,3}{
        \node (bb\x\y) at (\x*4,\y*4+1) {$B$};
         \draw (cE\x\y)--(bb\x\y) -- (aO\x\y) ;
         \ifnum \y=1
            \node (bb\x2) at (\x*4-2,\y*2+1) {$B$};
            \pgfmathtruncatemacro{\j}{\x-1}
            \draw ({cO\j0})--(bb\x2)--(aE\j1);
         \fi
         }
  }
 
\node (bcO-10) at (-1*2+1,2) {$B$};
\node (baE-11) at (-1*2+1,4) {$B$};
\draw (cO-10)--(bcO-10)--(aO00);
\draw (aE-11)--(baE-11)--(cE01);
\node (x) at (5,-2) {$x$};

\foreach \x in {-1,1.6,3,3.5,3.8,4,4.2,4.5,5, ,6,7.5,9}{
    \draw (\x+1,-1)--(x);
    }
    \node[draw=black, rectangle, align=left, inner sep=4pt] (label) at (15,3) {\\ $A$ is complete to $C$\\ $x$ is complete to $B$};

% \draw[line width=1mm,darkgray] (3.1,-1.8) to[bend left] (6.9,-2);
%    \node[line width=0.5mm,draw=darkgray, rectangle, align=left] (label) at (8.6,-2) {$X$ is complete to $B$\\ $A$ is complete to $C$};
;

\end{tikzpicture}
\caption{The graph~$H_n$ from the proof of Lemma~\ref{lem:diamond-P2+P4} ($n=3$ shown).
Vertices are denoted~$A$, $B$ or~$C$ if they are in the corresponding set.}
\label{fig:diamond-P2+P4}
\end{center}
\end{figure}
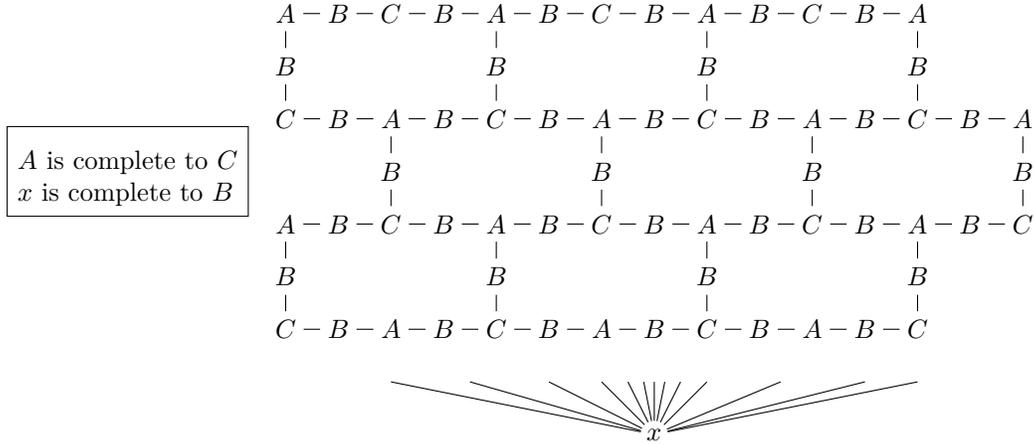

Now~$G_n$ is $(\overline{2P_1+P_2},P_2+\nobreak P_4)$-free, so if~$H_n$ contains an induced copy of~$\overline{2P_1+P_2}$ or~$P_2+\nobreak P_4$, then one of its vertices must be~$x$.
The neighbourhood of~$x$ in~$H_n$ is~$B$, which is an independent set.
Every vertex of~$\overline{2P_1+P_2}$ has two neighbours that are adjacent to each other, so~$H_n$ is $\overline{2P_1+P_2}$-free.
Suppose, for contradiction, that~$H_n$ contains an induced~$P_2+\nobreak P_4$, say with vertex set~$Y$.
As observed above, $x \in Y$.
Now~$x$ has either one or two neighbours in~$H_n[Y]$.
If~$x$ has one neighbour in~$H_n[Y]$, then this neighbour must be in~$B$, and then there can be no other vertices in $B \cap Y$, so $|(A \cup C) \cap Y| =4$, but $H_n[A \cup C]$ is a complete bipartite graph, so $H_n[(A \cup C) \cap Y]$ is isomorphic to $C_4,K_{1,3}$ or~$4P_1$, contradicting the fact that~$P_2+\nobreak P_4$ is $(C_4,K_{1,3},4P_1)$-free.
Therefore~$x$ has two neighbours in~$H_n[Y]$, so it is in the~$P_4$ component of~$H_n[Y]$.
These two neighbours of~$x$ must be in~$B$, so the~$P_4$ component containing~$x$ must contain a vertex of~$A$ or~$C$ and the remaining~$P_2$ component of~$H_n[Y]$ must lie in $A \cup C$.
Since $H_n[A \cup C]$ is a complete bipartite graph, it follows that there is an edge between the~$P_2$ component and the~$P_4$ component.
This contradiction implies that~$H_n$ is indeed $(P_2+\nobreak P_4)$-free.

Suppose, for contradiction, that~$H_n$ contains an induced~$P_1+\nobreak P_6$, say with vertex set~$\{v\}\cup Y$ where~$v$ is the vertex in the~$P_1$ component and~$Y$ is the vertex set of the~$P_6$ component.
If there are three vertices in~$B \cap Y$ and $x \in Y$, then~$H_n[Y]$ contains an induced~$K_{1,3}$, a contradiction.
Note that~$P_6$ has two vertices of degree~$1$ and four vertices of degree~$2$, but every vertex in~$B$ has only two neighbours apart from~$x$: one in each of~$A$ and~$C$.
Therefore, if there are three vertices in~$B \cap Y$, then one of these vertices~$b$ must have neighbours $a \in A \cap Y$ and $c \in C \cap Y$, in which case $H_n[a,b,c]$ is a~$K_3$, a contradiction.
We conclude that there are at most two vertices in~$B \cap Y$.
If $x\notin Y$ then there are at least four vertices in $(A\cup C)\cap Y$, contradicting the fact that~$P_6$ is $(C_4,K_{1,3},4P_1)$-free. 
Therefore, $x\in Y$, and so $v\in A\cup C$ (say~$A$) because~$x$ is complete to~$B$.
Since $\{x\} \cup A$ is independent and~$P_6$ is $4P_1$-free, it follows that $|A \cap Y| \leq 2$, and so there is at least one vertex $c \in C \cap Y$.
But~$c$ is complete to~$A$, so it is adjacent to~$v$, a contradiction.
This contradiction implies that $H_n$ is indeed $(P_1+\nobreak P_6)$-free.

It remains to show that~$H_n$ and~$\overline{H_n}$ are atoms.
Suppose, for contradiction, that~$H_n$ contains a clique cut-set~$X$.
If $x \in X$ then~$X$ contains at most one additional vertex, which must lie in~$B$; it is easy to verify that $H_n \setminus X$ is connected in this case.
We may therefore assume that $x \notin X$.
Since~$A$, $B$ and~$C$ are independent sets, $X$ contains at most one vertex in each of these sets.
Since $x\notin X$, and~$x$ is complete to~$B$, all vertices of~$B\setminus X$ are in the same component of $H_n \setminus X$.
Since every vertex of~$B$ has a neighbour in~$A$ and~$C$, there must be a vertex in $B \setminus X$ that has neighbours in both $A \setminus X$ and $C \setminus X$.
Since~$A$ is complete to~$C$, it follows that every vertex in $V(H_n) \setminus X$ is in the same component of $H_n \setminus X$.
This contradiction implies that~$H_n$ is indeed an atom.
Now suppose, for contradiction, that~$\overline{H_n}$ contains a clique cut-set~$X$.
Since~$A$ is anti-complete to~$C$ in~$\overline{H_n}$, $X$ cannot contain vertices in both~$A$ and~$C$; by symmetry we may assume that~$X$ does not contain any vertices of~$C$.
Now~$C$ is a clique and, since every vertex of~$B$ has a neighbour in~$C$, every vertex in $(B \cup C) \setminus X$ is in the same component of $\overline{H_n} \setminus X$.
If $x \notin X$, then every vertex in $A\setminus X$ is adjacent to~$x$, which is complete to~$C$, so every vertex in $V(\overline{H_n}) \setminus X$ is in the same component of $\overline{H_n} \setminus X$, a contradiction.
We may therefore assume that $x \in X$.
Then no vertex of~$B$ is in~$X$, so $X \subseteq A \cup \{x\}$.
Since every vertex of~$A$ has a neighbour in~$B$, it follows that every vertex of~$A$ has a neighbour in $B\setminus X=B$.
Therefore every vertex of $V(\overline{H_n}) \setminus X$ is in the same component of $\overline{H_n} \setminus X$.
This contradiction implies that~$\overline{H_n}$ is indeed an atom.
\end{proof}

\begin{figure}[h]
\begin{center}
\begin{tabular}{ccccc}
\scalebox{0.6}{
\begin{minipage}{0.25\textwidth}
\begin{center}%2P1+2P2
\begin{tikzpicture}[every node/.style={circle,fill, minimum size=0.07cm}]
\node at (0,0) {};
\node at (0,2) {};
\node at (1,0) {};
\node at (1,2) {};
\node at (2,0) {};
\node at (2,2) {};
\draw (1,0) -- (1,2);
\draw (2,0) -- (2,2);
\end{tikzpicture}
\end{center}
\end{minipage}}
&
\scalebox{0.6}{
\begin{minipage}{0.25\textwidth}
\begin{center}%2P1+P4
\begin{tikzpicture}[every node/.style={circle,fill, minimum size=0.07cm}]
\node at (0,0) {};
\node at (0,2) {};
\node at (1,0) {};
\node at (1,2) {};
\node at (2,0) {};
\node at (2,2) {};
\draw (1,0) -- (1,2);
\draw (2,0) -- (2,2) -- (1,0);
\end{tikzpicture}
\end{center}
\end{minipage}}
&
\scalebox{0.6}{
\begin{minipage}{0.25\textwidth}
\begin{center}%4P1+P2
\begin{tikzpicture}[every node/.style={circle,fill, minimum size=0.07cm}]
\node at (0,0) {};
\node at (0,2) {};
\node at (1,0) {};
\node at (1,2) {};
\node at (2,0) {};
\node at (2,2) {};
\draw (2,0) -- (2,2);
\end{tikzpicture}
\end{center}
\end{minipage}}
&
\scalebox{0.6}{
\begin{minipage}{0.25\textwidth}
\begin{center}%3P2
\begin{tikzpicture}[every node/.style={circle,fill, minimum size=0.07cm}]
\node at (0,0) {};
\node at (0,2) {};
\node at (1,0) {};
\node at (1,2) {};
\node at (2,0) {};
\node at (2,2) {};
\draw (0,0) -- (0,2);
\draw (1,0) -- (1,2);
\draw (2,0) -- (2,2);
\end{tikzpicture}
\end{center}
\end{minipage}}
&
\scalebox{0.6}{
\begin{minipage}{0.25\textwidth}
\begin{center}%C_3
\begin{tikzpicture}[every node/.style={circle,fill, minimum size=0.07cm}]
\node at (0,0) {};
\node at (1,1.73205080757) {};%sqrt(3)
\node at (2,0) {};
\draw (0,0) -- (1,1.73205080757) -- (2,0) -- (0,0);
\end{tikzpicture}
\end{center}
\end{minipage}}
\\
\\
$2P_1+\nobreak 2P_2$ & $2P_1+\nobreak P_4$ & $4P_1+\nobreak P_2$ & $3P_2$ & $C_{2k+1}$ for $k \geq 1$\\
&&&& ($k=1$ shown)
\\
\\
\scalebox{0.6}{
\begin{minipage}{0.25\textwidth}
\begin{center}%co-2P1+2P2
\begin{tikzpicture}[every node/.style={circle,fill, minimum size=0.07cm}]
\node at (0:1) {};
\node at (60:1) {};
\node at (120:1) {};
\node at (180:1) {};
\node at (240:1) {};
\node at (300:1) {};
\draw (120:1) -- (180:1) -- (240:1) -- (300:1) -- (0:1) -- (60:1) -- (120:1);
\draw (0:1) -- (120:1) -- (240:1) -- (0:1);
\draw (60:1) -- (180:1) -- (300:1) -- (60:1);
\draw (0:1) -- (180:1);
\end{tikzpicture}
\end{center}
\end{minipage}}
&
\scalebox{0.6}{
\begin{minipage}{0.25\textwidth}
\begin{center}%co-2P1+P4
\begin{tikzpicture}[every node/.style={circle,fill, minimum size=0.07cm},rotate=180]
\node at (0:1) {};
\node at (60:1) {};
\node at (120:1) {};
\node at (180:1) {};
\node at (240:1) {};
\node at (300:1) {};
\draw (120:1) -- (180:1) -- (240:1) -- (300:1) -- (0:1) -- (60:1) -- (120:1);
\draw (120:1) -- (240:1) -- (0:1);
\draw (180:1) -- (300:1) -- (60:1);
\draw (60:1) -- (240:1);
\draw (120:1) -- (300:1);
\end{tikzpicture}
\end{center}
\end{minipage}}
&
\scalebox{0.6}{
\begin{minipage}{0.25\textwidth}
\begin{center}%co-4P1+P2
\begin{tikzpicture}[every node/.style={circle,fill, minimum size=0.07cm}]
\node at (0:1) {};
\node at (60:1) {};
\node at (120:1) {};
\node at (180:1) {};
\node at (240:1) {};
\node at (300:1) {};
\draw (120:1) -- (180:1) -- (240:1) -- (300:1) -- (0:1) -- (60:1) -- (120:1);
\draw (0:1) -- (120:1) -- (240:1) -- (0:1);
\draw (60:1) -- (180:1) -- (300:1) -- (60:1);
\draw (60:1) -- (240:1);
\draw (120:1) -- (300:1);
\end{tikzpicture}
\end{center}
\end{minipage}}
&
\scalebox{0.6}{
\begin{minipage}{0.25\textwidth}
\begin{center}%co-3P2
\begin{tikzpicture}[every node/.style={circle,fill, minimum size=0.07cm}]
\node at (0:1) {};
\node at (60:1) {};
\node at (120:1) {};
\node at (180:1) {};
\node at (240:1) {};
\node at (300:1) {};
\draw (120:1) -- (180:1) -- (240:1) -- (300:1) -- (0:1) -- (60:1) -- (120:1);
\draw (0:1) -- (120:1) -- (240:1) -- (0:1);
\draw (60:1) -- (180:1) -- (300:1) -- (60:1);
\end{tikzpicture}
\end{center}
\end{minipage}}
&
\scalebox{0.6}{
\begin{minipage}{0.25\textwidth}
\begin{center}%3P_1
\begin{tikzpicture}[every node/.style={circle,fill, minimum size=0.07cm}]
\node at (0,0) {};
\node at (1,1.73205080757) {};%sqrt(3)
\node at (2,0) {};
\end{tikzpicture}
\end{center}
\end{minipage}}
\\
\\
$\overline{2P_1+2P_2}$ & $\overline{2P_1+P_4}$ & $\overline{4P_1+P_2}$ & $\overline{3P_2}$ & $\overline{C_{2k+1}}$ for $k \geq 1$\\
&&&& ($k=1$ shown)
\end{tabular}
\end{center}
\caption{\label{fig:bip-2P12P2-2P1P4-4P1P2-3P2-forb}The forbidden induced subgraphs for the classes of $(2P_1+\nobreak 2P_2,2P_1+\nobreak P_4, 4P_1+\nobreak P_2, 3P_2)$-free bipartite graphs and $(\overline{2P_1+2P_2},\overline{2P_1+P_4},\overline{4P_1+P_2},\overline{3P_2})$-free co-bipartite graphs mentioned in Lemma~\ref{lem:bip-2P12P2-2P1P4-4P1P2-3P2}.}
\end{figure}
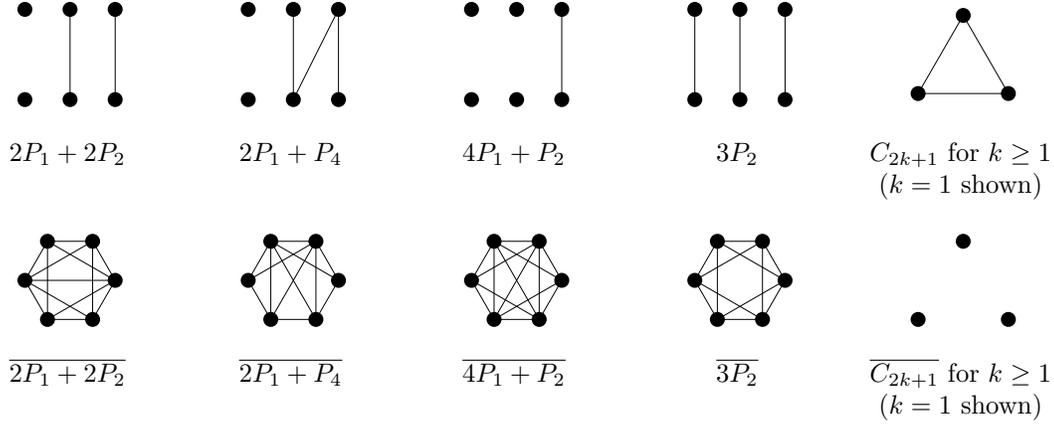

\begin{lemma}\label{lem:bip-2P12P2-2P1P4-4P1P2-3P2}
The class of $(2P_1+\nobreak 2P_2,2P_1+\nobreak P_4, 4P_1+\nobreak P_2, 3P_2)$-free bipartite atoms and the class of $(\overline{2P_1+2P_2},\overline{2P_1+P_4},\overline{4P_1+P_2},\overline{3P_2})$-free co-bipartite atoms have unbounded clique-width (see \figurename~\ref{fig:bip-2P12P2-2P1P4-4P1P2-3P2-forb} for illustrations of the forbidden induced subgraphs).
\end{lemma}

\begin{proof}
Let~$H_n$ be a $1$-subdivided wall of height $n \geq 2$ and note that the class of such graphs has unbounded clique-width by Lemma~\ref{lem:walls}.
Note that~$H_n$ is connected and bipartite, say with parts~$V_1$ and~$V_2$.
Let~$G_n$ be the graph obtained from~$H_n$ by applying a bipartite complementation between~$V_1$ and~$V_2$.
By Fact~\ref{fact:bip}, the family of such graphs~$G_n$ also has unbounded clique-width and by Fact~\ref{fact:comp} the family of graphs~$\overline{G_n}$ also has unbounded clique-width (see also \figurename~\ref{fig:bip-2P12P2-2P1P4-4P1P2-3P2}).
Now~$G_n$ is a $(2P_1+\nobreak 2P_2,2P_1+\nobreak P_4, 4P_1+\nobreak P_2, 3P_2)$-free bipartite graph (by inspection, or see e.g.~\cite{DP14,LV08}).

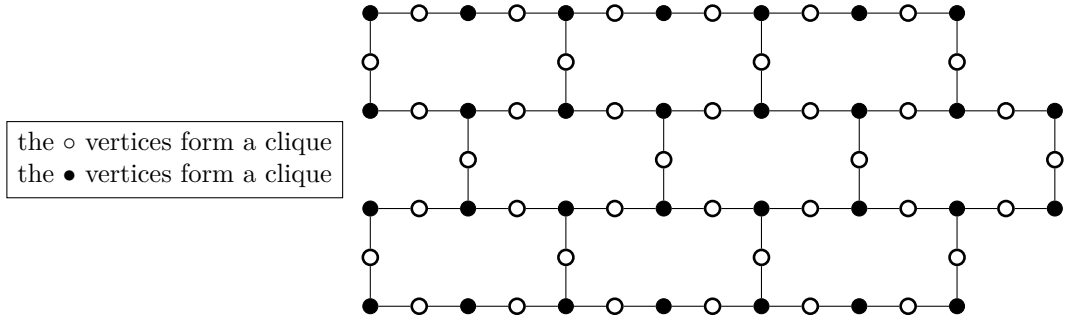
\begin{figure}[h]
\begin{center}
\begin{tikzpicture}[scale=0.65, type1/.style={circle, draw, fill, inner sep=2pt,minimum size=0.5mm}, type2/.style={circle,thick, draw=black, line width=1pt, inner sep=2pt,minimum size=0.5mm},rotate=180   ]

\node[type1]  (aa-11) at (-1*2,2)  {};
\node[type1]  (aa-12) at (-1*2,4)  {};
\node[type2] (bm-11) at (-1*2+1,2) {};
\node[type2]  (bm-12) at (-1*2+1,4) {};
  
   \foreach \y in {0,1,2,3}{
     \foreach \x in {0,1,2}{
        \node[type1] (a\x\y)   at (\x*4,\y*2)    {};
        \node[type1] (aa\x\y)   at (\x*4+2,\y*2)    {};
        \node[type2] (b\x\y)  at (\x*4+1,\y*2)  {};
        \node[type2]  (bb\x\y)  at (\x*4+3,\y*2)  {};
        \draw (a\x\y)--(b\x\y)--(aa\x\y)--(bb\x\y);
         }
    \node[type1] (a3\y) at (3*4,\y*2) {};
    
    \foreach \x in {0,1,2}{
            \pgfmathtruncatemacro{\j}{\x+1}
       \draw (bb\x\y)--(a\j\y);
    }
 }

 \foreach \x in {0,1,2,3}{
    \foreach \y in {0,2}{
        \node[type2]  (bm\x\y) at (\x*4,\y*2+1) {};
        \pgfmathtruncatemacro{\j}{\y+1}
        \draw (a\x\y)--(bm\x\y) -- (a\x\j) ;
        }
     
     \node[type2]  (bm\x2) at (\x*4-2,1*2+1) {};
     \pgfmathtruncatemacro{\j}{\x-1}
     \draw (aa\j1)--(bm\x2)--(aa\j2);
}
 
\draw (aa-11)--(bm-11)--(a01);
\draw (aa-12)--(bm-12)--(a02);

\node[draw=black, rectangle, align=left, inner sep=4pt] (label) at (16,3) {the $\circ$ vertices form a clique \\ the $\bullet$ vertices form a clique};

;

\end{tikzpicture}
 \caption{The graph~$\overline{G_n}$ from the proof of Lemma~\ref{lem:bip-2P12P2-2P1P4-4P1P2-3P2} ($n=3$ shown).}
\label{fig:bip-2P12P2-2P1P4-4P1P2-3P2}
\end{center}
\end{figure}

It remains to show that~$G_n$ and~$\overline{G_n}$ are atoms.
Suppose, for contradiction, that~$X$ is a clique cut-set of~$G_n$.
Since~$V_1$ and~$V_2$ are independent, $X$ contains at most one vertex from each of these sets.
Since every vertex of~$V_1$ has at most three non-neighbours in~$V_2$ and vice versa, it follows that $G_n \setminus X$ is connected.
This contradiction shows that~$G_n$ is indeed an atom.
Now suppose, for contradiction, that~$X$ is a clique cut-set of~$\overline{G_n}$.
If~$X$ is a subset of either~$V_1$ or~$V_2$, say $V_1$, then every vertex of~$V_2$ lies outside~$X$.
Since every vertex of~$V_1$ has a neighbour in~$V_2$, it follows that $\overline{G_n} \setminus X$ is connected in this case.
Therefore~$X$ must contain at least one vertex of~$V_1$ and at least one vertex of~$V_2$.
In~$\overline{G_n}$, every vertex in~$V_1$ has at most three neighbours in~$V_2$ and vice versa, so~$X$ contains at most three vertices from~$V_1$ and at most three vertices from~$V_2$.
In~$\overline{G_n}$, every vertex in~$V_1$ has at most three neighbours in~$V_2$, and every vertex in~$V_2$ has at least one neighbour in~$V_1$, and $|V_2|>12=9+3$.
Hence, there must be a vertex in $V_2 \setminus X$ with a neighbour in $V_1 \setminus X$.
Since~$V_1$ and~$V_2$ are cliques, it follows that $\overline{G_n} \setminus X$ is connected.
This completes the proof.
\end{proof}

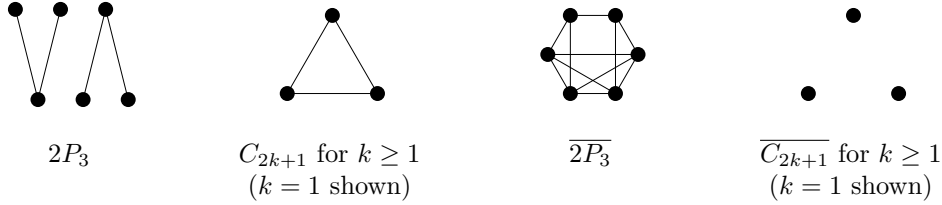
\begin{figure}[h]
\begin{center}
\begin{tabular}{cccc}
\scalebox{0.6}{
\begin{minipage}{0.3\textwidth}
\begin{center}%2P3
\begin{tikzpicture}[every node/.style={circle,fill, minimum size=0.07cm}]
\node at (0.5,0) {};
\node at (0,2) {};
\node at (1.5,0) {};
\node at (1,2) {};
\node at (2.5,0) {};
\node at (2,2) {};
\draw (1,2) -- (0.5,0) -- (0,2);
\draw (2.5,0) -- (2,2) -- (1.5,0);
\end{tikzpicture}
\end{center}
\end{minipage}}
&
\scalebox{0.6}{
\begin{minipage}{0.3\textwidth}
\begin{center}%C_3
\begin{tikzpicture}[every node/.style={circle,fill, minimum size=0.07cm}]
\node at (0,0) {};
\node at (1,1.73205080757) {};%sqrt(3)
\node at (2,0) {};
\draw (0,0) -- (1,1.73205080757) -- (2,0) -- (0,0);
\end{tikzpicture}
\end{center}
\end{minipage}}
&
\scalebox{0.6}{
\begin{minipage}{0.3\textwidth}
\begin{center}%co-2P3
\begin{tikzpicture}[every node/.style={circle,fill, minimum size=0.07cm}]
\node at (0:1) {};
\node at (60:1) {};
\node at (120:1) {};
\node at (180:1) {};
\node at (240:1) {};
\node at (300:1) {};
\draw (120:1) -- (180:1) -- (240:1) -- (300:1) -- (0:1) -- (60:1) -- (120:1);
\draw (240:1) -- (0:1) -- (180:1) -- (300:1);
\draw (60:1) -- (300:1) (120:1) -- (240:1);
\end{tikzpicture}
\end{center}
\end{minipage}}
&
\scalebox{0.6}{
\begin{minipage}{0.3\textwidth}
\begin{center}%3P_1
\begin{tikzpicture}[every node/.style={circle,fill, minimum size=0.07cm}]
\node at (0,0) {};
\node at (1,1.73205080757) {};%sqrt(3)
\node at (2,0) {};
\end{tikzpicture}
\end{center}
\end{minipage}}
\\
\\
$2P_3$ & $C_{2k+1}$ for $k\geq 1$ & $\overline{2P_3}$ & $\overline{C_{2k+1}}$ for $k\geq 1$\\
& ($k=1$ shown) & & ($k=1$ shown)
\end{tabular}
\end{center}
\caption{\label{fig:bip-2P3}The forbidden induced subgraphs for the classes of $2P_3$-free bipartite graphs and $\overline{2P_3}$-free co-bipartite graphs mentioned in Lemma~\ref{lem:bip-2P3}.}
\end{figure}

\newpage
\begin{lemma}\label{lem:bip-2P3}
The class of $2P_3$-free bipartite atoms and the class of $\overline{2P_3}$-free co-bipartite atoms have unbounded clique-width (see \figurename~\ref{fig:bip-2P3} for illustrations of the forbidden induced subgraphs).
\end{lemma}

\begin{proof}
We adapt the construction of the graph~$G_n$, which was used by Lozin and Volz~\cite{LV08} to show that the class of $2P_3$-free bipartite graphs has unbounded clique-width.
For $n \geq 3$, construct the graph~$G_n$ as follows.
Let the vertex set of~$G_n$ be $\{v_{i,j} \; | \; i \in \{0,\ldots,n\}, j \in \{1,\ldots,n\}\} \cup \{w_{i,j} \; | \; i \in \{1,\ldots,n\}, j \in \{0,\ldots,n\}\}$.
For $i,j,k \in \{1,\ldots,n\}$, add an edge between~$v_{i,j}$ and~$w_{k,0}$ if $k \geq i$ and add an edge between $w_{i,j}$ and~$v_{0,k}$ if $k \geq j$.
For each $i,j \in \{1,\ldots,n\}$, add an edge between~$v_{i,j}$ and~$w_{i,j}$ and an edge between~$v_{0,j}$ and~$w_{i,0}$.
Let~$G_n$ be the resulting graph.
Then~$G_n$ is a $2P_3$-free bipartite graph and the family of such graphs has unbounded clique-width~\cite{LV08} (the former can also be seen by inspection and the latter follows from Lemma~\ref{lem:generalunbounded}).
Therefore the class of $2P_3$-free bipartite graphs has unbounded clique-width, and so Fact~\ref{fact:comp} implies that the class of $\overline{2P_3}$-free co-bipartite graphs has unbounded clique-width.

We observe that every graph in $\{\overline{2P_3},\overline{C_3},\overline{C_5},\overline{C_7},\ldots\}$ has no dominating vertex and no two non-adjacent vertices that are complete to the remainder of the graph.
Therefore, by Lemma~\ref{lem:no-comp-P1or2P1}, it follows that the class of $\overline{2P_3}$-free co-bipartite atoms has unbounded clique-width.
Now let~$H_n$ be the graph obtained from~$G_n$ by deleting~$v_{n,n}$ and~$w_{n,n}$ (see \figurename~\ref{fig:L15}) and note that~$H_n$ is a $2P_3$-free bipartite graph.
By Fact~\ref{fact:del-vert}, the family of graphs~$H_n$ has unbounded clique-width.

\begin{figure}[h]
\begin{center}
\includegraphics[scale=0.7,page=1]{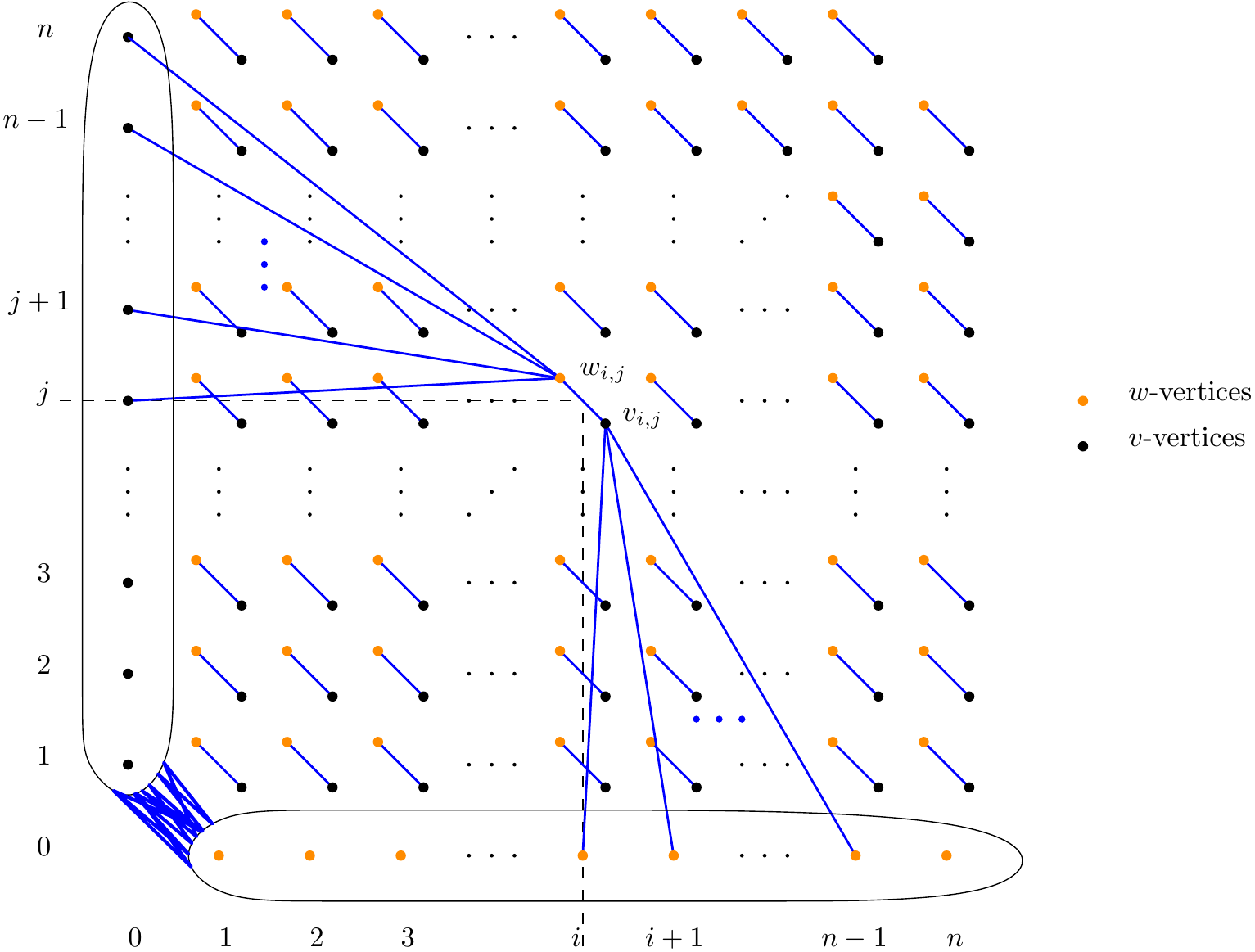}
\end{center}
\caption{The graph~$H_n$ from the proof of Lemma~\ref{lem:bip-2P3}.
For clarity, the edges incident to~$v_{i,j}$ and~$w_{i,j}$ when $i,j \in \{1,\ldots,n\}$ are depicted for only one such pair of vertices.}
\label{fig:L15}
\end{figure}

We now prove that the class of $2P_3$-free bipartite atoms has unbounded clique-width.
Let~$V$ be the set of vertices in the $2P_3$-free bipartite graph~$H_n$ of the form~$v_{i,j}$ and let~$W$ be the set of vertices in~$H_n$ of the form~$w_{i,j}$ and note that~$V$ and~$W$ are independent sets.
Suppose, for contradiction, that~$H_n$ has a clique cut-set~$X$.
Since~$H_n$ is bipartite, every clique cut-set in~$H_n$ contains at most one vertex from each part, so $|X| \leq 2$.
If~$X$ does not contain~$v_{0,n}$, then every vertex in $W\setminus X$ is in the same component of $H_n \setminus X$.
Since every vertex in~$V$ has at least two neighbours in~$W$, and at most one vertex of~$W$ is in~$X$, it follows that every vertex of $V \setminus X$ is in the same component of $H_n \setminus X$, and so~$H_n \setminus X$ is connected.
This contradiction implies that $v_{0,n} \in X$.
By symmetry, $w_{n,0} \in X$.
By construction, deleting~$v_{0,n}$ and~$w_{n,0}$ does not disconnect~$H_n$, so~$H_n$ is indeed an atom.
This completes the proof.
\end{proof}

\begin{figure}[h]
\begin{center}
\begin{tabular}{cccc}
\scalebox{0.6}{
\begin{minipage}{0.3\textwidth}
\begin{center}%K_4
\begin{tikzpicture}[every node/.style={circle,fill, minimum size=0.07cm}]
\node at (0,0) {};
\node at (0,2) {};
\node at (2,0) {};
\node at (2,2) {};
\draw (0,0) -- (0,2) -- (2,2) -- (2,0) -- (0,0) -- (2,2);
\draw (0,2) -- (2,0);
\end{tikzpicture}
\end{center}
\end{minipage}}
&
\scalebox{0.6}{
\begin{minipage}{0.3\textwidth}
\begin{center}%P_1+P_4
\begin{tikzpicture}[every node/.style={circle,fill, minimum size=0.07cm}]
\node at (0,0) {};
\node at (0,2) {};
\node at (2,0) {};
\node at (2,2) {};
\node at (1,1) {};
\draw (2,2) -- (2,0) -- (0,0) -- (0,2);
\end{tikzpicture}
\end{center}
\end{minipage}}
&
\scalebox{0.6}{
\begin{minipage}{0.3\textwidth}
\begin{center}%4P_1
\begin{tikzpicture}[every node/.style={circle,fill, minimum size=0.07cm}]
\node at (0,0) {};
\node at (0,2) {};
\node at (2,0) {};
\node at (2,2) {};
\end{tikzpicture}
\end{center}
\end{minipage}}
&
\scalebox{0.6}{
\begin{minipage}{0.3\textwidth}
\begin{center}%gem
\begin{tikzpicture}[every node/.style={circle,fill, minimum size=0.07cm}]
\node at (0,0) {};
\node at (0,2) {};
\node at (2,0) {};
\node at (2,2) {};
\node at (1,1) {};
\draw (2,2) -- (2,0) -- (0,0) -- (0,2);
\draw (0,0) -- (1,1);
\draw (0,2) -- (1,1);
\draw (2,0) -- (1,1);
\draw (2,2) -- (1,1);
\end{tikzpicture}
\end{center}
\end{minipage}}\\
\\
$K_4$ & $P_1+\nobreak P_4$ & $4P_1$ & $\overline{P_1+P_4}$\\
\end{tabular}
\end{center}
\caption{\label{fig:4P_1-gem}The forbidden induced subgraphs for the classes of $(K_4,P_1+\nobreak P_4)$-free graphs and $(4P_1,\overline{P_1+P_4})$-free graphs mentioned in Lemma~\ref{lem:4P_1-gem}.}
\end{figure}

\begin{lemma}\label{lem:4P_1-gem}
The class of $(K_4,P_1+\nobreak P_4)$-free atoms and the class of $(4P_1,\overline{P_1+P_4})$-free atoms have unbounded clique-width (see \figurename~\ref{fig:4P_1-gem} for illustrations of the forbidden induced subgraphs).
\end{lemma}

\begin{proof}
For this proof we use a construction that is implicit in the proof in~\cite[Theorem~3]{KS12} that {\sc Graph Isomorphism} is {\sc GI}-complete on the class of $(K_4,P_1+\nobreak P_4)$-free graphs; we give an explicit construction.
Consider a $1$-subdivided wall of height $n \geq 2$.
This graph is bipartite; let~$P$ and~$Q$ be the two parts of its bipartition with the vertices in~$Q$ being the vertices added by the subdivision.
Consider a $3$-subdivision of this $1$-subdivided wall (so the resulting graph is a $7$-subdivided wall).
Partition the vertices introduced by this $3$-subdivision as follows: let~$A$ be the set of vertices that are adjacent to vertices of~$P$,
let~$C$ be the set of vertices that are adjacent to vertices of~$Q$, and let~$B$ be the set of remaining vertices introduced by the $3$-subdivision (which have a neighbour in both~$A$ and~$C$.
Apply complementations to $P\cup C$, $Q\cup A$, and~$B$ (these sets will become cliques).
Let~$H_n$ be the resulting graph (see also \figurename~\ref{fig:4P1-gem}) and note that~$\overline{H_n}$ is $(K_4,P_1+\nobreak P_4)$-free and that the family of such graphs has unbounded clique-width~\cite{Sc17} (the former statement can be seen by inspection and the latter can be seen by combining Lemma~\ref{lem:walls} and Fact~\ref{fact:comp}).
Therefore the class of $(K_4,P_1+\nobreak P_4)$-free graphs has unbounded clique-width.
We observe that neither~$K_4$ nor~$P_1+\nobreak P_4$ contains a pair of false twins. 
Therefore, by Lemma~\ref{lem:no-false-twin}, the class of $(K_4,P_1+\nobreak P_4)$-free atoms has unbounded clique-width.

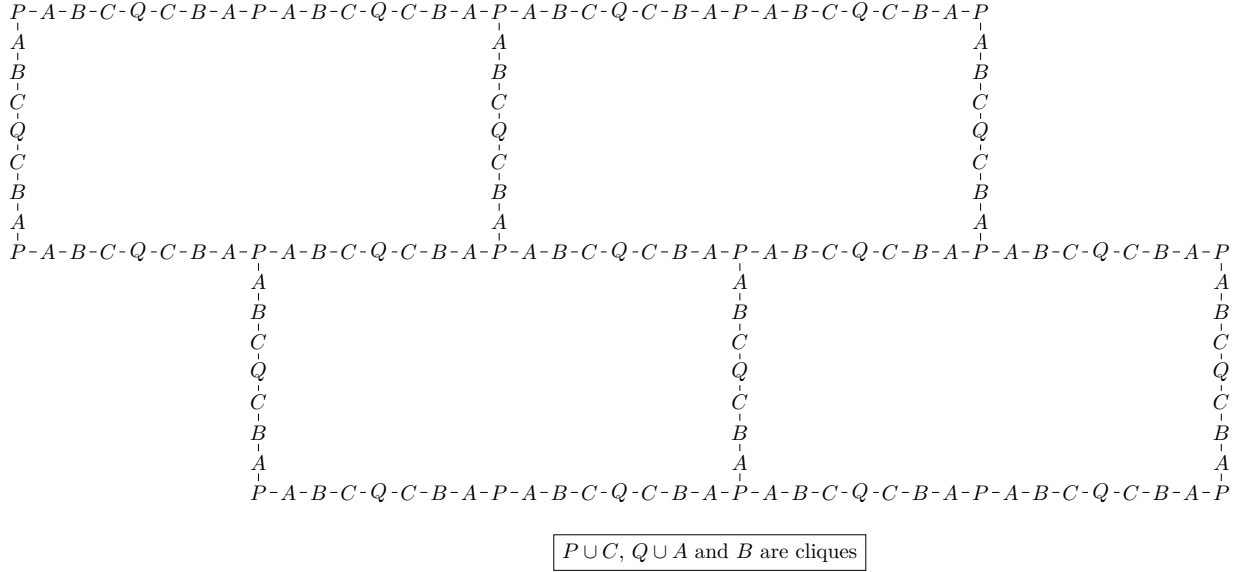
\begin{figure}[h]
\begin{center}
\scalebox{0.8}{
\begin{tikzpicture}[scale=0.5, every node/.style={inner sep=0pt,minimum size=0mm,circle}  ]

\node (p-08-00) at (08,00) {$P$};
\node (a-09-00) at (09,00) {$A$};
\node (b-10-00) at (10,00) {$B$};
\node (c-11-00) at (11,00) {$C$};
\node (q-12-00) at (12,00) {$Q$};
\node (c-13-00) at (13,00) {$C$};
\node (b-14-00) at (14,00) {$B$};
\node (a-15-00) at (15,00) {$A$};
\node (p-16-00) at (16,00) {$P$};
\node (a-17-00) at (17,00) {$A$};
\node (b-18-00) at (18,00) {$B$};
\node (c-19-00) at (19,00) {$C$};
\node (q-20-00) at (20,00) {$Q$};
\node (c-21-00) at (21,00) {$C$};
\node (b-22-00) at (22,00) {$B$};
\node (a-23-00) at (23,00) {$A$};
\node (p-24-00) at (24,00) {$P$};
\node (a-25-00) at (25,00) {$A$};
\node (b-26-00) at (26,00) {$B$};
\node (c-27-00) at (27,00) {$C$};
\node (q-28-00) at (28,00) {$Q$};
\node (c-29-00) at (29,00) {$C$};
\node (b-30-00) at (30,00) {$B$};
\node (a-31-00) at (31,00) {$A$};
\node (p-32-00) at (32,00) {$P$};
\node (a-33-00) at (33,00) {$A$};
\node (b-34-00) at (34,00) {$B$};
\node (c-35-00) at (35,00) {$C$};
\node (q-36-00) at (36,00) {$Q$};
\node (c-37-00) at (37,00) {$C$};
\node (b-38-00) at (38,00) {$B$};
\node (a-39-00) at (39,00) {$A$};
\node (p-40-00) at (40,00) {$P$};

\draw (p-08-00) -- (a-09-00) -- (b-10-00) -- (c-11-00) -- (q-12-00) -- (c-13-00) -- (b-14-00) -- (a-15-00) -- (p-16-00) -- (a-17-00) -- (b-18-00) -- (c-19-00) -- (q-20-00) -- (c-21-00) -- (b-22-00) -- (a-23-00) -- (p-24-00) -- (a-25-00) -- (b-26-00) -- (c-27-00) -- (q-28-00) -- (c-29-00) -- (b-30-00) -- (a-31-00) -- (p-32-00) -- (a-33-00) -- (b-34-00) -- (c-35-00) -- (q-36-00) -- (c-37-00) -- (b-38-00) -- (a-39-00) -- (p-40-00);

\node (p-00-08) at (00,08) {$P$};
\node (a-01-08) at (01,08) {$A$};
\node (b-02-08) at (02,08) {$B$};
\node (c-03-08) at (03,08) {$C$};
\node (q-04-08) at (04,08) {$Q$};
\node (c-05-08) at (05,08) {$C$};
\node (b-06-08) at (06,08) {$B$};
\node (a-07-08) at (07,08) {$A$};
\node (p-08-08) at (08,08) {$P$};
\node (a-09-08) at (09,08) {$A$};
\node (b-10-08) at (10,08) {$B$};
\node (c-11-08) at (11,08) {$C$};
\node (q-12-08) at (12,08) {$Q$};
\node (c-13-08) at (13,08) {$C$};
\node (b-14-08) at (14,08) {$B$};
\node (a-15-08) at (15,08) {$A$};
\node (p-16-08) at (16,08) {$P$};
\node (a-17-08) at (17,08) {$A$};
\node (b-18-08) at (18,08) {$B$};
\node (c-19-08) at (19,08) {$C$};
\node (q-20-08) at (20,08) {$Q$};
\node (c-21-08) at (21,08) {$C$};
\node (b-22-08) at (22,08) {$B$};
\node (a-23-08) at (23,08) {$A$};
\node (p-24-08) at (24,08) {$P$};
\node (a-25-08) at (25,08) {$A$};
\node (b-26-08) at (26,08) {$B$};
\node (c-27-08) at (27,08) {$C$};
\node (q-28-08) at (28,08) {$Q$};
\node (c-29-08) at (29,08) {$C$};
\node (b-30-08) at (30,08) {$B$};
\node (a-31-08) at (31,08) {$A$};
\node (p-32-08) at (32,08) {$P$};
\node (a-33-08) at (33,08) {$A$};
\node (b-34-08) at (34,08) {$B$};
\node (c-35-08) at (35,08) {$C$};
\node (q-36-08) at (36,08) {$Q$};
\node (c-37-08) at (37,08) {$C$};
\node (b-38-08) at (38,08) {$B$};
\node (a-39-08) at (39,08) {$A$};
\node (p-40-08) at (40,08) {$P$};

\draw (p-00-08) -- (a-01-08) -- (b-02-08) -- (c-03-08) -- (q-04-08) -- (c-05-08) -- (b-06-08) -- (a-07-08) -- (p-08-08) -- (a-09-08) -- (b-10-08) -- (c-11-08) -- (q-12-08) -- (c-13-08) -- (b-14-08) -- (a-15-08) -- (p-16-08) -- (a-17-08) -- (b-18-08) -- (c-19-08) -- (q-20-08) -- (c-21-08) -- (b-22-08) -- (a-23-08) -- (p-24-08) -- (a-25-08) -- (b-26-08) -- (c-27-08) -- (q-28-08) -- (c-29-08) -- (b-30-08) -- (a-31-08) -- (p-32-08) -- (a-33-08) -- (b-34-08) -- (c-35-08) -- (q-36-08) -- (c-37-08) -- (b-38-08) -- (a-39-08) -- (p-40-08);

\node (p-00-16) at (00,16) {$P$};
\node (a-01-16) at (01,16) {$A$};
\node (b-02-16) at (02,16) {$B$};
\node (c-03-16) at (03,16) {$C$};
\node (q-04-16) at (04,16) {$Q$};
\node (c-05-16) at (05,16) {$C$};
\node (b-06-16) at (06,16) {$B$};
\node (a-07-16) at (07,16) {$A$};
\node (p-08-16) at (08,16) {$P$};
\node (a-09-16) at (09,16) {$A$};
\node (b-10-16) at (10,16) {$B$};
\node (c-11-16) at (11,16) {$C$};
\node (q-12-16) at (12,16) {$Q$};
\node (c-13-16) at (13,16) {$C$};
\node (b-14-16) at (14,16) {$B$};
\node (a-15-16) at (15,16) {$A$};
\node (p-16-16) at (16,16) {$P$};
\node (a-17-16) at (17,16) {$A$};
\node (b-18-16) at (18,16) {$B$};
\node (c-19-16) at (19,16) {$C$};
\node (q-20-16) at (20,16) {$Q$};
\node (c-21-16) at (21,16) {$C$};
\node (b-22-16) at (22,16) {$B$};
\node (a-23-16) at (23,16) {$A$};
\node (p-24-16) at (24,16) {$P$};
\node (a-25-16) at (25,16) {$A$};
\node (b-26-16) at (26,16) {$B$};
\node (c-27-16) at (27,16) {$C$};
\node (q-28-16) at (28,16) {$Q$};
\node (c-29-16) at (29,16) {$C$};
\node (b-30-16) at (30,16) {$B$};
\node (a-31-16) at (31,16) {$A$};
\node (p-32-16) at (32,16) {$P$};

\draw (p-00-16) -- (a-01-16) -- (b-02-16) -- (c-03-16) -- (q-04-16) -- (c-05-16) -- (b-06-16) -- (a-07-16) -- (p-08-16) -- (a-09-16) -- (b-10-16) -- (c-11-16) -- (q-12-16) -- (c-13-16) -- (b-14-16) -- (a-15-16) -- (p-16-16) -- (a-17-16) -- (b-18-16) -- (c-19-16) -- (q-20-16) -- (c-21-16) -- (b-22-16) -- (a-23-16) -- (p-24-16) -- (a-25-16) -- (b-26-16) -- (c-27-16) -- (q-28-16) -- (c-29-16) -- (b-30-16) -- (a-31-16) -- (p-32-16);

\node (a-00-15) at (00,15) {$A$};
\node (b-00-14) at (00,14) {$B$};
\node (c-00-13) at (00,13) {$C$};
\node (q-00-12) at (00,12) {$Q$};
\node (c-00-11) at (00,11) {$C$};
\node (b-00-10) at (00,10) {$B$};
\node (a-00-09) at (00,09) {$A$};

\draw (p-00-16) -- (a-00-15) -- (b-00-14) -- (c-00-13) -- (q-00-12) -- (c-00-11) -- (b-00-10) -- (a-00-09) -- (p-00-08);

\node (a-16-15) at (16,15) {$A$};
\node (b-16-14) at (16,14) {$B$};
\node (c-16-13) at (16,13) {$C$};
\node (q-16-12) at (16,12) {$Q$};
\node (c-16-11) at (16,11) {$C$};
\node (b-16-10) at (16,10) {$B$};
\node (a-16-09) at (16,09) {$A$};

\draw (p-16-16) -- (a-16-15) -- (b-16-14) -- (c-16-13) -- (q-16-12) -- (c-16-11) -- (b-16-10) -- (a-16-09) -- (p-16-08);

\node (a-32-15) at (32,15) {$A$};
\node (b-32-14) at (32,14) {$B$};
\node (c-32-13) at (32,13) {$C$};
\node (q-32-12) at (32,12) {$Q$};
\node (c-32-11) at (32,11) {$C$};
\node (b-32-10) at (32,10) {$B$};
\node (a-32-09) at (32,09) {$A$};

\draw (p-32-16) -- (a-32-15) -- (b-32-14) -- (c-32-13) -- (q-32-12) -- (c-32-11) -- (b-32-10) -- (a-32-09) -- (p-32-08);

\node (a-08-07) at (08,07) {$A$};
\node (b-08-06) at (08,06) {$B$};
\node (c-08-05) at (08,05) {$C$};
\node (q-08-04) at (08,04) {$Q$};
\node (c-08-03) at (08,03) {$C$};
\node (b-08-02) at (08,02) {$B$};
\node (a-08-01) at (08,01) {$A$};

\draw (p-08-08) -- (a-08-07) -- (b-08-06) -- (c-08-05) -- (q-08-04) -- (c-08-03) -- (b-08-02) -- (a-08-01) -- (p-08-00);

\node (a-24-07) at (24,07) {$A$};
\node (b-24-06) at (24,06) {$B$};
\node (c-24-05) at (24,05) {$C$};
\node (q-24-04) at (24,04) {$Q$};
\node (c-24-03) at (24,03) {$C$};
\node (b-24-02) at (24,02) {$B$};
\node (a-24-01) at (24,01) {$A$};

\draw (p-24-08) -- (a-24-07) -- (b-24-06) -- (c-24-05) -- (q-24-04) -- (c-24-03) -- (b-24-02) -- (a-24-01) -- (p-24-00);

\node (a-40-07) at (40,07) {$A$};
\node (b-40-06) at (40,06) {$B$};
\node (c-40-05) at (40,05) {$C$};
\node (q-40-04) at (40,04) {$Q$};
\node (c-40-03) at (40,03) {$C$};
\node (b-40-02) at (40,02) {$B$};
\node (a-40-01) at (40,01) {$A$};

\draw (p-40-08) -- (a-40-07) -- (b-40-06) -- (c-40-05) -- (q-40-04) -- (c-40-03) -- (b-40-02) -- (a-40-01) -- (p-40-00);

    \node[draw=black, rectangle, align=right, inner sep=4pt] (label) at (23,-2) {$P\cup C$, $Q\cup A$ and $B$ are cliques};

\end{tikzpicture}}
\caption{The graph~$H_n$ from the proof of Lemma~\ref{lem:4P_1-gem} ($n=2$ shown).}
\label{fig:4P1-gem}

\end{center}
\end{figure}

We now prove that the class of $(4P_1,\overline{P_1+P_4})$-free atoms has unbounded clique-width.
As~$\overline{H_n}$ is $(K_4,P_1+\nobreak P_4)$-free, it follows that~$H_n$ is $(4P_1,\overline{P_1+P_4})$-free.
By Fact~\ref{fact:comp}, it follows that the family of graphs~$H_n$ also has unbounded clique-width.
It remains to show that~$H_n$ is an atom.
Suppose, for contradiction, that~$H_n$ has a clique cut-set~$X$.
Recall that $P\cup C$, $Q\cup A$, and~$B$ are cliques.
As~$A$ and $C$ are anti-complete to each other, it follows that~$X$ contains vertices from at most one of~$A$ or~$C$.
Similarly, since~$P$, $Q$ and~$B$ are pairwise anti-complete, it follows that~$X$ contains vertices from at most one of~$P$, $Q$ or~$B$.
Note that every vertex from $P\cup Q\cup B$ has a neighbour in~$A$ and in~$C$.
Since~$A$ and~$C$ are cliques and~$X$ contains vertices in at most one of these sets, it follows that the vertices in $(P\cup Q\cup B)\setminus X$ all lie in the same component of $H_n\setminus X$.
Similarly, every vertex of $A \cup C$ has a neighbour in both~$P$ and~$Q$.
Since~$X$ contains vertices from at most one of~$P$ or~$Q$, it follows that the vertices in $(A\cup C)\setminus X$ are in the same component of $H_n \setminus X$ as the vertices of $(P\cup Q\cup B)\setminus X$.
Therefore $H_n \setminus X$ is connected.
This contradiction implies that~$H_n$ is indeed an atom.
This completes the proof.
\end{proof}

\begin{figure}[h]
\begin{center}
\begin{tabular}{cccc}
\scalebox{0.6}{
\begin{minipage}{0.3\textwidth}
\begin{center}%4P_1
\begin{tikzpicture}[every node/.style={circle,fill, minimum size=0.07cm}]
\node at (0,0) {};
\node at (0,2) {};
\node at (2,0) {};
\node at (2,2) {};
\end{tikzpicture}
\end{center}
\end{minipage}}
&
\scalebox{0.6}{
\begin{minipage}{0.3\textwidth}
\begin{center}%\overline{3P_1+P_2}
\begin{tikzpicture}[every node/.style={circle,fill, minimum size=0.07cm}]
\node at (18:1) {};
\node at (90:1) {};
\node at (162:1) {};
\node at (234:1) {};
\node at (306:1) {};
\draw (18:1) -- (90:1) -- (162:1) -- (234:1) (306:1) -- (18:1);
\draw (18:1) -- (162:1) -- (306:1) -- (90:1) -- (234:1) -- (18:1);
\end{tikzpicture}
\end{center}
\end{minipage}}
&
\scalebox{0.6}{
\begin{minipage}{0.3\textwidth}
\begin{center}%K_4
\begin{tikzpicture}[every node/.style={circle,fill, minimum size=0.07cm}]
\node at (0,0) {};
\node at (0,2) {};
\node at (2,0) {};
\node at (2,2) {};
\draw (0,0) -- (0,2) -- (2,2) -- (2,0) -- (0,0) -- (2,2);
\draw (0,2) -- (2,0);
\end{tikzpicture}
\end{center}
\end{minipage}}
&
\scalebox{0.6}{
\begin{minipage}{0.3\textwidth}
\begin{center}%3P_1+P_2
\begin{tikzpicture}[every node/.style={circle,fill, minimum size=0.07cm}]
\node at (18:1) {};
\node at (90:1) {};
\node at (162:1) {};
\node at (234:1) {};
\node at (306:1) {};
\draw (234:1) -- (306:1);
\end{tikzpicture}
\end{center}
\end{minipage}}\\
\\
$4P_1$ & $\overline{3P_1+P_2}$ & $\overline{4P_1}$ & $3P_1+\nobreak P_2$\\
\end{tabular}
\end{center}
\caption{\label{fig:K4-3P1P2}The forbidden induced subgraphs for the classes of $(4P_1,\overline{3P_1+P_2})$-free graphs and $(\overline{4P_1},\allowbreak 3P_1+\nobreak P_2)$-free graphs mentioned in Lemma~\ref{lem:K4-3P1P2}.}
\end{figure}
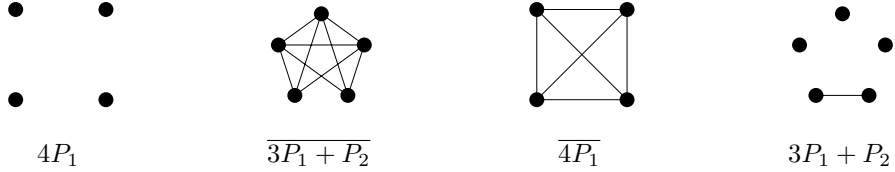

\begin{lemma}\label{lem:K4-3P1P2}
The class of $(4P_1,\overline{3P_1+P_2})$-free atoms and the class of $(\overline{4P_1},\allowbreak 3P_1+\nobreak P_2)$-free atoms have unbounded clique-width (see \figurename~\ref{fig:K4-3P1P2} for illustrations of the forbidden induced subgraphs).
\end{lemma}
\begin{proof}
We use the construction of~\cite{DGP14} for proving that the class of $(4P_1,\overline{3P_1+P_2})$-free graphs has unbounded clique-width.
Let $n \geq 7$ and consider an $n \times n$ grid~$H_n$ and for $i,j \in \{0,\ldots,n-1\}$, let~$v_{i,j}$ be the vertex of~$H_n$ with $x$-coordinate~$i$ and $y$-coordinate~$j$.
For $k \in \{0,1,2\}$, let $V_k=\{v_{i,j} \; | \; i+j\equiv k \bmod 3\}$ (see also \figurename~\ref{fig:4P1-co3P1P2-free} for a depiction of this $3$-colouring).
Apply a complementation to each~$V_k$.
Let~$G_n$ be the resulting graph.
The resulting graph~$G_n$ is $(4P_1,\overline{3P_1+P_2})$-free and the family of graphs~$G_n$ has unbounded clique-width~\cite{DGP14} (the first of these statements can also be seen by inspection and the latter follows from combining Lemma~\ref{lem:generalunbounded} and Fact~\ref{fact:comp}.
By Fact~\ref{fact:comp}, it follows that the family of graphs~$\overline{G_n}$ also has unbounded clique-width.

\begin{figure}[h]
\begin{center}
\begin{tikzpicture}[scale=0.7,yscale=-1]
\path
      (0,4) node[](x0y0) {0}
      (1,4) node[](x1y0) {1}
      (2,4) node[](x2y0) {2}
      (3,4) node[](x3y0) {0}
      (4,4) node[](x4y0) {1}
      (5,4) node[](x5y0) {2}
      (6,4) node[](x6y0) {0}

      (0,3) node[](x0y1) {1}
      (1,3) node[](x1y1) {2}
      (2,3) node[](x2y1) {0}
      (3,3) node[](x3y1) {1}
      (4,3) node[](x4y1) {2}
      (5,3) node[](x5y1) {0}
      (6,3) node[](x6y1) {1}

      (0,2) node[](x0y2) {2}
      (1,2) node[](x1y2) {0}
      (2,2) node[](x2y2) {1}
      (3,2) node[](x3y2) {2}
      (4,2) node[](x4y2) {0}
      (5,2) node[](x5y2) {1}
      (6,2) node[](x6y2) {2}

      (0,1) node[](x0y3) {0}
      (1,1) node[](x1y3) {1}
      (2,1) node[](x2y3) {2}
      (3,1) node[](x3y3) {0}
      (4,1) node[](x4y3) {1}
      (5,1) node[](x5y3) {2}
      (6,1) node[](x6y3) {0}
                           
      (0,0) node[](x0y4) {1}
      (1,0) node[](x1y4) {2}
      (2,0) node[](x2y4) {0}
      (3,0) node[](x3y4) {1}
      (4,0) node[](x4y4) {2}
      (5,0) node[](x5y4) {0}
      (6,0) node[](x6y4) {1}

      (0,-1) node[](x0y5) {2}
      (1,-1) node[](x1y5) {0}
      (2,-1) node[](x2y5) {1}
      (3,-1) node[](x3y5) {2}
      (4,-1) node[](x4y5) {0}
      (5,-1) node[](x5y5) {1}
      (6,-1) node[](x6y5) {2}

      (0,-2) node[](x0y6) {0}
      (1,-2) node[](x1y6) {1}
      (2,-2) node[](x2y6) {2}
      (3,-2) node[](x3y6) {0}
      (4,-2) node[](x4y6) {1}
      (5,-2) node[](x5y6) {2}
      (6,-2) node[](x6y6) {0}
;

\draw (x0y0) -- (x1y0) -- (x2y0) -- (x3y0) -- (x4y0) -- (x5y0) -- (x6y0) ;
\draw (x0y1) -- (x1y1) -- (x2y1) -- (x3y1) -- (x4y1) -- (x5y1) -- (x6y1) ;
\draw (x0y2) -- (x1y2) -- (x2y2) -- (x3y2) -- (x4y2) -- (x5y2) -- (x6y2) ;
\draw (x0y3) -- (x1y3) -- (x2y3) -- (x3y3) -- (x4y3) -- (x5y3) -- (x6y3) ;
\draw (x0y4) -- (x1y4) -- (x2y4) -- (x3y4) -- (x4y4) -- (x5y4) -- (x6y4) ;
\draw (x0y5) -- (x1y5) -- (x2y5) -- (x3y5) -- (x4y5) -- (x5y5) -- (x6y5) ;
\draw (x0y6) -- (x1y6) -- (x2y6) -- (x3y6) -- (x4y6) -- (x5y6) -- (x6y6) ;

\draw (x0y0) -- (x0y1) -- (x0y2) -- (x0y3) -- (x0y4) -- (x0y5) -- (x0y6) ;
\draw (x1y0) -- (x1y1) -- (x1y2) -- (x1y3) -- (x1y4) -- (x1y5) -- (x1y6) ;
\draw (x2y0) -- (x2y1) -- (x2y2) -- (x2y3) -- (x2y4) -- (x2y5) -- (x2y6) ;
\draw (x3y0) -- (x3y1) -- (x3y2) -- (x3y3) -- (x3y4) -- (x3y5) -- (x3y6) ;
\draw (x4y0) -- (x4y1) -- (x4y2) -- (x4y3) -- (x4y4) -- (x4y5) -- (x4y6) ;
\draw (x5y0) -- (x5y1) -- (x5y2) -- (x5y3) -- (x5y4) -- (x5y5) -- (x5y6) ;
\draw (x6y0) -- (x6y1) -- (x6y2) -- (x6y3) -- (x6y4) -- (x6y5) -- (x6y6) ;

\node[draw=black, rectangle, align=left, inner sep=4pt] (label) at (10,1) {\\ For $i\in\{0,1,2\}$ the\\ $i$-vertices form a clique};
 
 \foreach \x in {0,1,2,3,4,5,6}{
    \node  at (\x,4.7)  {\tiny{\x}};
          }
    \node  at (-0.7,-2)  {\tiny{6}};
    \node  at (-0.7,-1)  {\tiny{5}};
    \node  at (-0.7,0)  {\tiny{4}};
    \node  at (-0.7,1)  {\tiny{3}};
    \node  at (-0.7,2)  {\tiny{2}};
    \node  at (-0.7,3)  {\tiny{1}};
    \node  at (-0.7,4)  {\tiny{0}};

\end{tikzpicture}
\caption{The graph~$G_n$ from the proof of Lemma~\ref{lem:K4-3P1P2} ($n=7$ shown).}
\label{fig:4P1-co3P1P2-free}
\end{center}
\end{figure}
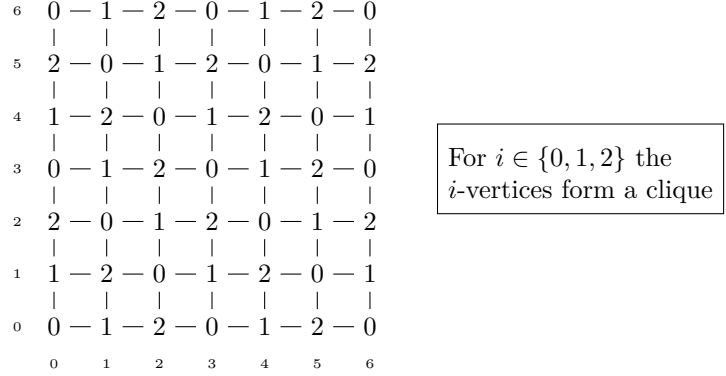

It remains to show that~$G_n$ and~$\overline{G_n}$ are atoms.
Suppose, for contradiction, that~$G_n$ has a clique cut-set~$X$.
If $X \subseteq V_i$ for some $i \in \{0,1,2\}$, then all vertices of $G_n \setminus V_i$ are in the same component of $G_n \setminus X$.
Since every vertex in~$V_i$ has at least one neighbour outside of~$V_i$, it follows that every vertex of $G_n \setminus X$ is in the same component of $G_n \setminus X$ in this case, a contradiction.
We may therefore assume that~$X$ contains vertices in at least two sets~$V_i$.
By construction, each vertex in a set~$V_i$ has at most two neighbours in each~$V_j$ for $j \in \{0,1,2\} \setminus \{i\}$.
Therefore~$X$ has at most two vertices in each~$V_i$.
Since $n \geq 7$, there must be at least~15 vertices in~$V_0$ that have neighbours in both~$V_1$ and~$V_2$ (see also \figurename~\ref{fig:4P1-co3P1P2-free}).
Since very vertex in $V_1 \cup V_2$ has at most two neighbours in~$V_0$, there must be a vertex in $V_0\setminus X$ that has a neighbour in both $V_1 \setminus X$ and $V_2 \setminus X$.
Since each set~$V_i$ is a clique, it follows that $G_n \setminus X$ is connected.
This contradiction implies that~$G_n$ is indeed an atom.
Now suppose, for contradiction, that~$\overline{G_n}$ has a clique cut-set~$X$.
Since $V_0$, $V_1$ and~$V_2$ are independent sets in~$\overline{G_n}$, $X$ contains at most one vertex of any~$V_i$.
Since every vertex of~$V_i$ has at most two non-neighbours in each~$V_j$ for $j \in \{0,1,2\} \setminus \{i\}$, it follows that $\overline{G_n} \setminus X$ must be connected.
This contradiction implies that~$\overline{G_n}$ is indeed an atom.
\end{proof}

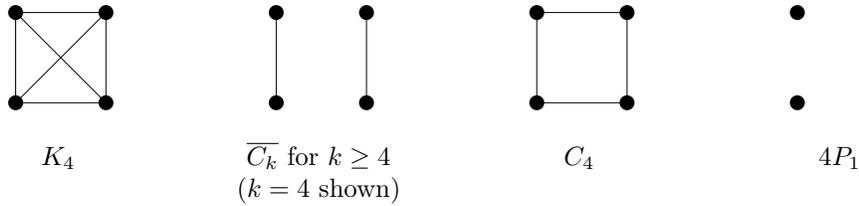
\begin{figure}[h]
\begin{center}
\begin{tabular}{cccc}
\scalebox{0.6}{
\begin{minipage}{0.3\textwidth}
\begin{center}%K_4
\begin{tikzpicture}[every node/.style={circle,fill, minimum size=0.07cm}]
\node at (0,0) {};
\node at (0,2) {};
\node at (2,0) {};
\node at (2,2) {};
\draw (0,0) -- (0,2) -- (2,2) -- (2,0) -- (0,0) -- (2,2);
\draw (0,2) -- (2,0);
\end{tikzpicture}
\end{center}
\end{minipage}}
&
\scalebox{0.6}{
\begin{minipage}{0.3\textwidth}
\begin{center}%2P_2
\begin{tikzpicture}[every node/.style={circle,fill, minimum size=0.07cm}]
\node at (0,0) {};
\node at (0,2) {};
\node at (2,0) {};
\node at (2,2) {};
\draw (0,0) -- (0,2);
\draw (2,0) -- (2,2);
\end{tikzpicture}
\end{center}
\end{minipage}}
&
\scalebox{0.6}{
\begin{minipage}{0.3\textwidth}
\begin{center}%C_4
\begin{tikzpicture}[every node/.style={circle,fill, minimum size=0.07cm}]
\node at (0,0) {};
\node at (0,2) {};
\node at (2,0) {};
\node at (2,2) {};
\draw (0,0) -- (0,2) -- (2,2) -- (2,0) -- (0,0);
\end{tikzpicture}
\end{center}
\end{minipage}}
&
\scalebox{0.6}{
\begin{minipage}{0.3\textwidth}
\begin{center}%4P_1
\begin{tikzpicture}[every node/.style={circle,fill, minimum size=0.07cm}]
\node at (0,0) {};
\node at (0,2) {};
\node at (2,0) {};
\node at (2,2) {};
\end{tikzpicture}
\end{center}
\end{minipage}}\\
\\
$K_4$ & $\overline{C_k}$ for $k \geq 4$ & $C_4$& $4P_1$\\
&($k=4$ shown)&&
\end{tabular}
\end{center}
\caption{\label{fig:K4-2P2}The forbidden induced subgraphs for the classes of $K_4$-free co-chordal graphs and $(C_4,4P_1)$-free graphs mentioned in Lemma~\ref{lem:K4-2P2}.}
\end{figure}

\begin{lemma}\label{lem:K4-2P2}
The class of $K_4$-free co-chordal atoms and the class of $(C_4,4P_1)$-free atoms have unbounded clique-width (see \figurename~\ref{fig:K4-2P2} for illustrations of the forbidden induced subgraphs).
\end{lemma}

\begin{proof}
In~\cite[Theorem~11]{BELL06}, Brandst\"adt et al. constructed a family of graphs~$G_n$ that are $K_4$-free co-chordal and have unbounded clique-width.
The construction of~$G_n$ for $n \geq 3$ is as follows.
Let the vertex set of~$G_n$ be $\{v_{i,j} \; | \; i,j \in \{0,\ldots,n\}, (i,j) \neq (0,0)\}$.
For $i,j,k \in \{1,\ldots,n\}$, add an edge between~$v_{i,j}$ and~$v_{k,0}$ if $k \geq i$ and add an edge between~$v_{i,j}$ and~$v_{0,k}$ if $k \geq j$.
For each $i,j \in \{1,\ldots,n\}$, add an edge between~$v_{i,0}$ and~$v_{0,j}$.
As shown in the proof of~\cite[Theorem~11]{BELL06}, $G_n$ is a $K_4$-free co-chordal graph and the family of such graphs has unbounded clique-width (the former property can also be seen by inspection and the latter follows from Lemma~\ref{lem:generalunbounded}).
Therefore the class of $K_4$-free co-chordal graphs had unbounded clique-width.
We observe that neither~$K_4$ nor~$\overline{C_r}$ for any $r\geq 4$ contains a pair of false twins.
Therefore, by Lemma~\ref{lem:no-false-twin}, the class of $K_4$-free co-chordal atoms has unbounded clique-width.

We now prove that the class of $(C_4,4P_1)$-free atoms has unbounded clique-width.
Observe that the family of graphs of the form~$\overline{G_n}$ is $4P_1$-free and chordal, and by Fact~\ref{fact:comp}, it has unbounded clique-width.
Now~$\overline{G_n}$ is not an atom, since the set of vertices $\{v_{i,j} \; | \; i,j \in \{1,\ldots,n\}\}$ is a clique cut-set in~$\overline{G_n}$.
We construct a graph~$J_n$ from~$\overline{G_{n+1}}$ as follows  (see also \figurename~\ref{fig:L18}).
Delete the vertices in the set $\{v_{1,i},v_{i,1}\;|\; i \in \{1,\ldots,n+1\}$ and add the edge~$v_{0,n+1}v_{n+1,0}$.
Let~$J_n$ be the resulting graph.
Now~$J_n$ contains~$\overline{G_{n-1}}$ as an induced subgraph, so the family of graphs of the form~$J_n$ has unbounded clique-width.
We claim that~$J_n$ is $(C_4,4P_1)$-free.
Note that $J_n \setminus \{v_{0,n+1}\}$ and $J_n \setminus \{v_{n+1,0}\}$ are induced subgraphs of~$\overline{G_{n+1}}$, which is $(C_4,4P_1)$-free.
Therefore we only need to verify that there is no induced~$4P_1$ or~$C_4$ in~$J_n$ that contains both~$v_{0,n+1}$ and~$v_{n+1,0}$.
Since~$v_{0,n+1}$ is adjacent to~$v_{n+1,0}$, there cannot be an induced~$4P_1$ in~$J_n$ that contains both these vertices.
Now $N(v_{0,n+1})=\{v_{n+1,0}\} \cup \{v_{0,1},\ldots,v_{0,n}\}$ and $N(v_{n+1,0})=\{v_{0,n+1}\} \cup \{v_{1,0},\ldots,v_{n,0}\}$.
Since no vertex in $\{v_{0,1},\ldots,v_{0,n}\}$ has a neighbour in $\{v_{1,0},\ldots,v_{n,0}\}$ in~$J_n$, it follows that~$J_n$ is indeed $C_4$-free.

\begin{figure}[h]
\begin{center}
\includegraphics[scale=0.55,page=3]{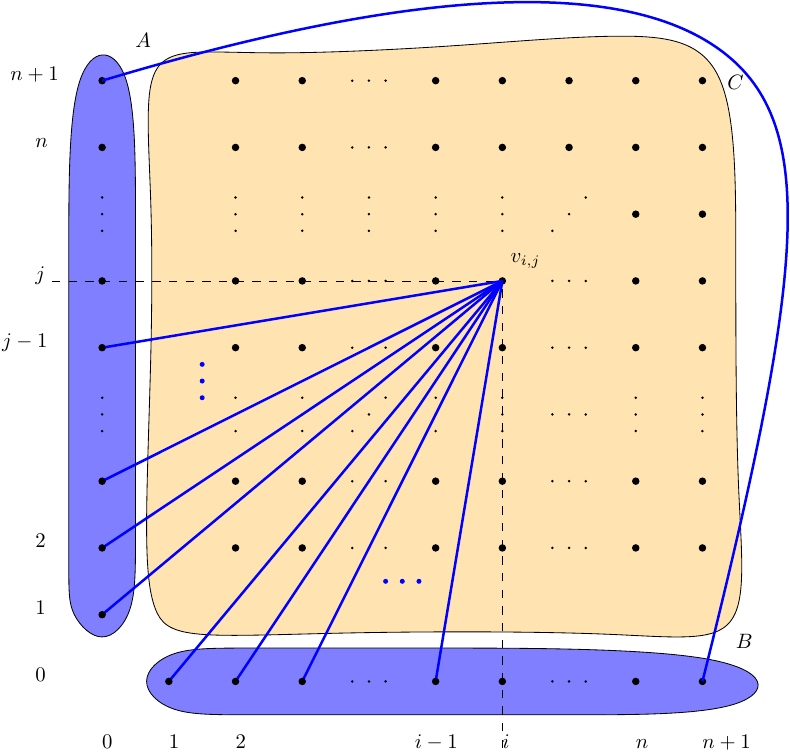}
\end{center}
\caption{The graph~$J_n$ from the proof of Lemma~\ref{lem:K4-2P2}.
The sets~$A$, $B$ and~$C$ are cliques.
For clarity, the edges between $A\cup B$ and~$C$ are depicted for only one vertex in~$C$.}
\label{fig:L18}
\end{figure}

It remains to show that~$J_n$ is an atom.
Suppose, for contradiction, that~$X$ is a clique cut-set of~$J_n$.
First, suppose that $v_{0,n+1}$ is in~$X$. 
Then $v_{1,0}\notin X$ and $v_{n+1,n+1}\notin X$ as~$v_{0,n+1}$ is non-adjacent to these vertices.
As~$v_{1,0}$ and~$v_{n+1,n+1}$ are adjacent and every vertex in $J_n\setminus\{v_{0,n+1}\}$ is adjacent to at least one of these vertices, we find that~$X$ is not a clique cut-set. 
We may therefore assume that $v_{0,n+1} \notin X$ and by symmetry, that~$v_{n+1,0} \notin X$.
We partition the vertices~$v_{i,j}$ in~$J_n$ into three sets~$A$, $B$ and~$C$, if $j=0$, $i=0$ or $i,j \neq 0$, respectively, and note that each of these sets is a clique.
Since~$A$ and~$B$ are cliques and the vertices~$v_{0,n+1}$ and~$v_{n+1,0}$ are adjacent, it follows that all vertices in $(A\cup B)\setminus X$ are in the same component of~$J_n \setminus X$.
Note that every vertex from~$C$ has at least one neighbour in both~$A$ and~$B$.
However, $X$ cannot contain vertices from both~$A$ and~$B$ since $A\setminus\{v_{0,n+1}\}$ and $B\setminus\{v_{n+1,0}\}$ are anti-complete.
Therefore $J_n \setminus X$ is connected.
This contradiction implies that~$J_n$ is an atom.
\end{proof}

\begin{figure}[h]
\begin{center}
\begin{tabular}{cccc}
\scalebox{0.6}{
\begin{minipage}{0.3\textwidth}
\begin{center}%gem
\begin{tikzpicture}[every node/.style={circle,fill, minimum size=0.07cm}]
\node at (0,0) {};
\node at (0,2) {};
\node at (2,0) {};
\node at (2,2) {};
\node at (1,1) {};
\draw (2,2) -- (2,0) -- (0,0) -- (0,2);
\draw (0,0) -- (1,1);
\draw (0,2) -- (1,1);
\draw (2,0) -- (1,1);
\draw (2,2) -- (1,1);
\end{tikzpicture}
\end{center}
\end{minipage}}
&
\scalebox{0.6}{
\begin{minipage}{0.3\textwidth}
\begin{center}%P_1+2P_2
\begin{tikzpicture}[every node/.style={circle,fill, minimum size=0.07cm}]
\node at (0,0) {};
\node at (0,2) {};
\node at (2,0) {};
\node at (2,2) {};
\node at (1,1) {};
\draw (0,0) -- (0,2)  (2,2) -- (2,0);
\end{tikzpicture}
\end{center}
\end{minipage}}
&
\scalebox{0.6}{
\begin{minipage}{0.3\textwidth}
\begin{center}%P_1+P_4
\begin{tikzpicture}[every node/.style={circle,fill, minimum size=0.07cm}]
\node at (0,0) {};
\node at (0,2) {};
\node at (2,0) {};
\node at (2,2) {};
\node at (1,1) {};
\draw (2,2) -- (2,0) -- (0,0) -- (0,2);
\end{tikzpicture}
\end{center}
\end{minipage}}
&
\scalebox{0.6}{
\begin{minipage}{0.3\textwidth}
\begin{center}%\overline{P_1+2P_2}
\begin{tikzpicture}[every node/.style={circle,fill, minimum size=0.07cm}]
\node at (0,0) {};
\node at (0,2) {};
\node at (2,0) {};
\node at (2,2) {};
\node at (1,1) {};
\draw (0,0) -- (0,2) -- (2,2) -- (2,0) -- (0,0);
\draw (0,0) -- (1,1);
\draw (0,2) -- (1,1);
\draw (2,0) -- (1,1);
\draw (2,2) -- (1,1);
\end{tikzpicture}
\end{center}
\end{minipage}}\\
\\
$\overline{P_1+P_4}$ & $P_1+\nobreak 2P_2$ & $P_1+\nobreak P_4$ & $\overline{P_1+2P_2}$\\
\end{tabular}
\end{center}
\caption{\label{fig:gem-P12P2}The forbidden induced subgraphs for the classes of $(\overline{P_1+P_4},P_1+\nobreak 2P_2)$-free graphs and $(P_1+\nobreak P_4,\allowbreak \overline{P_1+2P_2})$-free graphs mentioned in Lemma~\ref{lem:gem-P12P2}.}
\end{figure}
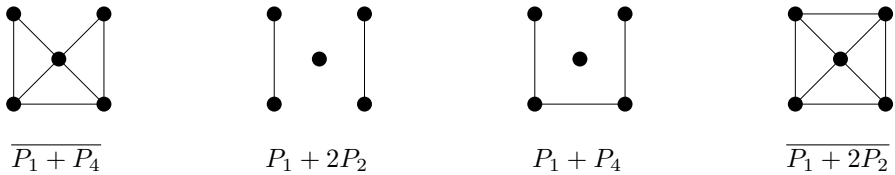

\newpage
\begin{lemma}\label{lem:gem-P12P2}
The class of $(\overline{P_1+P_4},P_1+\nobreak 2P_2)$-free atoms and the class of $(P_1+\nobreak P_4,\allowbreak \overline{P_1+2P_2})$-free atoms have unbounded clique-width (see \figurename~\ref{fig:gem-P12P2} for illustrations of the forbidden induced subgraphs).
\end{lemma}

\begin{proof}
We use the construction from~\cite{BDJP18}, which was used to show that $(\overline{P_1+P_4},P_1+\nobreak 2P_2)$-free graphs have unbounded clique-width.
We copy this construction below.
Let $t \geq 2$ and let~$G$ be the $t \times t$ square grid.
Let $v^G_1,\ldots,v^G_n$ be the vertices of~$G$ and let $e^G_1,\ldots,e^G_m$ be the edges of~$G$.
We construct a graph~$q(G)$ from~$G$ as follows (see also \figurename~\ref{fig:L19}):
\begin{enumerate}
\item Create a complete multi-partite graph with partition $(A^G_1,\ldots,A^G_n)$, where $|A^G_i|=d_G(v^G_i)$ for $i \in \{1,\ldots,n\}$ and let $A^G=\bigcup A^G_i$.
\item Create a complete multi-partite graph with partition $(B^G_1,\ldots,B^G_m)$, where $|B^G_i|=2$ for $i \in \{1,\ldots,m\}$ and let $B^G=\bigcup B^G_i$.
\item Take the disjoint union of the two graphs above, then for each edge $e^G_i=v^G_{i_1}v^G_{i_2}$ in~$G$ in turn, add an edge from one vertex of~$B^G_i$ to a vertex of~$A^G_{i_1}$ and an edge from the other vertex of~$B^G_i$ to a vertex of~$A^G_{i_2}$.
Do this in such a way that the edges added between~$A^G$ and~$B^G$ form a perfect matching.
\end{enumerate}

\begin{figure}[h]
\begin{center}
\includegraphics[scale=0.55,page=4]{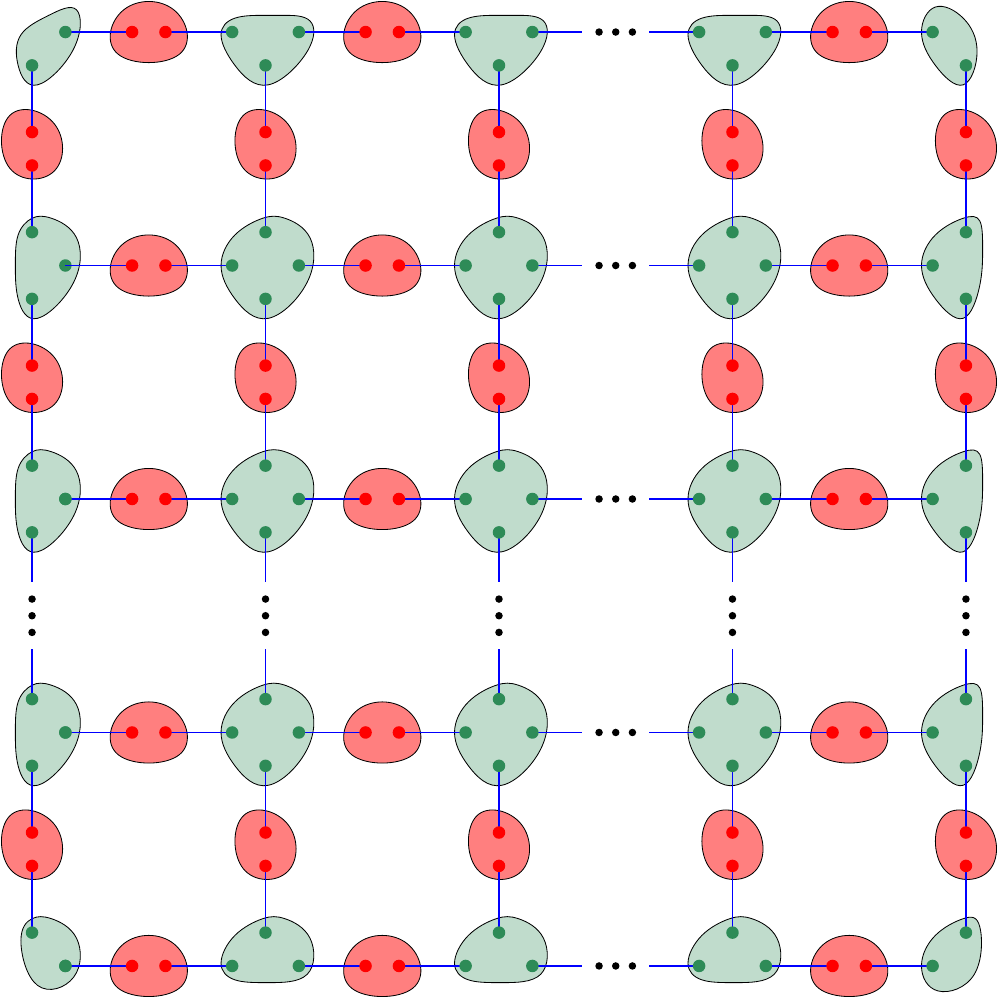}
\end{center}
\caption{The graph~$q[G]$ in Lemma~\ref{lem:gem-P12P2}.
Edges between partition classes~$A^G_i$ and edges between partition classes~$B^G_i$ are not shown.}
\label{fig:L19}
\end{figure}

In~\cite{BDJP18} it was shown that the graph~$q(G)$ is $(\overline{P_1+P_4},P_1+\nobreak 2P_2)$-free and that the clique-width of such graphs is unbounded.
Therefore the class of $(\overline{P_1+P_4},P_1+\nobreak 2P_2)$-free graphs has unbounded clique-width.
We observe that neither~$\overline{P_1+P_4}$ nor~$P_1+\nobreak 2P_2$ contains a pair of false twins.
Therefore, by Lemma~\ref{lem:no-false-twin}, it follows that the class of $(\overline{P_1+P_4},P_1+\nobreak 2P_2)$-free atoms has unbounded clique-width.

We now prove that the class of $(P_1+\nobreak P_4,\allowbreak \overline{P_1+2P_2})$-free atoms has unbounded clique-width.
By Fact~\ref{fact:comp}, the class of $(P_1+\nobreak P_4,\allowbreak \overline{P_1+2P_2})$-free graphs of the form~$\overline{q(G)}$ also has unbounded clique-width.
It therefore suffices to show that~$\overline{q(G)}$ is an atom.
Suppose, for contradiction, that~$X$ is a clique cut-set in~$\overline{q(G)}$.
In~$\overline{q(G)}$, $A^G$ induces a disjoint union of cliques of the form~$A^G_i$, and~$B^G$ induces a disjoint union of cliques of the form~$B^G_i$.
Therefore $X \in A^G_i \cup B^G_j$ for some~$i,j$.
Let~$k \neq i$ and $\ell \neq j$.
Then in~$\overline{q(G)}$ every vertex in~$B^G$ has a neighbour in~$A^G_k$ and every vertex in~$A^G$ has a neighbour in~$B^G_\ell$.
Since~$A_k^G$ and~$B_\ell^G$ are cliques, it follows that $\overline{q(G)} \setminus X$ is connected, a contradiction.
Therefore~$\overline{q(G)}$ is indeed an atom.
\end{proof}

\begin{figure}[h]
\begin{center}
\begin{tabular}{cc}
\scalebox{0.6}{
\begin{minipage}{0.4\textwidth}
\begin{center}%2P_1+P_2
\begin{tikzpicture}[every node/.style={circle,fill, minimum size=0.07cm}]
\node at (0,0) {};
\node at (0,2) {};
\node at (2,0) {};
\node at (2,2) {};
\draw (2,0) -- (0,0);
\end{tikzpicture}
\end{center}
\end{minipage}}
&
\scalebox{0.6}{
\begin{minipage}{0.4\textwidth}
\begin{center}%co-P6
\begin{tikzpicture}[every node/.style={circle,fill, minimum size=0.07cm}]
\node at (0,1.73205080757) {};%sqrt(3)
\node at (0,-1.73205080757) {};%sqrt(3)
\node at (1,0) {};
\node at (3,0) {};
\node at (-1,0) {};
\node at (-3,0) {};
\draw (-3,0) -- (-1,0) -- (1,0) -- (3,0);
\draw (0,1.73205080757) -- (1,0) -- (0,-1.73205080757);
\draw (0,1.73205080757) -- (3,0) -- (0,-1.73205080757);
\draw (-1,0) -- (0,-1.73205080757);
\draw (0,1.73205080757) -- (-3,0) -- (0,-1.73205080757);
\end{tikzpicture}
\end{center}
\end{minipage}}
\\
\\
$2P_1+\nobreak P_2$ & $\overline{P_6}$
\end{tabular}
\end{center}
\caption{\label{fig:codiamond-coP6}The forbidden induced subgraphs for the class of $(2P_1+\nobreak P_2,\overline{P_6})$-free graphs mentioned in Lemma~\ref{lem:codiamond-coP6}.}
\end{figure}

\begin{lemma}\label{lem:codiamond-coP6}
The class of $(2P_1+\nobreak P_2,\overline{P_6})$-free atoms has unbounded clique-width (see \figurename~\ref{fig:codiamond-coP6} for illustrations of the forbidden induced subgraphs).
\end{lemma}
\begin{proof}
By Theorem~\ref{thm:classification2}.\ref{thm:classification2:unbdd:2P_1+P_2}, the class of $(2P_1+\nobreak P_2,\overline{P_6})$-free graphs has unbounded clique-width.
We observe that neither~$2P_1+\nobreak P_2$ nor~$\overline{P_6}$ has a dominating vertex or a pair of non-adjacent vertices that are complete to the remainder of the graph.
Therefore, by Lemma~\ref{lem:no-comp-P1or2P1}, it follows that the class of $(2P_1+\nobreak P_2,\overline{P_6})$-free atoms has unbounded clique-width.
\end{proof}

\begin{figure}[h]
\begin{center}
\begin{tabular}{cc}
\scalebox{0.6}{
\begin{minipage}{0.3\textwidth}
\begin{center}%gem
\begin{tikzpicture}[every node/.style={circle,fill, minimum size=0.07cm}]
\node at (0,0) {};
\node at (0,2) {};
\node at (2,0) {};
\node at (2,2) {};
\node at (1,1) {};
\draw (2,2) -- (2,0) -- (0,0) -- (0,2);
\draw (0,0) -- (1,1);
\draw (0,2) -- (1,1);
\draw (2,0) -- (1,1);
\draw (2,2) -- (1,1);
\end{tikzpicture}
\end{center}
\end{minipage}}
&
\scalebox{0.6}{
\begin{minipage}{0.4\textwidth}
\begin{center}%P6
\begin{tikzpicture}[every node/.style={circle,fill, minimum size=0.07cm}]
\node at (0:1) {};
\node at (60:1) {};
\node at (120:1) {};
\node at (180:1) {};
\node at (240:1) {};
\node at (300:1) {};
\draw (120:1) -- (180:1) -- (240:1) -- (300:1) -- (0:1) -- (60:1);
\end{tikzpicture}
\end{center}
\end{minipage}}
\\
\\
$\overline{P_1+P_4}$ & $P_6$
\end{tabular}
\end{center}
\caption{\label{fig:gem-P6}The forbidden induced subgraphs for the class of $(\overline{P_1+P_4},P_6)$-free graphs mentioned in Lemma~\ref{lem:gem-P6}.}
\end{figure}
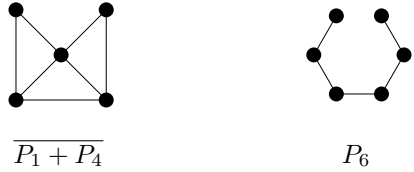

\begin{lemma}\label{lem:gem-P6}
The class of $(\overline{P_1+P_4},P_6)$-free atoms has unbounded clique-width (see \figurename~\ref{fig:gem-P6} for illustrations of the forbidden induced subgraphs).
\end{lemma}
\begin{proof}
By Theorem~\ref{thm:classification2}.\ref{thm:classification2:unbdd:2P_1+P_2}, the class of $(\overline{P_1+P_4},P_6)$-free graphs has unbounded clique-width.
We observe that neither~$\overline{P_1+P_4}$ nor~$P_6$ contains a pair of false twins.
Therefore, by Lemma~\ref{lem:no-false-twin}, it follows that the class of $(\overline{P_1+P_4},P_6)$-free atoms has unbounded clique-width.
\end{proof}

\begin{figure}[h]
\begin{center}
\begin{tabular}{cccc}
\scalebox{0.6}{
\begin{minipage}{0.3\textwidth}
\begin{center}%co-P_2+P_3
\begin{tikzpicture}[every node/.style={circle,fill, minimum size=0.07cm}]
\node at (0,0) {};
\node at (0,2) {};
\node at (2,0) {};
\node at (2,2) {};
\node at (1,1) {};
\draw (0,0) -- (0,2) -- (2,2) -- (2,0) -- cycle;
\draw (2,2) -- (1,1)--(0,0);
\draw (0,2) -- (1,1);
\end{tikzpicture}
\end{center}
\end{minipage}}
&
\scalebox{0.6}{
\begin{minipage}{0.3\textwidth}
\begin{center}%\overline{P_1+2P_2}
\begin{tikzpicture}[every node/.style={circle,fill, minimum size=0.07cm}]
\node at (0,0) {};
\node at (0,2) {};
\node at (2,0) {};
\node at (2,2) {};
\node at (1,1) {};
\draw (0,0) -- (0,2) -- (2,2) -- (2,0) -- (0,0);
\draw (0,0) -- (1,1);
\draw (0,2) -- (1,1);
\draw (2,0) -- (1,1);
\draw (2,2) -- (1,1);
\end{tikzpicture}
\end{center}
\end{minipage}}
&
\scalebox{0.6}{
\begin{minipage}{0.3\textwidth}
\begin{center}%P_1+2P_2
\begin{tikzpicture}[every node/.style={circle,fill, minimum size=0.07cm}]
\node at (0,0) {};
\node at (0,2) {};
\node at (2,0) {};
\node at (2,2) {};
\node at (1,1) {};
\draw (0,0) -- (0,2)  (2,2) -- (2,0);
\end{tikzpicture}
\end{center}
\end{minipage}}
&
\scalebox{0.6}{
\begin{minipage}{0.4\textwidth}
\begin{center}%P6
\begin{tikzpicture}[every node/.style={circle,fill, minimum size=0.07cm}]
\node at (0:1) {};
\node at (60:1) {};
\node at (120:1) {};
\node at (180:1) {};
\node at (240:1) {};
\node at (300:1) {};
\draw (120:1) -- (180:1) -- (240:1) -- (300:1) -- (0:1) -- (60:1);
\end{tikzpicture}
\end{center}
\end{minipage}}
\\
\\
$\overline{P_2+P_3}$ & $\overline{P_1+2P_2}$ & $P_1+\nobreak 2P_2$ & $P_6$
\end{tabular}
\end{center}
\caption{\label{fig:coP2P3-coP12P2-P12P2-P6-forb}The forbidden induced subgraphs for the class of $(\overline{P_2+P_3},\overline{P_1+2P_2},P_1+\nobreak 2P_2,P_6)$-free graphs mentioned in Lemma~\ref{lem:coP2P3-coP12P2-P12P2-P6}.}
\end{figure}
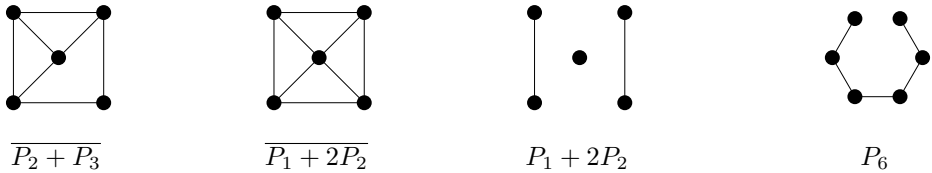

\begin{lemma}\label{lem:coP2P3-coP12P2-P12P2-P6}
The class of $(\overline{P_2+P_3},\overline{P_1+2P_2},P_1+\nobreak 2P_2,P_6)$-free atoms has unbounded clique-width (see \figurename~\ref{fig:coP2P3-coP12P2-P12P2-P6-forb} for illustrations of the forbidden induced subgraphs).
\end{lemma}

\begin{proof}
Consider a $1$-subdivision of a wall of height $n \geq 2$.
Let~$A$ be the set of original vertices of the wall and let~$B$ be the set of vertices introduced by the subdivision.
Apply a complementation to~$A$ and add a vertex~$x$ complete to~$B$.
Let~$H_n$ be the resulting graph (see \figurename~\ref{fig:coP2P3-coP12P2-P12P2-P6}).
By Lemma~\ref{lem:walls}, combined with Facts~\ref{fact:del-vert} and~\ref{fact:comp}, the family of such graphs has unbounded clique-width.
Note that~$x$ is complete to the independent set~$B$ and anti-complete to the clique~$A$.
Every vertex in~$B$ has exactly two neighbours in~$A$ and every vertex in~$A$ has either two or three neighbours in~$B$.
Furthermore, no two vertices in~$B$ have the same pair of neighbours in~$A$.

\begin{figure}[h]
\begin{center}
\usetikzlibrary{calc}
\begin{tikzpicture}[scale=0.7, type1/.style={circle, draw, fill, inner sep=2pt,minimum size=0.5mm}, type2/.style={circle,thick, draw=black, line width=1pt, inner sep=2pt,minimum size=0.5mm},rotate=180   ]

\node  (aa-11) at (-1*2,2)  {$A$};
\node  (aa-12) at (-1*2,4)  {$A$};
\node (bm-11) at (-1*2+1,2) {$B$};
\node  (bm-12) at (-1*2+1,4) {$B$};
  
   \foreach \y in {0,1,2,3}{
     \foreach \x in {0,1,2}{
        \node (a\x\y)   at (\x*4,\y*2)    {$A$};
        \node (aa\x\y)   at (\x*4+2,\y*2)    {$A$};
        \node (b\x\y)  at (\x*4+1,\y*2)  {$B$};
        \node  (bb\x\y)  at (\x*4+3,\y*2)  {$B$};
        \draw (a\x\y)--(b\x\y)--(aa\x\y)--(bb\x\y);
         }
    \node (a3\y) at (3*4,\y*2) {$A$};
    
    \foreach \x in {0,1,2}{
            \pgfmathtruncatemacro{\j}{\x+1}
       \draw (bb\x\y)--(a\j\y);
    }
 }

 \foreach \x in {0,1,2,3}{
    \foreach \y in {0,2}{
        \node  (bm\x\y) at (\x*4,\y*2+1) {$B$};
        \pgfmathtruncatemacro{\j}{\y+1}
        \draw (a\x\y)--(bm\x\y) -- (a\x\j) ;
        }
     
     \node  (bm\x2) at (\x*4-2,1*2+1) {$B$};
     \pgfmathtruncatemacro{\j}{\x-1}
     \draw (aa\j1)--(bm\x2)--(aa\j2);
}
 
\draw (aa-11)--(bm-11)--(a01);
\draw (aa-12)--(bm-12)--(a02);

\node[circle, inner sep=1pt] (x) at (5,-2) {$x$};

\foreach \x in {-1,1.6,3.6,3.8,4,4.2,4.5,5.1 ,5.8,6.4,9}{
    \draw (\x+1,-1)--(x);
    }

\foreach \x in {0,1,2}{
    \draw ($ (bb\x0) !.3! (x) $)--(x);
    \draw ($ (b\x0) !.3! (x) $)--(x);
}

\node[draw=black, rectangle, align=left, inner sep=4pt, anchor=west] (label) at (17.5,3) {
    $A$ is a clique\\
    $x$ is complete to $B$
    };

;

\end{tikzpicture}
\caption{The graph~$H_n$ from the proof of Lemma~\ref{lem:coP2P3-coP12P2-P12P2-P6} ($n=3$ shown).
Vertices are denoted~$A$ or~$B$ if they are in the corresponding set.}
\label{fig:coP2P3-coP12P2-P12P2-P6}
\end{center}
\end{figure}

We prove that~$H_n$ is an atom.
Suppose, for contradiction, that~$X$ is a clique cut-set of~$H_n$.
If $x \in X$, then~$X$ can contain at most one other vertex (which must be in the independent set~$B$).
Since every vertex in~$B$ has neighbours in the clique~$A$, it follows that $H_n\setminus X$ is connected.
We may therefore assume that $x \notin X$.
Since~$B$ is an independent set, $|B \cap X| \leq 1$, so every vertex of~$A$ has a neighbour in $B \setminus X$.
Since~$x$ is complete to~$B$, it follows that every vertex of $B \setminus X$ is in the same component of~$B \setminus X$ as~$x$, and so $H_n \setminus X$ is connected, a contradiction.
We conclude that~$H_n$ is indeed an atom.

It remains to show that~$H_n$ is $(\overline{P_2+P_3},\overline{P_1+2P_2},P_1+\nobreak 2P_2,P_6)$-free.
Note that $H_n[A \cup B]$ is a split graph, so it is $(2P_2,\overline{2P_2})$-free.
Therefore every induced~$2P_2$ or~$\overline{2P_2}$ in~$H_n$ contains the vertex~$x$.
Suppose, for contradiction, that~$H_n$ contains an induced~$\overline{P_2+P_3}$ or~$\overline{P_1+2P_2}$, say on vertex set~$Y$.
Since~$\overline{P_2+P_3}$ and~$\overline{P_1+2P_2}$ each contain an induced~$\overline{2P_2}$, it follows that $x \in Y$.
Since~$x$ has two neighbours and one non-neighbour in this~$\overline{2P_2}$, this~$\overline{2P_2}$ consists of the vertex~$x$, two vertices in~$B$ and one vertex in~$A$.
Now~$Y$ contains one more vertex~$y$, which is adjacent to either three or four of the remaining vertices of~$Y$.
Now~$y$ cannot be in~$B$, since~$B$ is an independent set and there are two vertices in $(B \cap Y) \setminus \{y\}$, so $y \in A$.
Therefore $A \cap Y$ contains two vertices with two common neighbours in~$B$, contradicting the fact that no two vertices of~$B$ have the same two neighbours in~$A$.
We conclude that~$H_n$ is indeed $(\overline{P_2+P_3},\overline{P_1+2P_2})$-free.
Now suppose, for contradiction, that~$H_n$ contains an induced~$P_1+\nobreak 2P_2$ or an induced~$P_6$, say on vertex set~$Y$.
Since~$P_1+\nobreak 2P_2$ and~$P_6$ each contain an induced~$2P_2$, we find that $x\in Y$.
Every vertex in~$P_1+\nobreak 2P_2$ and~$P_6$ has at least three non-neighbours.
Since~$x$ is complete to~$B$, it follows that $|A \cap Y| \geq 3$.
But~$A$ is a clique and so $H_n[A \cap Y]$ contains a~$K_3$, which is a contradiction, since~$P_1+\nobreak 2P_2$ and~$P_6$ are $K_3$-free.
We conclude that~$H_n$ is $(P_1+\nobreak 2P_2, P_6)$-free. 
Hence, $H_n$ is $(\overline{P_2+P_3},\overline{P_1+2P_2},P_1+\nobreak 2P_2,P_6)$-free.
\end{proof}

\begin{figure}[h]
\begin{center}
\begin{tabular}{cc}
\scalebox{0.6}{
\begin{minipage}{0.45\textwidth}
\begin{center}%2P_2
\begin{tikzpicture}[every node/.style={circle,fill, minimum size=0.07cm}]
\node at (0,0) {};
\node at (0,2) {};
\node at (2,0) {};
\node at (2,2) {};
\draw (0,0) -- (0,2);
\draw (2,0) -- (2,2);
\end{tikzpicture}
\end{center}
\end{minipage}}
&
\scalebox{0.6}{
\begin{minipage}{0.45\textwidth}
\begin{center}%co-P2+P4
\begin{tikzpicture}[every node/.style={circle,fill, minimum size=0.07cm}]
\node at (0,1.73205080757) {};%sqrt(3)
\node at (0,-1.73205080757) {};%sqrt(3)
\node at (1,0) {};
\node at (3,0) {};
\node at (-1,0) {};
\node at (-3,0) {};
\draw (-3,0) -- (-1,0) -- (1,0) -- (3,0);
\draw (0,1.73205080757) -- (1,0) -- (0,-1.73205080757);
\draw (0,1.73205080757) -- (3,0) -- (0,-1.73205080757);
\draw (0,1.73205080757) -- (-1,0) -- (0,-1.73205080757);
\draw (0,1.73205080757) -- (-3,0) -- (0,-1.73205080757);
\end{tikzpicture}
\end{center}
\end{minipage}}
\\
\\
$2P_2$ & $\overline{P_2+P_4}$
\end{tabular}
\end{center}
\caption{\label{fig:2P2-coP2P4}The forbidden induced subgraphs for the class of $(2P_2,\overline{P_2+P_4})$-free graphs mentioned in Lemma~\ref{lem:2P2-coP2P4}.}
\end{figure}
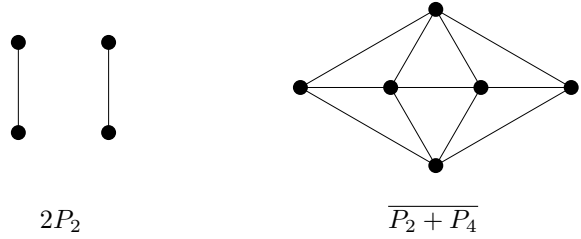

\newpage
\begin{lemma}\label{lem:2P2-coP2P4}
The class of $(2P_2,\overline{P_2+P_4})$-free atoms has unbounded clique-width (see \figurename~\ref{fig:2P2-coP2P4} for illustrations of the forbidden induced subgraphs).
\end{lemma}
\begin{proof}
Let~$n \geq 2$ and construct the graph~$G_n$ as follows (see also \figurename~\ref{fig:L20}).
Let the vertex set of~$G_n$ be $\{v_{i,j} \; | \; i,j \in \{0,\ldots,n\}, (i,j) \neq (0,0)\} \cup \{v_{0,n+1}\}$.
For $i,j,k \in \{1,\ldots,n\}$, add an edge between $v_{i,j}$ and~$v_{k,0}$ if $k \geq i$, add an edge between~$v_{i,j}$ and~$v_{0,k}$ if $k \geq j$, and add an edge between~$v_{i,j}$ and~$v_{0,n+1}$.
Let $A = \{v_{i,0} \; | \; i \in \{1,\ldots,n\}\}$, $B = \{v_{0,j} \; | \; j \in \{1,\ldots,n+1\}\}$, and $C = \{v_{i,j} \; | \; i,j \in \{1,\ldots,n\}\}$.
Apply a bipartite complementation between~$A$ and~$B$ and apply a complementation to~$A$.
By Lemma~\ref{lem:generalunbounded} combined with Facts~\ref{fact:del-vert}, \ref{fact:comp} and~\ref{fact:bip}, the family of graphs~$G_n$ has unbounded clique-width.

\begin{figure}[h]
\begin{center}
\includegraphics[scale=0.55,page=5]{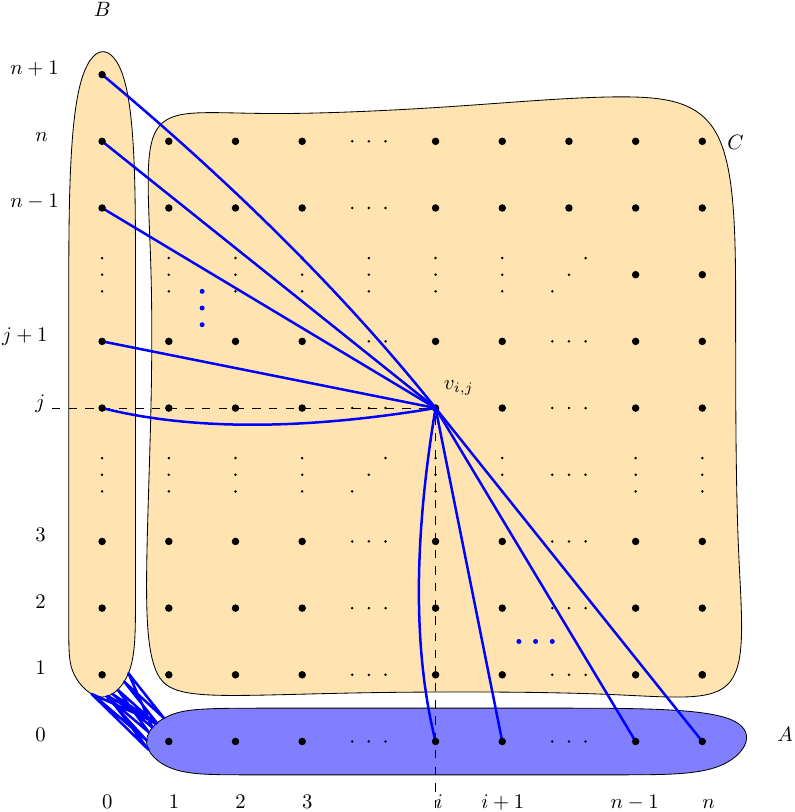}
\end{center}
\caption{The graph~$G_n$ constructed in the proof of Lemma~\ref{lem:2P2-coP2P4}.
The set~$A$ is a clique.
For clarity, the edges between $A\cup B$ and~$C$ are depicted for only one vertex in~$C$.}
\label{fig:L20}
\end{figure}

We claim that~$G_n$ is an atom.
Suppose, for contradiction, that~$X$ is a clique cut-set of~$G_n$.
Since~$v_{0,n}$ and~$v_{0,n+1}$ are non-adjacent, but complete to~$C$, it follows that every vertex of~$C$ is in the same component of~$G_n \setminus X$.
Since~$C$ is independent, at most one vertex of~$C$ is in~$X$.
Since every vertex in $A \cup B$ has at least two neighbours in~$C$, it follows that every vertex of $G_n \setminus X$ is in the same component of~$G_n \setminus X$, a contradiction.
Therefore~$G_n$ is indeed an atom.

Now suppose, for contradiction, that~$G_n$ contains an induced subgraph isomorphic to~$2P_2$, say on vertex set~$Y$.
Since~$G_n[A \cup C]$ is a split graph, it is $2P_2$-free, so~$Y$ must contain at least one vertex of~$B$.
Since~$G_n[B \cup C]$ is a bipartite chain graph, it is $2P_2$-free, so~$Y$ contains at least one vertex of~$A$.
Since~$A$ is complete to~$B$, it follows that $Y$ contains exactly one vertex of~$A$ and exactly one vertex of~$B$, and these two vertices are adjacent.
Therefore~$Y \cap C$ contains two adjacent vertices, contradicting the fact that~$C$ is independent.
We conclude that~$G_n$ is $2P_2$-free.

Next suppose, for contradiction, that~$G_n$ contains an induced subgraph isomorphic to~$\overline{P_2+P_4}$, say on vertex set~$Y$.
Since~$\overline{P_2+P_4}$ is $3P_1$-free, and~$B$ and~$C$ are independent, it follows that $|B \cap Y| \leq 2$ and $|C \cap Y| \leq 2$.
Therefore $|A \cap Y| \geq 6-2-2=2$.
Now $G_n[A \cup B]$ and $G_n[A \cup C]$ are split graphs and therefore~$C_4$-free.
Since~$\overline{P_2+P_4}$ contains an induced~$\overline{2P_2} = C_4$, it follows that $|C \cap Y| \geq 1$ and $|B \cap Y| \geq 1$.
Since~$A$ is a clique that is complete to~$B$ and~$\overline{P_2+P_4}$ is $K_4$-free, it follows that $|A\cap Y|\leq 2$.     
Therefore $|A \cap Y| = |B \cap Y| = |C \cap Y|=2$.
Since~$B$ and~$C$ are cliques in~$\overline{G_n}$, we observe that all vertices in~$B\cap Y$ are in the same component of~$\overline{G_n}[Y]$, and all  vertices in~$C\cap Y$ are in the same component of~$\overline{G_n}[Y]$.
Furthermore, since in~$\overline{G_n}$ the set~$A$ is independent and anti-complete to~$B$, the vertices in $(A\cup C) \cap Y$ form the $P_4$-component of~$\overline{G_n}[Y]$, and vertices in $B \cap Y$ form the $P_2$-component of~$\overline{G_n}[Y]$.
Therefore $\overline{G_n}[(A \cup C) \cap Y]$ is isomorphic to~$P_4$, which means that in~$\overline{G_n}$ there must be two vertices in the clique~$C$ that have private neighbours in the independent set~$A$.
By the construction of~$\overline{G_n}$, the vertices in~$C$ can be linearly ordered according to their neighbourhood in~$A$, a contradiction.
Therefore~$G_n$ is indeed $\overline{P_2+P_4}$-free.
\end{proof}

\begin{figure}[h]
\begin{center}
\begin{tabular}{ccc}
\scalebox{0.6}{
\begin{minipage}{0.3\textwidth}
\begin{center}%2P_2
\begin{tikzpicture}[every node/.style={circle,fill, minimum size=0.07cm}]
\node at (0,0) {};
\node at (0,2) {};
\node at (2,0) {};
\node at (2,2) {};
\draw (0,0) -- (0,2);
\draw (2,0) -- (2,2);
\end{tikzpicture}
\end{center}
\end{minipage}}
&
\scalebox{0.6}{
\begin{minipage}{0.3\textwidth}
\begin{center}%co-P5
\begin{tikzpicture}[every node/.style={circle,fill, minimum size=0.07cm}]
\node at (18:1) {};
\node at (90:1) {};
\node at (162:1) {};
\node at (234:1) {};
\node at (306:1) {};
\draw (234:1) -- (306:1) -- (18:1) -- (162:1) -- (234:1);
\draw (18:1) -- (90:1) -- (162:1);
\end{tikzpicture}
\end{center}
\end{minipage}}
&
\scalebox{0.6}{
\begin{minipage}{0.3\textwidth}
\begin{center}%co-3P2
\begin{tikzpicture}[every node/.style={circle,fill, minimum size=0.07cm}]
\node at (0:1) {};
\node at (60:1) {};
\node at (120:1) {};
\node at (180:1) {};
\node at (240:1) {};
\node at (300:1) {};
\draw (120:1) -- (180:1) -- (240:1) -- (300:1) -- (0:1) -- (60:1) -- (120:1);
\draw (0:1) -- (120:1) -- (240:1) -- (0:1);
\draw (60:1) -- (180:1) -- (300:1) -- (60:1);
\end{tikzpicture}
\end{center}
\end{minipage}} \\
\\
$2P_2$ & $\overline{P_5}$ & $\overline{3P_2}$
\end{tabular}
\end{center}
\caption{\label{fig:2P_2-coP5-co3P2}The forbidden induced subgraphs for the class of $(2P_2,\overline{P_5},\overline{3P_2})$-free graphs mentioned in Lemma~\ref{lem:2P_2-coP5-co3P2}.}
\end{figure}
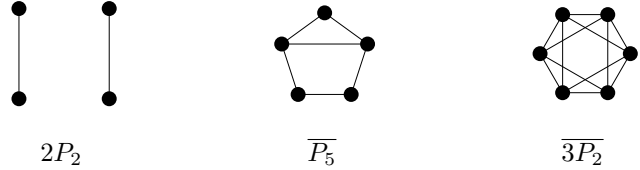

\begin{lemma}\label{lem:2P_2-coP5-co3P2}
The class of $(2P_2,\overline{P_5},\overline{3P_2})$-free atoms has unbounded clique-width (see \figurename~\ref{fig:2P_2-coP5-co3P2} for illustrations of the forbidden induced subgraphs).
\end{lemma}
\begin{proof}
Consider a wall of height $k \geq 2$ and note that it is a bipartite graph, say with parts~$A$ and~$B$.
Apply a complementation to~$A$ and add two vertices~$x$ and~$y$ that are complete to $A \cup B$.
Let~$H_k$ be the resulting graph (see \figurename~\ref{fig:2P2_coP5}).
By Lemma~\ref{lem:walls}, combined with Facts~\ref{fact:del-vert} and~\ref{fact:comp}, the class of such graphs has unbounded clique-width.

\begin{figure}[h]
\begin{center}
\begin{tikzpicture}[scale=0.7, type1/.style={circle, draw, fill, inner sep=2pt,minimum size=0.5mm}, type2/.style={circle,thick, draw=black, line width=1pt, inner sep=2pt,minimum size=0.5mm},rotate=180   ]

\node  (aa-11) at (-1*2,2)  {$A$};
\node  (aa-12) at (-1*2,4)  {$B$};
  
   \foreach \y in {0,2}{
     \foreach \x in {0,1,2}{
        \node (a\x\y)   at (\x*4,\y*2)    {$A$};
        \node (aa\x\y)   at (\x*4+2,\y*2)    {$B$};
        \draw (a\x\y)--(aa\x\y);
         }
    \node (a3\y) at (3*4,\y*2) {$A$};
    
    \foreach \x in {0,1,2}{
            \pgfmathtruncatemacro{\j}{\x+1}
       \draw (aa\x\y)--(a\j\y);
    }
 }

   \foreach \y in {1,3}{
     \foreach \x in {0,1,2}{
        \node (a\x\y)   at (\x*4,\y*2)    {$B$};
        \node (aa\x\y)   at (\x*4+2,\y*2)    {$A$};
        \draw (a\x\y)--(aa\x\y);
         }
    \node (a3\y) at (3*4,\y*2) {$B$};
    
    \foreach \x in {0,1,2}{
            \pgfmathtruncatemacro{\j}{\x+1}
       \draw (aa\x\y)--(a\j\y);
    }
 }

 \foreach \x in {0,1,2,3}{
    \foreach \y in {0,2}{
        \pgfmathtruncatemacro{\j}{\y+1}
        \draw (a\x\y)-- (a\x\j) ;
        }
     
     \pgfmathtruncatemacro{\j}{\x-1}
     \draw (aa\j1)--(aa\j2);
}
 
\draw (aa-11)--(a01);
\draw (aa-12)--(a02);

\node[circle, inner sep=1pt] (x) at (6.5,-2) {$x$};
\node[circle, inner sep=1pt] (y) at (4.5,-2) {$y$};

\foreach \x in {-4,-3,-2,-1,0,1,2,4,5,6,7,8,9,10,11,12,13,14}{
    \draw ($ (\x,0) !.3! (x) $)--(x);
    \draw ($ (\x,0) !.3! (y) $)--(y);
}

\node[draw=black, rectangle, align=left, inner sep=4pt, anchor=west] (label) at (18,3) {
    $A$ is a clique\\
    $x$ and $y$ are complete\\
     to $A\cup B$
    };

;

\end{tikzpicture}
 \caption{The graph~$H_k$ from the proof of Lemma~\ref{lem:2P_2-coP5-co3P2} ($k=3$ shown).
Vertices are denoted~$A$ or~$B$ if they are in the corresponding set.}
\label{fig:2P2_coP5}
\end{center}
\end{figure}
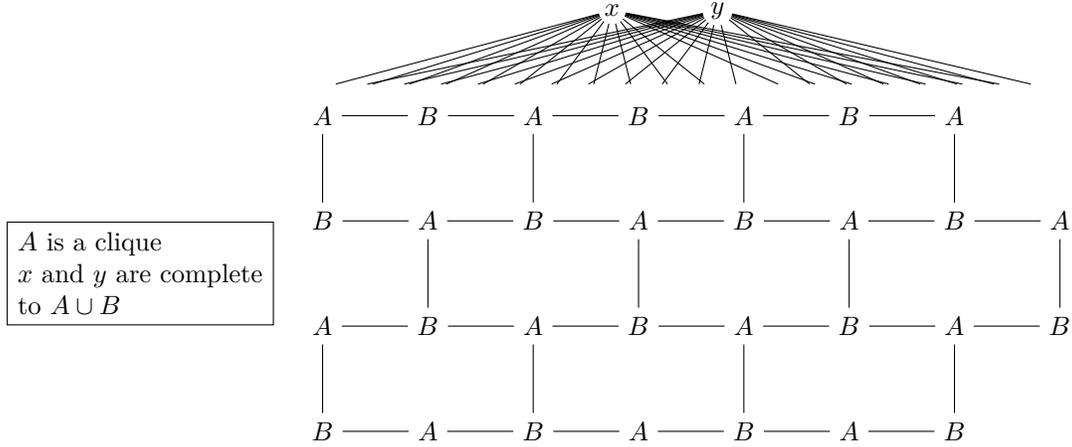

We claim that~$H_k$ is an atom.
Suppose, for contradiction, that~$X$ is a clique cut-set of~$H_k$.
Since~$x$ is non-adjacent to~$y$, at most one of~$x$ and~$y$ is in~$X$.
Without loss of generality, we may assume that $y \notin X$.
Now~$y$ is adjacent to every vertex of $A \cup B$, so $H_k[\{y\} \cup (A \cup B) \setminus X]$ is connected.
Since $A \cup B$ is not a clique, there must be at least one vertex in $(A \cup B) \setminus X$.
Therefore, if $x \notin X$ then~$x$ is in the same component of $H_k \setminus X$ as~$y$ is.
It follows that $H_k \setminus X$ is connected, a contradiction.
Therefore~$H_k$ is indeed an atom.

It remains to show that~$H_k$ is $(2P_2,\overline{P_5},\overline{3P_2})$-free.
First, note that $H_k \setminus \{x\}$ and $H_k \setminus \{y\}$ are split graphs, so they are $(2P_2,\overline{2P_2})$-free, and therefore $(2P_2,\overline{P_5},\allowbreak \overline{P_1+2P_2})$-free and note that this also implies that~$H_k$ is $\overline{3P_2}$-free.
Therefore, if~$H_k$ contains an induced~$2P_2$ or~$\overline{P_5}$, then this induced copy must contain both~$x$ and~$y$.
Since~$x$ and~$y$ are false twins in~$H_k$, but~$2P_2$ and~$\overline{P_5}$ do not contain two vertices that are false twins, it follows that~$H_k$ is $(2P_2,\overline{P_5})$-free.
This completes the proof.
\end{proof}

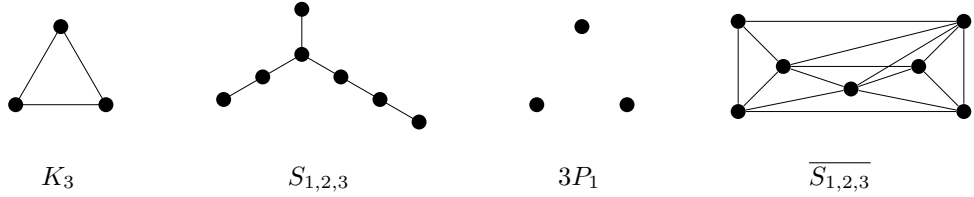
\begin{figure}[h]
\begin{center}
\begin{tabular}{cccc}
\scalebox{0.6}{
\begin{minipage}{0.3\textwidth}
\begin{center}%C_3
\begin{tikzpicture}[every node/.style={circle,fill, minimum size=0.07cm}]
\node at (0,0) {};
\node at (1,1.73205080757) {};%sqrt(3)
\node at (2,0) {};
\draw (0,0) -- (1,1.73205080757) -- (2,0) -- (0,0);
\end{tikzpicture}
\end{center}
\end{minipage}}
&
\scalebox{0.6}{
\begin{minipage}{0.3\textwidth}
\begin{center}%S123
\begin{tikzpicture}[every node/.style={circle,fill, minimum size=0.07cm}]
\node at (0,0) {};
\node at (90:1) {};
\node at (210:1) {};
\node at (210:2) {};
\node at (330:1) {};
\node at (330:2) {};
\node at (330:3) {};
\draw (210:2) -- (210:1) -- (0,0) -- (330:1) -- (330:2) -- (330:3);
\draw (0,0) -- (90:1);
\end{tikzpicture}
\end{center}
\end{minipage}}
&
\scalebox{0.6}{
\begin{minipage}{0.3\textwidth}
\begin{center}%3P_1
\begin{tikzpicture}[every node/.style={circle,fill, minimum size=0.07cm}]
\node at (0,0) {};
\node at (1,1.73205080757) {};%sqrt(3)
\node at (2,0) {};
\end{tikzpicture}
\end{center}
\end{minipage}}
&
\scalebox{0.6}{
\begin{minipage}{0.3\textwidth}
\begin{center}%3P_1
\begin{tikzpicture}[every node/.style={circle,fill, minimum size=0.07cm}]
\node(a) at (-2,2) {};
\node(b) at (-2,0) {};
\node(c) at (3,0) {};
\node(d) at (2,1) {};
\node(e) at (3,2) {};
\node(f) at (-1,1) {};
\node(g) at (0.5,0.5) {};
\draw (a) -- (e) (a) -- (b) (a) -- (f);
\draw (b) -- (c) (b) -- (f) (b) -- (g);
\draw (c) -- (d) (c) -- (e) (c) -- (g);
\draw (d) -- (e) -- (f) -- (g) -- (d) -- (f) (e) -- (g);
\end{tikzpicture}
\end{center}
\end{minipage}}
\\
\\
$K_3$ & $S_{1,2,3}$ & $3P_1$ & $\overline{S_{1,2,3}}$
\end{tabular}
\end{center}
\caption{\label{fig:K3S123-coK3coS123-equiv}The forbidden induced subgraphs for the classes of $(K_3,S_{1,2,3})$-free graphs and $(3P_1,\overline{S_{1,2,3}})$-free graphs mentioned in Lemma~\ref{lem:K3S123-coK3coS123-equiv}.}
\end{figure}

It is not known whether the clique-width of $(K_3,S_{1,2,3})$-free graphs is bounded or unbounded.
Recall that this case is equivalent to the open case for $(3P_1,\overline{S_{1,2,3}})$-free graphs; see also Open Problem~\ref{oprob:twographs}.
As a final result of this section, we observe that the boundedness for these cases also matches their atom counterparts.

\newpage
\begin{lemma}\label{lem:K3S123-coK3coS123-equiv}
The class of $(K_3,S_{1,2,3})$-free atoms has bounded clique-width if and only if the class of $(K_3,S_{1,2,3})$-free graphs has bounded clique-width if and only if the class of $(3P_1,\overline{S_{1,2,3}})$-free atoms has bounded clique-width (see \figurename~\ref{fig:K3S123-coK3coS123-equiv} for illustrations of the forbidden induced subgraphs).
\end{lemma}
\begin{proof}
We first observe that neither~$K_3$ nor~$S_{1,2,3}$ contains a pair of false twins.
Therefore, by Lemma~\ref{lem:no-false-twin}, the class of $(K_3,S_{1,2,3})$-free atoms has bounded clique-width if and only if the class of $(K_3,S_{1,2,3})$-free graphs has bounded clique-width.
By Fact~\ref{fact:comp}, the class of $(K_3,S_{1,2,3})$-free graphs has bounded clique-width if and only if the class of $(3P_1,\overline{S_{1,2,3}})$-free graphs has bounded clique-width.

We observe that neither~$3P_1$ nor~$\overline{S_{1,2,3}}$ has a dominating vertex or a pair of non-adjacent vertices that are complete to the remainder of the graph.
Therefore, by Lemma~\ref{lem:no-comp-P1or2P1}, the class of $(3P_1,\overline{S_{1,2,3}})$-free graphs has bounded clique-width if and only if the class of $(3P_1,\overline{S_{1,2,3}})$-free atoms has bounded clique-width.
\end{proof}

\section{The Proof of Theorem~\ref{thm:classification2-atoms}}\label{s-main}

Recall the definition of equivalent bigenic classes given at the start of Section~\ref{s-soa}.
To make Theorem~\ref{thm:classification2-atoms} easier to compare to Theorem~\ref{thm:classification2}, in this section we will use the following reformulation of it, where we group classes together if they are equivalent, and we will prove this reformulated version of the theorem instead (it is easy to verify that Theorems~\ref{thm:classification2-atoms} and~\ref{bthm:classification2-atoms} cover the same graph classes).

\newpage
\begin{theorem}\label{bthm:classification2-atoms}
Let~${\cal G}$ be a class of graphs defined by two forbidden induced subgraphs.
\begin{enumerate}
\item \label{bthm:classification2-atoms:easy-bdd} The class of atoms in~${\cal G}$ has bounded clique-width if it is equivalent to a class of $(H_1,H_2)$-free graphs such that one of the following holds:
\begin{enumerate}[(i)]
\renewcommand{\theenumii}{(\roman{enumii})}
\renewcommand{\labelenumii}{(\roman{enumii})}
\item \label{bthm:classification2-atoms:bdd:P4} $H_1$ or $H_2 \ssi P_4$
\item \label{bthm:classification2-atoms:bdd:ramsey} $H_1=K_s$ and $H_2=tP_1$ for some $s,t\geq 1$
\item \label{bthm:classification2-atoms:bdd:P_1+P_3} $H_1 \ssi \paw$ and $H_2 \ssi K_{1,3}+\nobreak 3P_1,\; K_{1,3}+\nobreak P_2,\;\allowbreak P_1+\nobreak P_2+\nobreak P_3,\;\allowbreak P_1+\nobreak P_5,\;\allowbreak P_1+\nobreak S_{1,1,2},\;\allowbreak P_2+\nobreak P_4,\;\allowbreak P_6,\; \allowbreak S_{1,1,3}$ or~$S_{1,2,2}$
\item \label{bthm:classification2-atoms:bdd:2P_1+P_2} $H_1 \ssi \diamondgraph$ and $H_2\ssi P_1+\nobreak 2P_2,\; 3P_1+\nobreak P_2$ or~$P_2+\nobreak P_3$
\item \label{bthm:classification2-atoms:bdd:P_1+P_4} $H_1 \ssi \gem$ and $H_2 \ssi P_1+\nobreak P_4$ or~$P_5$
\item \label{bthm:classification2-atoms:bdd:K_13} $H_1\ssi K_3+\nobreak P_1$ and $H_2 \ssi K_{1,3}$, or
\item \label{bthm:classification2-atoms:bdd:2P1_P3} $H_1\ssi \overline{2P_1+\nobreak P_3}$ and $H_2\ssi 2P_1+\nobreak P_3$.
\end{enumerate}
\item \label{bthm:classification2-atoms:hard-bdd} The class of atoms in~${\cal G}$ has bounded clique-width if~${\cal G}$ is a subclass of the class of:
\begin{enumerate}[(i)]
\renewcommand{\theenumii}{(\roman{enumii})}
\renewcommand{\labelenumii}{(\roman{enumii})}
\item \label{bthm:classification2-atoms:bdd2:P6} $(P_6,\overline{2P_2})$-free graphs or
\item \label{bthm:classification2-atoms:bdd2:2P2} $(2P_2,\overline{P_2+P_3})$-free graphs.
\end{enumerate}
\item \label{bthm:classification2-atoms:unchanged-unbdd} The class of atoms in~${\cal G}$ has unbounded clique-width if it is equivalent to a class of $(H_1,H_2)$-free graphs such that one of the following holds:
\begin{enumerate}[(i)]
\renewcommand{\theenumii}{(\roman{enumii})}
\renewcommand{\labelenumii}{(\roman{enumii})}
\item \label{bthm:classification2-atoms:unbdd:not-in-S} $H_1\not\in {\cal S}$ and $H_2 \not \in {\cal S}$
\item \label{bthm:classification2-atoms:unbdd:not-in-co-S} $H_1\notin \overline{{\cal S}}$ and $H_2 \not \in \overline{{\cal S}}$
\item \label{bthm:classification2-atoms:unbdd:K_13} $H_1 \si K_3+\nobreak P_1$ and $H_2 \si 4P_1$ or~$2P_2$
\item \label{bthm:classification2-atoms:unbdd:2P_1+P_2} $H_1 \si \diamondgraph$ and $H_2 \si K_{1,3},\; 5P_1$ or~$P_2+\nobreak P_4$
\item \label{bthm:classification2-atoms:unbdd:3P_1} $H_1 \si K_3$ and $H_2 \si 2P_1+\nobreak 2P_2,\; 2P_1+\nobreak P_4,\; 4P_1+\nobreak P_2,\; 3P_2$ or~$2P_3$
\item \label{bthm:classification2-atoms:unbdd:4P_1} $H_1 \si K_4$ and $H_2 \si P_1 +\nobreak P_4,\; 3P_1+\nobreak P_2$ or~$2P_2$, or
\item \label{bthm:classification2-atoms:unbdd:gem} $H_1 \si \gem$ and $H_2 \si P_1+\nobreak 2P_2$.
\end{enumerate}
\item The class of atoms in~${\cal G}$ has unbounded clique-width if it contains the class of $(H_1,H_2)$-free graphs such that one of the following holds:
\begin{enumerate}[(i)]
\renewcommand{\theenumii}{(\roman{enumii})}
\renewcommand{\labelenumii}{(\roman{enumii})}
\item \label{bthm:classification2-atoms:unbdd2:diamond} $H_1 \si \diamondgraph$ and $H_2 \si P_1+\nobreak P_6$
\item \label{bthm:classification2-atoms:unbdd2:co-diamond} $H_1 \si 2P_1+\nobreak P_2$ and $H_2 \si \overline{P_6}$
\item \label{bthm:classification2-atoms:unbdd2:gem} $H_1 \si \gem$ and $H_2 \si P_6$
\item \label{bthm:classification2-atoms:unbdd2:P12P2} $H_1 \si P_1+\nobreak 2P_2$ or~$P_6$ and $H_2 \si \overline{P_1+2P_2}$ or~$\overline{P_2+P_3}$, or
\item \label{bthm:classification2-atoms:unbdd2:2P2}$H_1 \si 2P_2$ and $H_2 \si \overline{P_2+P_4},\; \overline{3P_2}$ or~$\overline{P_5}$.
\end{enumerate}
\end{enumerate}
\end{theorem}
\begin{proof}
We start by considering the bounded cases.
Theorem~\ref{bthm:classification2-atoms}.\ref{bthm:classification2-atoms:easy-bdd} follows immediately from Theorem~\ref{thm:classification2}.\ref{thm:classification2:known-bdd}.
Theorem~\ref{bthm:classification2-atoms}.\ref{bthm:classification2-atoms:bdd2:P6} follows from the fact that $(P_6,\overline{2P_2})$-free atoms have bounded clique-width~\cite{GHP18}.
Theorem~\ref{bthm:classification2-atoms}.\ref{bthm:classification2-atoms:bdd2:2P2} follows from the fact that $(2P_2,\overline{P_2+P_3})$-free atoms have bounded clique-width (Theorem~\ref{thm:triplet}).
Next, we consider the unbounded cases.
Theorem~\ref{bthm:classification2-atoms}.\ref{bthm:classification2-atoms:unbdd:not-in-S} and Theorem~\ref{bthm:classification2-atoms}.\ref{bthm:classification2-atoms:unbdd:not-in-co-S} follow from Lemma~\ref{lem:need-S-co-S}.
Theorem~\ref{bthm:classification2-atoms}.\ref{bthm:classification2-atoms:unbdd:K_13} follows from Lemma~\ref{lem:C4-K13-K4-diamond}.
Theorem~\ref{bthm:classification2-atoms}.\ref{bthm:classification2-atoms:unbdd:2P_1+P_2} follows from Lemmas~\ref{lem:C4-K13-K4-diamond}, \ref{lem:diamond-5P1} and~\ref{lem:diamond-P2+P4}.
Theorem~\ref{bthm:classification2-atoms}.\ref{bthm:classification2-atoms:unbdd:3P_1} follows from Lemmas~\ref{lem:bip-2P12P2-2P1P4-4P1P2-3P2} and~\ref{lem:bip-2P3}.
Theorem~\ref{bthm:classification2-atoms}.\ref{bthm:classification2-atoms:unbdd:4P_1} follows from Lemma~\ref{lem:4P_1-gem}, \ref{lem:K4-3P1P2} and~\ref{lem:K4-2P2}.
Theorem~\ref{bthm:classification2-atoms}.\ref{bthm:classification2-atoms:unbdd:gem} follows from Lemma~\ref{lem:gem-P12P2}.
Theorem~\ref{bthm:classification2-atoms}.\ref{bthm:classification2-atoms:unbdd2:diamond} follows from Lemma~\ref{lem:diamond-P2+P4}.
Theorem~\ref{bthm:classification2-atoms}.\ref{bthm:classification2-atoms:unbdd2:co-diamond} follows from Lemma~\ref{lem:codiamond-coP6}.
Theorem~\ref{bthm:classification2-atoms}.\ref{bthm:classification2-atoms:unbdd2:gem} follows from Lemma~\ref{lem:gem-P6}.
Theorem~\ref{bthm:classification2-atoms}.\ref{bthm:classification2-atoms:unbdd2:P12P2} follows from Lemma~\ref{lem:coP2P3-coP12P2-P12P2-P6}.
Theorem~\ref{bthm:classification2-atoms}.\ref{bthm:classification2-atoms:unbdd2:2P2} follows from Lemmas~\ref{lem:2P2-coP2P4} and~\ref{lem:2P_2-coP5-co3P2}.
\end{proof}

In the open problem below, the cases marked with a~$^*$ are those for which even the boundedness of clique-width of the whole class of $(H_1,H_2)$-free graphs is unknown (see also Open Problem~\ref{oprob:twographs} in Section~\ref{s-soa}).

\begin{oproblem}\label{o-atoms}
Does the class of $(H_1,H_2)$-free atoms have bounded clique-width if
\begin{enumerate}[*(i)]
\renewcommand{\theenumi}{(\roman{enumi})}
\renewcommand{\labelenumi}{(\roman{enumi})}
\item \label{openprob:diamond}$H_1 = \diamondgraph$ and $H_2=P_6$
\item \label{openprob:C_4} $H_1=C_4$ and $H_2 \in \{P_1+\nobreak 2P_2, P_2+\nobreak P_4, 3P_2\}$
\item \label{openprob:oP_12P_2} $H_1=\overline{P_1+2P_2}$ and $H_2 \in \{2P_2, P_2+\nobreak P_3, P_5\}$ 
\item \label{openprob:oP_2P_3} $H_1=\overline{P_2+P_3}$ and $H_2 \in \{P_2+\nobreak P_3, P_5\}$        
\renewcommand{\labelenumi}{*(\roman{enumi})}
\item \label{openprob:K_3} $H_1=K_3$ and $H_2 \in \{P_1+\nobreak S_{1,1,3},\allowbreak S_{1,2,3}\}$
\item $H_1=3P_1$ and $H_2 = \overline{P_1+\nobreak S_{1,1,3}}$
\item $H_1= \diamondgraph$ and $H_2 \in \{P_1+\nobreak P_2+\nobreak P_3,\allowbreak P_1+\nobreak P_5\}$ 
\item $H_1=2P_1+\nobreak P_2$ and $H_2 \in \{\overline{P_1+\nobreak P_2+\nobreak P_3},\allowbreak \overline{P_1+\nobreak P_5}\}$ 
\item $H_1=\gem$ and $H_2=P_2+\nobreak P_3$, or
\item \label{openprob:P_1P_4}$H_1=P_1+\nobreak P_4$ and $H_2=\overline{P_2+\nobreak P_3}$.
\end{enumerate}
\end{oproblem}

Olariu~\cite{Olariu88} proved that every connected $\overline{P_1+P_3}$-free graph is either $K_3$-free or complete multi-partite.
Since complete multi-partite graphs and their complements have bounded clique-width, when looking at boundedness of clique-width of a hereditary class, forbidding~$\overline{P_1+P_3}$ as an induced subgraph is equivalent to forbidding~$K_3$ and forbidding~$P_1+\nobreak P_3$ is equivalent to forbidding~$3P_1$.
Thus, when studying boundedness of clique-width we may assume that we never explicitly forbid $\overline{P_1+P_3}$ or~$P_1+\nobreak P_3$.
Furthermore, by Lemma~\ref{lem:K3S123-coK3coS123-equiv}, the class of $(K_3,S_{1,2,3})$-free atoms has bounded clique-width if and only if the class of $(3P_1,\overline{S_{1,2,3}})$-free atoms has bounded clique-width, so we may assume $\{H_1,H_2\} \neq \{3P_1,\overline{S_{1,2,3}}\}$.
We now state the following theorem.

\begin{theorem}\label{thm:all-open-cases-listed}
Let~$H_1$ and~$H_2$ be graphs (which are not isomorphic to~$\overline{P_1+P_3}$ or~$P_1+\nobreak P_3$) with $\{H_1,H_2\} \neq \{3P_1,\overline{S_{1,2,3}}\}$  and let~${\cal G}$ be the class of $(H_1,H_2)$-free graphs.
Then (un)boundedness of clique-width for atoms in~${\cal G}$ does not follow from Theorem~\ref{bthm:classification2-atoms} if and only if this class is listed in Open Problem~\ref{o-atoms}.
\end{theorem}

\begin{proof}
First, note that Theorem~\ref{bthm:classification2-atoms} does not specify the (un)boundedness of clique-width for atoms in any of the classes listed in Open Problem~\ref{o-atoms}.

Consider the classes listed in Open Problem~\ref{oprob:twographs}.
For all bigenic classes~${\cal G}$ for which the (un)boundedness of clique-width of general graphs is not listed in Theorem~\ref{thm:classification2}, an equivalent class is listed in Open Problem~\ref{oprob:twographs} (see~\cite{DJP19} and~\cite{DP16}).
Since the results in Theorem~\ref{bthm:classification2-atoms}.\ref{bthm:classification2-atoms:hard-bdd} do not solve these cases when restricted to atoms, these classes (and their complements, apart from the $H_1=3P_1$, $H_2=\overline{S_{1,2,3}}$) appear in Open Problem~\ref{o-atoms}.\ref{openprob:K_3}-\ref{openprob:P_1P_4}.
The only other classes we need to consider are those for which Theorem~\ref{thm:classification2}.\ref{thm:classification2:known-unbdd} states that the class~${\cal G}$ has unbounded clique-width, but the class of atoms in~${\cal G}$ might not have unbounded clique-width.

There are two classes listed in Theorem~\ref{thm:classification2}.\ref{thm:classification2:known-unbdd} that are not listed in Theorem~\ref{bthm:classification2-atoms}.\ref{bthm:classification2-atoms:unchanged-unbdd}, namely the class of $(\overline{2P_2},2P_2)$-free graphs and the class of $(\overline{2P_1+P_2},P_6)$-free graphs.
The class of $(\overline{2P_2},2P_2)$-free graphs is only equivalent to itself.
The class of $(\overline{2P_1+P_2},P_6)$-free graphs equivalent to only one other class, namely the class of $(2P_1+\nobreak P_2,\overline{P_6})$-free graphs.
However, the class of $(2P_1+\nobreak P_2,\overline{P_6})$-free atoms has unbounded clique-width by Theorem~\ref{bthm:classification2-atoms}.\ref{bthm:classification2-atoms:unbdd2:co-diamond}.
We therefore only need to consider the class of $(\overline{2P_2},2P_2)$-free graphs and the class of $(2P_1+\nobreak P_2,\overline{P_6})$-free, together with any bigenic classes~${\cal G'}$ that are extensions of these classes such that Theorem~\ref{bthm:classification2-atoms} does not specify that the atoms of~${\cal G'}$ have unbounded clique-width.

We start by considering extensions of the classes of $(\overline{2P_1+P_2},P_6)$-free graphs.
Consider graphs $H_1$, $H_2$ with $\overline{2P_1+\nobreak P_2} \ssi H_1$ and $P_6 \ssi H_2$ such that the class of $(H_1,H_2)$-free atoms has bounded clique-width, but Theorem~\ref{bthm:classification2-atoms} does not state that $(H_1,H_2)$-free atoms have unbounded clique-width.
By Theorem~\ref{bthm:classification2-atoms}.\ref{bthm:classification2-atoms:unbdd:not-in-co-S}, it follows that $\overline{H_1} \in {\cal S}$.
By Theorem~\ref{bthm:classification2-atoms}.\ref{bthm:classification2-atoms:unbdd:K_13}, it follows that~$\overline{H_1}$ is $K_{1,3}$-free, so it is a linear forest.
By Theorem~\ref{bthm:classification2-atoms}.\ref{bthm:classification2-atoms:unbdd:4P_1}, it follows that~$\overline{H_1}$ is $4P_1$-free.
By Theorem~\ref{bthm:classification2-atoms}.\ref{bthm:classification2-atoms:unbdd2:gem}, $\overline{H_1}$ must be $(P_1+\nobreak P_4)$-free.
By Theorem~\ref{bthm:classification2-atoms}.\ref{bthm:classification2-atoms:unbdd2:P12P2}, $\overline{H_1}$ must be $(P_1+\nobreak 2P_2,P_2+\nobreak P_3)$-free.
The $1$-vertex extensions of~$2P_1+\nobreak P_2$ in~${\cal S}$ are $3P_1+\nobreak P_2$, $P_1+\nobreak 2P_2, 2P_1 +\nobreak P_3, P_1+\nobreak P_4, P_2+\nobreak P_3$ and~$S_{1,1,2}$, none of which are $(K_{1,3},4P_1,P_1+\nobreak P_4,P_1+\nobreak 2P_2,P_2+\nobreak P_3)$-free.
We conclude that $\overline{H_1}=2P_1+\nobreak P_2$.
By Theorem~\ref{bthm:classification2-atoms}.\ref{bthm:classification2-atoms:unbdd:not-in-S}, it follows that $H_2 \in {\cal S}$.
By Theorem~\ref{bthm:classification2-atoms}.\ref{bthm:classification2-atoms:unbdd:2P_1+P_2}, it follows that~$H_2$ is $(K_{1,3},P_2+\nobreak P_4)$-free.
By Theorem~\ref{bthm:classification2-atoms}.\ref{bthm:classification2-atoms:unbdd2:diamond}, it follows that~$H_2$ is $(P_1+\nobreak P_6)$-free.
The $1$-vertex extensions of~$P_6$ that are in~${\cal S}$ are $P_1+\nobreak P_6$, $P_7$, $S_{1,1,4}$ and~$S_{1,2,3}$, none of which are $(K_{1,3},P_2+\nobreak P_4, P_1+\nobreak P_6)$-free.
We conclude that $H_2=P_6$.
Therefore, we do not need to consider any extensions of $(\overline{2P_1+\nobreak P_2},P_6)$-free graphs, apart from the class of $(\overline{2P_1+\nobreak P_2},P_6)$-free graphs itself, and this is listed in Open Problem~\ref{o-atoms}.\ref{openprob:diamond}.

Now consider graphs~$H_1$, $H_2$ with $2P_2 \ssi H_1, \overline{H_2}$ such that the class of $(H_1,H_2)$-free atoms has bounded clique-width, but Theorem~\ref{bthm:classification2-atoms} does not state that $(H_1,H_2)$-free atoms have unbounded clique-width.
By Theorem~\ref{bthm:classification2-atoms}.\ref{bthm:classification2-atoms:unbdd:not-in-S} and Theorem~\ref{bthm:classification2-atoms:unbdd:not-in-co-S}, respectively, $H_1$ and~$\overline{H_2}$ must both be in~${\cal S}$.
By Theorem~\ref{bthm:classification2-atoms}.\ref{bthm:classification2-atoms:unbdd:K_13}, it follows that~$H_1$ and~$\overline{H_2}$ are $K_{1,3}$-free, so they are both linear forests.
By Theorem~\ref{bthm:classification2-atoms}.\ref{bthm:classification2-atoms:unbdd:4P_1}, $H_1$ and~$\overline{H_2}$ are $4P_1$-free, and because they are bipartite, this means they each contain at most six vertices.
Since~$H_1$ and~$\overline{H_2}$ are linear forests on at most six vertices containing an induced~$2P_2$, it follows that $H_1,\overline{H_2} \in \{2P_2, P_1+\nobreak 2P_2, P_2+\nobreak P_3, P_5, 2P_1+\nobreak 2P_2, 3P_2, P_1+\nobreak P_2 + \nobreak P_3, P_2+\nobreak P_4, P_1+\nobreak P_5, 2P_3, P_6\}$.
Since~$H_1$ and~$\overline{H_2}$ are $4P_1$-free, it follows that $H_1,\overline{H_2} \in \{2P_2, P_1+\nobreak 2P_2, P_2+\nobreak P_3, P_5, 3P_2,\allowbreak P_2+\nobreak P_4, P_6\}$.
By Theorem~\ref{bthm:classification2-atoms}.\ref{bthm:classification2-atoms:unbdd2:2P2}, $\overline{H_2}$ is $(P_2+\nobreak P_4,3P_2,P_5)$-free, and so $\overline{H_2} \in \{2P_2,P_1+\nobreak 2P_2,P_2+\nobreak P_3\}$.
Now if $\overline{H_2}=2P_2$, then by Theorem~\ref{bthm:classification2-atoms}.\ref{bthm:classification2-atoms:bdd2:P6}, we may assume that~$H_1$ is not an induced subgraph of~$P_6$, so $H_1 \in \{P_1+\nobreak 2P_2, 3P_2, P_2+\nobreak P_4\}$ and these cases are listed in Open Problem~\ref{o-atoms}.\ref{openprob:C_4}.
Otherwise, $\overline{H_2} \in \{P_1+\nobreak 2P_2,P_2+\nobreak P_3\}$.
In this case by Theorem~\ref{bthm:classification2-atoms}.\ref{bthm:classification2-atoms:unbdd2:P12P2} $H_1$ is $(P_1+\nobreak 2P_2,P_6)$-free, so $H_1 \in \{2P_2, P_2+\nobreak P_3, P_5\}$.
If $\overline{H_2}=P_1+\nobreak 2P_2$ then $H_1 \in \{2P_2, P_2+\nobreak P_3, P_5\}$ and these cases are listed in Open Problem~\ref{o-atoms}.\ref{openprob:oP_12P_2}.
If $\overline{H_2}=P_2+\nobreak P_3$ then by Theorem~\ref{bthm:classification2-atoms}.\ref{bthm:classification2-atoms:bdd2:2P2}, $H_1$ is not an induced subgraph of~$2P_2$, so $H_1 \in \{P_2+\nobreak P_3, P_5\}$ and these cases are listed in Open Problem~\ref{o-atoms}.\ref{openprob:oP_2P_3}.
\end{proof}

\section{Conclusions}\label{s-con}
Motivated by algorithmic applications, we determined a new class of $(H_1,H_2)$-free graphs of unbounded clique-width whose atoms have {\it bounded} clique-width, namely when $(H_1,H_2)= (2P_2,\overline{P_2+P_3})$
(in fact, our proof for $(2P_2,\overline{P_2+P_3})$-free atoms also works for {\it linear} clique-width).
We also identified a number of classes of $(H_1,H_2)$-free graphs of unbounded clique-width whose atoms still have {\it unbounded} clique-width.
In particular, our results show that boundedness of clique-width of $(H_1,H_2)$-free atoms does not necessarily imply boundedness of clique-width of $(\overline{H_1},\overline{H_2})$-free atoms.
For example, $(C_4,P_5)$-free atoms have bounded clique-width~\cite{GHP18}, but we proved that $(\overline{C_4},\overline{P_5})$-free atoms have unbounded clique-width (Lemma~\ref{lem:2P_2-coP5-co3P2}).

We also presented a summary theorem (Theorem~\ref{thm:classification2-atoms}), from which we deduced a list of {\bf 18} remaining cases of pairs $(H_1,H_2)$ for which we do not know whether the clique-width of $(H_1,H_2)$-free atoms is bounded; see also Open Problem~\ref{o-atoms} and Theorem~\ref{thm:all-open-cases-listed}.
In particular, we ask whether boundedness of clique-width of $(2P_2,\overline{P_2+P_3})$-free atoms can be extended to $(P_5,\overline{P_2+P_3})$-free atoms.
Is boundedness of clique-width the underlying reason why {\sc Colouring} is polynomial-time solvable on $(P_5,\overline{P_2+P_3})$-free graphs~\cite{ML17}?
Brandst\"adt and Ho\`ang~\cite{BH07} showed that $(P_5,\overline{P_2+P_3})$-free atoms with no dominating vertices and no vertex pairs~$\{x,y\}$ with $N(x)\subseteq N(y)$ are either isomorphic to some specific graph~$G^*$ or all their induced~$C_5$s are dominating.
Recently, Huang and Karthick~\cite{HK20} proved a more refined decomposition.
However, it is not clear how to use these results to prove boundedness of clique-width of $(P_5,\overline{P_2+P_3})$-free atoms, and additional insights are needed.

\bibliographystyle{plainurl}
\bibliography{mybib-no-url}

% \appendix

% \section{The Missing Proof of Lemma~\ref{l-tripletwithc5}}\label{a-first}

\end{document}